\documentclass[acmsmall, screen]{acmart}
\usepackage[utf8]{inputenc}

\usepackage{thmtools}
\usepackage{thm-restate}

\usepackage{mathpartir}
\usepackage{setspace}
\usepackage[normalem]{ulem}
\usepackage{relsize}
\usepackage{pdfsync}
\usepackage{listings}
\usepackage{wrapfig}
\usepackage{acronym}
\lstdefinestyle{wrapfig}{
  xleftmargin=0pt,
  belowcaptionskip=0pt,
  aboveskip=0pt,
  belowskip=0pt
}

\makeatletter
\lst@AddToHook{OnEmptyLine}{\vspace{\dimexpr-\baselineskip+\medskipamount}}
\makeatother

\acmConference[PLDI'24]{}{}{}
\acmYear{2024}
\copyrightyear{2018}
\acmISBN{}
\acmDOI{}
\setcopyright{none}

\author{Joseph W. Cutler}
\email{jwc@seas.upenn.edu}
\orcid{0000-0001-9399-9308} 
\affiliation{%
  \institution{University of Pennsylvania}
  \city{Philadelphia}
  \state{Pennsylvania}
  \country{USA}
}
\author{Christopher Watson}
\email{ccwatson@seas.upenn.edu}
\orcid{0000-0003-3716-516X} 
\affiliation{%
  \institution{University of Pennsylvania}
  \city{Philadelphia}
  \state{Pennsylvania}
  \country{USA}
}

\author{Emeka Nkurumeh}
\email{enkurume@caltech.edu}
\orcid{0009-0003-4124-8598} 
\affiliation{%
  \institution{California Institute of Technology}
  \city{Pasadena}
  \state{California}
  \country{USA}
}

\author{Phillip Hilliard}
\email{pdh@seas.upenn.edu}
\orcid{} 
\affiliation{%
  \institution{University of Pennsylvania}
  \city{Philadelphia}
  \state{Pennsylvania}
  \country{USA}
}

\author{Harrison Goldstein}
\email{hgo@seas.upenn.edu}
\orcid{0000-0001-9631-1169} 
\affiliation{%
  \institution{University of Pennsylvania}
  \city{Philadelphia}
  \state{Pennsylvania}
  \country{USA}
}

\author{Caleb Stanford}
\email{cdstanford@ucdavis.edu}
\orcid{0000-0002-8428-7736} 
\affiliation{%
  \institution{University of California, Davis}
  \city{Davis}
  \state{California}
  \country{USA}
}

\author{Benjamin C. Pierce}
\email{bcpierce@cis.upenn.edu}
\orcid{0000-0001-7839-1636} 
\affiliation{%
  \institution{University of Pennsylvania}
  \city{Philadelphia}
  \state{Pennsylvania}
  \country{USA}
}

\renewcommand{\shortauthors}{Cutler et al.}

\AtBeginDocument{%
  \providecommand\BibTeX{{%
    \normalfont B\kern-0.5em{\scshape i\kern-0.25em b}\kern-0.8em\TeX}}} 

\newcommand{\core}{$\lambda^{\text{ST}}$}
\newcommand{\lang}{delta}
\newcommand{\ruleName}[1]{\textsc{#1}}

\newcommand{\SIMPLEHEADER}[1]{\emph{#1}.}

\newif\ifdraft\draftfalse
\newif\ifaftersubmission\aftersubmissionfalse
\newif\ifextended\extendedfalse

\definecolor{dkred}{rgb}{0.7,0,0}
\definecolor{dkpurple}{HTML}{4e02eb}
\definecolor{dkgreen}{HTML}{006329}
\definecolor{dkblue}{HTML}{2d5491}
\definecolor{dkorange}{HTML}{825a23}
\definecolor{ltgreen}{HTML}{3a9e67}
\definecolor{teal}{HTML}{007982}
\definecolor{fuchsia}{HTML}{8C368C}

\newcommand{\comm}[3]{\ifdraft\textcolor{#1}{[#2: #3]}\fi}

\newcommand{\bcp}[1]{\comm{dkpurple}{BCP}{#1}}

\newcommand{\jwc}[1]{\comm{dkgreen}{JWC}{#1}}

\newcommand{\cutifneeded}[1]{#1}

\newcommand*{\INFERRULE}[3][]{%
  \inferrule*[right=#1,
              rightskip=2em,
              rightstyle=\sc]%
  {#2}{#3}
}
\newcommand{\ENDOFLINE}{\vspace*{1.7ex} \\}
\newcommand{\BEFORECAPTIONSPACE}{\vspace*{-2.5ex}}

\definecolor{tmcolor}{RGB}{50, 48, 200}
\definecolor{pfxcolor}{RGB}{219, 48, 122}
\definecolor{batcolor}{RGB}{219, 48, 122}
\definecolor{}{RGB}{219, 48, 122}

\newcommand{\prefixOf}[1]{\left(#1\right)^\circ}

\newcommand{\prefixConcatRel}[3]{#1 \cdot #2 \sim #3}
\newcommand{\prefixConcat}[2]{#1 \cdot #2}

\newcommand{\semicctx}[2]{#1 \,; #2}
\newcommand{\commactx}[2]{#1 \,, #2}
\newcommand{\semicctxcompact}[2]{#1 \mathord{;} #2}

\newcommand{\prefixHasType}[2]{#1 \, :\, \texttt{prefix}\left(#2\right)}
\newcommand{\envHasType}[2]{#1 \, :\, \texttt{env}\left(#2\right)}

\newcommand{\deriv}[2]{\delta_{#1}\left(#2\right)}
\newcommand{\derivrel}[3]{\delta_{#1}\left(#2\right) \sim #3}

\newcommand{\flatten}[1]{\langle #1 \rangle}
\newcommand{\histValToPrefix}[2]{\texttt{toPrefix}_{#2}\left(#1\right)}

\newcommand{\nullable}[1]{#1\;\texttt{null}}
\newcommand{\inert}[1]{#1\;\texttt{inert}}
\newcommand{\I}{\texttt{I}}
\newcommand{\J}{\texttt{J}}

\newcommand{\tck}[6]{#1 \mid #2 \vdash_{#3} #4 : #5 \, @\, #6}

\newcommand{\norec}{\texttt{NoRec}}

\newcommand{\sinkTm}[1]{\texttt{sink}_{#1}}

\newcommand{\eventfor}[2]{#1 \, : \, \texttt{event}\left(#2\right)}
\newcommand{\eventsfor}[2]{#1 \, : \, \texttt{events}\left(#2\right)}

\newcommand{\evSize}[1]{\textit{size}(#1)}
\newcommand{\evSizeBound}[1]{\textit{evSizeBound}(#1)}

\newcommand{\donebatch}[2]{#1 \; \texttt{done}\; #2}
\newcommand{\EToP}[3]{#1 \hookrightarrow^{#2} #3}
\newcommand{\ESToP}[3]{#1 \hookrightarrow^{#2} #3}

\newcommand{\PToES}[3]{#1 \dagger^{#2} #3}

\newcommand{\subty}[2]{#1 \, <: \, #2}

\newcommand{\fixsubst}[2]{#1\left[#2/\texttt{rec}\right]}

\newcommand{\XS}{\mathit{xs}}

\newcommand{\lctck}[3]{#1 \vdash #2 : #3}
\newcommand{\histsubsttck}[3]{#1 : #2 \to #3}

\newcommand{\prefixstep}[5]{#1 \Rightarrow #2 \downarrow^{#3} #4 \Rightarrow #5}
\newcommand{\argsstep}[6]{#1 \Rightarrow #2 @ #3 \downarrow^{#4} #5 \Rightarrow #6}

\newcommand{\onet}{1}
\newcommand{\epst}{\varepsilon}

\newcommand{\intt}{\texttt{Int}}

\newcommand{\inl}[1]{\texttt{inl}(#1)} 
\newcommand{\inr}[1]{\texttt{inr}(#1)} 

\newcommand{\emp}[1]{\texttt{emp}_{#1}}
\newcommand{\isEmpty}[1]{#1\; \texttt{empty}}
\newcommand{\isMaximal}[1]{#1\; \texttt{maximal}}

\newcommand{\emptyOn}[2]{#1 \, \texttt{emptyOn} \, #2}
\newcommand{\maximalOn}[2]{#1\, \texttt{maximalOn} \, #2}

\newcommand{\agree}[4]{\texttt{agree}\left(#1,#2,#3,#4\right)}

\newcommand{\onepA}{\texttt{oneEmp}}
\newcommand{\onepB}{\texttt{oneFull}}
\newcommand{\epsp}{\texttt{epsEmp}}

\newcommand{\stpEmp}{\texttt{starEmp}}
\newcommand{\stpDone}{\texttt{starDone}}
\newcommand{\stpA}[1]{\texttt{starFirst}(#1)}
\newcommand{\stpB}[2]{\texttt{starRest}(#1 , #2)}

\newcommand{\parp}[2]{\texttt{parPair}(#1,#2)}
\newcommand{\sumpA}[1]{\texttt{sumInl}(#1)}
\newcommand{\sumpB}[1]{\texttt{sumInr}(#1)}
\newcommand{\sumpEmp}{\texttt{sumEmp}} 

\newcommand{\catpA}[1]{\texttt{catFst}(#1)}
\newcommand{\catpB}[2]{\texttt{catBoth}(#1 , #2)}

\newcommand{\oneev}{\texttt{oneev}}
\newcommand{\parevA}[1]{\texttt{parevA}\left(#1\right)}
\newcommand{\parevB}[1]{\texttt{parevB}\left(#1\right)}

\newcommand{\parevAmap}[1]{\widehat{\texttt{parevA}}\left(#1\right)}
\newcommand{\parevBmap}[1]{\widehat{\texttt{parevB}}\left(#1\right)}

\newcommand{\catevA}[1]{\texttt{catevA}\left(#1\right)}
\newcommand{\catevAmap}[1]{\widehat{\texttt{catevA}}\left(#1\right)}
\newcommand{\catPunc}{\cdot_\texttt{punc}}

\newcommand{\sumPuncA}{+_\texttt{puncA}}
\newcommand{\sumPuncB}{+_\texttt{puncB}}

\newcommand{\cutTm}[3]{\texttt{let}\, #1=#2\,\texttt{in}\,#3}
\newcommand{\nilTm}{\texttt{nil}}
\newcommand{\consTm}[2]{#1 \, \texttt{::}\, #2}

\newcommand{\starcaseTm}[7]{\texttt{case}_{#1}\left(#2;#3,#4,#5.#6.#7\right)}
\newcommand{\waitTm}[4]{\texttt{wait}_{#1,#2}(#3)\left(#4\right)}
\newcommand{\waitTmNOBUFFER}[3]{\texttt{wait}_{#1}(#2)\left(#3\right)}

\newcommand{\sumInlTm}[1]{\texttt{inl}\left(#1\right)}
\newcommand{\sumInrTm}[1]{\texttt{inr}\left(#1\right)}
\newcommand{\sumcaseTm}[7]{\texttt{case}_{#1}\left(#2;#3,#4.#5,#6.#7\right)}

\newcommand{\letcatTm}[5]{\texttt{let}_{#1}\,\left(#2;#3\right)\,=\,#4\,\texttt{in}\,#5 }
\newcommand{\letparTm}[4]{\texttt{let}\,\left(#1,#2\right)\,=\,#3\,\texttt{in}\,#4 }

\newcommand{\catpairTm}[2]{\left(#1 ; #2 \right)}
\newcommand{\parpairTm}[2]{\left(#1 , #2 \right)}

\newcommand{\epsTm}{\texttt{sink}}
\newcommand{\oneTm}{\texttt{()}}

\newcommand{\fixTm}[2]{\texttt{fix}\left(#1\right).\left(#2\right)}
\newcommand{\fixTmHistargs}[3]{\texttt{fix}\left\{#1\right\}\left(#2\right).\left(#3\right)}
\newcommand{\recTm}[2]{\texttt{rec}\left\{#1\right\} \left(#2\right)}

\newcommand{\histPgmTm}[2]{\langle #1 : #2\rangle}

\newcommand{\catpairArgsTmA}[2]{\left(#1 ; #2 \right)}
\newcommand{\catpairArgsTmB}[1]{\left(\cdot ; #1 \right)}
\newcommand{\parpairArgsTm}[2]{\left(#1 , #2 \right)}

\newcommand{\cons}[2]{#1 \texttt{::} #2}

\mathchardef\mhyphen="2D

\newcommand{\superstar}{\ensuremath{^\star}}

\newcommand{\dom}[1]{\text{Dom}\left(#1\right)}

\usepackage[framemethod=default]{mdframed}
\usepackage{changepage}

\newcommand{\jtheorem}[2]{

  \noindent \textbf{#1}

  \noindent #2 \qed

}

\newmdenv[
  usetwoside=false,
  topline=false,
  bottomline=false,
  rightline=false,
  leftmargin=0.2in,
  linewidth=0.75pt,
  skipabove=\topsep,
  skipbelow=\topsep,
  nobreak=false
]{leftrule}

\newcommand{\allrule}[1]{
  \vspace{\topsep}

  \noindent\hspace{10pt}\fbox{#1}

  \vspace{\topsep}
}

\newcommand{\jgivengoal}[2]{
  \vspace{0.5em}

  {
    \setlength{\parskip}{0em}
    \noindent $\blacktriangleright$ \textbf{Given:}
    {
      \begin{adjustwidth}{0.1in}{0in}
        #1
      \end{adjustwidth}
    }

    \noindent $\blacktriangleright$ \textbf{Goal:}

    \begin{adjustwidth}{0.1in}{0in}
      \allrule{#2}
    \end{adjustwidth}

  }
  \noindent \ignorespaces
}

\newcommand{\jgivengoalTwo}[3]{

  {
    \setlength{\parskip}{0em}
    \noindent $\blacktriangleright$ \textbf{Given:}
    {
      \begin{adjustwidth}{0.1in}{0in}
        #1
      \end{adjustwidth}
    }

    \noindent $\blacktriangleright$ \textbf{Goal A:}

    \begin{adjustwidth}{0.1in}{0in}
      \allrule{#2}
    \end{adjustwidth}

    \noindent $\blacktriangleright$ \textbf{Goal B:}

    \begin{adjustwidth}{0.1in}{0in}
      \allrule{#3}
    \end{adjustwidth}

  }
  \noindent \ignorespaces
}

\newcommand{\jgivengoalThree}[4]{

  {
    \setlength{\parskip}{0em}
    \noindent $\blacktriangleright$ \textbf{Given:}
    {
      \begin{adjustwidth}{0.1in}{0in}
        #1
      \end{adjustwidth}
    }

    \noindent $\blacktriangleright$ \textbf{Goal A:}

    \begin{adjustwidth}{0.1in}{0in}
      \allrule{#2}
    \end{adjustwidth}

    \noindent $\blacktriangleright$ \textbf{Goal B:}

    \begin{adjustwidth}{0.1in}{0in}
      \allrule{#3}
    \end{adjustwidth}
    
    \noindent $\blacktriangleright$ \textbf{Goal C:}

    \begin{adjustwidth}{0.1in}{0in}
      \allrule{#4}
    \end{adjustwidth}

  }
  \noindent \ignorespaces
}

\newcommand{\jcase}[3]{

  \noindent $\blacktriangleright$ \textbf{Case #1:} \textit{#2.}

  {
    \begin{leftrule}
      \vspace{0.35em}
      #3
    \end{leftrule}
  }
  
  \noindent \ignorespaces
}

\newcommand{\caseText}[1]
  {\noindent #1}

\newcommand{\caseFact}[1]
  {\noindent \hspace{10pt}(#1)\hspace{5pt}}

\newcommand{\caseFactPl}[1]
  {\caseFact{#1}}

\usepackage{ stmaryrd }

\title{Stream Types}

\begin{abstract}
We propose a rich foundational theory of typed data streams and stream
transformers, motivated by two high-level goals: (1) The type of
a stream should be able to express complex \emph{sequential patterns}
of events over time. And (2) it should describe the internal \emph{parallel
  structure} of the stream to support deterministic stream processing
on parallel and distributed systems.
To these ends, we introduce \emph{stream types}, with
operators capturing sequential composition, parallel composition, and
iteration, plus a core calculus \core{}
of \emph{transformers} over typed streams
which naturally supports
a number of common streaming idioms, including punctuation, windowing, and
parallel partitioning, as first-class constructions.
\core{} exploits a Curry-Howard-like correspondence with
an ordered variant of the logic of Bunched Implication to
program with streams
compositionally and uses Brzozowski-style derivatives to enable an incremental,
prefix-based operational semantics.
To illustrate the programming style supported by the rich types of \core{}, we
present a number of examples written in \lang{}, a prototype high-level language design based on \core{}.
\end{abstract}

\keywords{Type Systems, Stream Processing, Ordered Logic, Bunched Implication}

\begin{document}
\maketitle

\section{Introduction}\label{sec:intro}
What is the type of a stream?
A straightforward answer, dating back to
the early days of functional programming~\cite{burge1975stream}, is that a stream is an unbounded
sequence of items of a single fixed type, produced by one
part of a system (or the external world) and consumed by another.
This simple perspective has been immensely successful:
the current programming models exposed by the most popular
distributed stream processing eDSLs
(e.g., Flink~\cite{Flink,Flink2015}, Beam~\cite{Beam}, Storm~\cite{Storm}, and Heron~\cite{Heron})
typically offer just one type, \texttt{Stream t}.

This homogeneous treatment of streams leaves something to
be desired.
For one thing, streaming data sometimes arrives at a processing node from
multiple sources {\em in parallel}.
Using arrival times to impose an ``incidental'' order on such parallel data
can make it difficult to ensure that processing is
deterministic, because downstream results may then depend on
factors like network latency~\cite{mamouras2019data, schneider2013safe}.
Another issue with the homogeneous stream abstraction is that
temporal patterns
 like \emph{bracketedness} (every ``begin'' event has a following
 ``end'') or the fact that with exactly $k$ events are expected to
 arrive on a stream are invisible
in its type. Programmers get no help from the type system to
ensure such properties when producing a stream, nor can they rely on
them when consuming a
stream.

Our principal contribution is a novel logical foundation for typed stream processing
that can precisely describe streams with both
complex sequential patterns and parallel structure.
On this foundation,
we build a calculus called \core{} that
is (a) expressive and type-safe for streams with such complex
temporal patterns and (b) deterministic, even when inputs can arrive
from multiple sources in parallel. We also present \lang{}, an experimental language design based on \core{}.
(A full-blown distributed implementation of \lang{} is left for future
work.)

Programs in \core{} are intuitively batch processors
that operate over entire streams at once. But,
since streams are in general unbounded,
stream transformers can't actually wait for ``the entire input stream'' to
arrive before producing any output.
The operational semantics of the \core{} calculus is therefore designed
to be {\em
  incremental}, producing partial outputs from partial inputs on the
fly.

%
A \core{} program is interpreted as a function mapping
any prefix of its input(s) to
a prefix of its output plus a ``resultant'' term to
transform the rest of the
inputs to the rest of the output.


Our stream types include two kinds of products, one representing
a pair of streams in temporal sequence, the other a pair
of streams in parallel.
This structure is inspired both by Concurrent Kleene
Algebras~\cite{hoare2009concurrent, kappePomset}, which syntactically describe  partially
ordered series/parallel data, and by
work by \citet{synch-schemas} and
\citet{mamouras2019data}, where streams are modeled as partially ordered sets.
We discover a suitable proof theory for this two-product formalism in
a variant of O'Hearn and Pym's Logic of Bunched
Implications (BI) \cite{o1999logic}.
BI is well known as a foundation for separation logic
\cite{reynolds2002separation}, where
its ``separating
conjunction'' allows for local reasoning about
separate regions of
the heap in imperative programs.
In \core{}, we replace
spatial separation with \emph{temporal separation}:
one product describes pairs of streams separated
sequentially in time; the other describes
pairs of temporally independent streams whose elements may arrive in
interleaved fashion.

Concretely, our contributions are:
\begin{enumerate}
  \item We propose \emph{stream types}, a static discipline for
  distributed stream processing
  that generalizes the traditional homogeneous view of streams to a richer
  nested-parallel-and-sequential structure, and define \core{}, a calculus of stream processing transformers
  inspired by a Curry-Howard-like correspondence with an ordered
  variant of BI. \cutifneeded{Terms in \core{}
  are high level programs in a functional style that conceptually
  transform whole streams at once.}

  \item
  We equip \core{} with an operational semantics interpreting terms as
  incremental
  transformers that accept and produce finite prefixes of streams.
  Our main result is a powerful \emph{homomorphism theorem} (Theorem~\ref{thm:hom-prop}) guaranteeing that
  the result of a transformer does not depend on how the input stream is
  divided into prefixes. This theorem implies
  that the semantics is
  deterministic: all interleavings of
    parallel sub-streams yield the same final
    result.

  \item
  We present \lang{}, an experimental high-level functional language
  prototype based on the \core{} calculus that serves as a tool for exploring the potential of richly typed stream programming.
  We demonstrate by example how \lang{} enables type-safe
  programming for streams with complex patterns
  and how it prevents nondeterminism.
  Programming patterns from stream
  processing practice are elegantly supported by this richer model,
  including MapReduce-like pipelines,
  temporal integrity constraints, windowing, punctuation, parallelism
  control, routing, and side outputs.

\end{enumerate}

Section~\ref{sec:extended-motivation} explores some concrete cases
where \core{}'s structured types can prevent
common stream processing bugs and enable cleaner programming patterns.
Section~\ref{sec:core-lang} presents 
Kernel \core{}, a minimal
subset with just the features needed to state and understand the main
technical results.  Section~\ref{sec:full-core-lang}
extends this presentation to Full \core{}.
Section~\ref{sec:examples} develops several further
examples.
Sections~\ref{sec:related-work}~and~\ref{sec:future-work} discuss
related and future work.
An overview of our prototype implementation of \lang{}
can be found in Appendix~\ref{app:artifact};
technical details omitted from the main paper in
Appendix~\ref{app:technical}; and expanded versions of the examples in
Appendix~\ref{app:examples}.

\section{Motivating examples}
\label{sec:extended-motivation}

\emph{Types for temporal invariants.}
Consider a stream of brightness data coming from a motion
sensor, where
each event in the stream is a number between 0 and 100.
Suppose we want a stream transformer that acts as a threshold filter,
sending out a ``\texttt{Start}'' event when the brightness level
goes above level 50,
forwarding along brightness values until the level
dips below the threshold, and sending a final ``\texttt{Stop}''
event.  For example:
\[
  11,30,{\color{tmcolor} 52,56,53},30,10,{\color{pfxcolor} 60},10,\dots \ \ \Longrightarrow \ \  \texttt{Start},{\color{tmcolor}52,56,53},\texttt{Stop},\texttt{Start},{\color{pfxcolor} 60},\texttt{Stop},\dots
\]
%
%
The output of the transformer should satisfy the following {\em temporal invariant}: each start
event must be followed by one or more data events and then one end
event.
%
Conventional stream processing systems would give this
transformation a type like
$
  \texttt{Stream Int} \to \texttt{Stream (Start + Int + Stop)}
$,
which
expresses only the types of events in the output, not the temporal
invariant that the \texttt{Start} must come before all the data and
the \texttt{Stop} after.

These simple types are even more problematic when {\em consuming} streams.
Suppose another transformer wants to consume the output stream of type
\texttt{Stream (Start + Int + Stop)} and compute the average brightness between
each start/end pair. We know {\em a priori} that the stream is well bracketed,
but the type does not say so.
Thus, the second transformer must \emph{re-parse} the stream to
compute the averages, requiring additional logic for various special
cases (e.g., \texttt{Stop} before \texttt{Start}, empty \texttt{Start}/\texttt{Stop} pairs) that {cannot actually occur} in
the stream it will see.

In \core{}, we can express the required invariant with the type
$\left(\texttt{Start} \cdot \texttt{Int} \cdot \texttt{Int}^\star
\cdot \texttt{End} \right)^\star$, specifying that the stream consists
of a start message, at
least one \texttt{Int}, and an end message, repeatedly.
A well-typed transformer with this output type is guaranteed to
enforce this
invariant; conversely, a downstream transformer can assume that its
input will adhere to it.

\paragraph*{Enforcing deterministic parallelism}
A second limitation of homogeneous streams
is that they impose a
\emph{total ordering} on their component events.
In other words, for each pair of events in the stream, the transformer can
tell which came first.
This is problematic
in a world where
stream transformers work over data that is logically only partially
ordered---e.g., because it comes
from separate sources.%
\footnote{The same objection applies for stream processing
systems that impose total \emph{per key} ordering of a parallelized
stream---cf.~\texttt{KeyedStream} in Flink---since data associated with a given
key may also come from multiple sources in parallel.}

For example,
consider a system with two sensors, each producing one reading
per second and and sending them via different network
connections to a single transformer that averages them pairwise,
producing a composite reading each second.
A natural way to do this is to merge the two
streams into a single one, group adjacent pairs of elements (i.e.,
impose a size-two
tumbling window), and average the pairs. But
this is subtly wrong: a network delay could cause a pair
of consecutive elements in the merged stream to come from the same
sensor, after which the averages will all be bogus.

The problem with this transformer is that it is
not deterministic: its result
can depend on external factors like network latency.
Bugs of this type can easily occur in practice
\cite{schneider2013safe, mamouras2019data} and
can be very difficult to track down, since they may only manifest
under rare conditions~\cite{kallas2020diffstream}.

Once again, this is a failure of type structure.
In \core{}, we can prevent it by giving the merged stream the type
$\left(\texttt{Sensor1} \| \texttt{Sensor2}\right)^\star$, capturing
the fact that it is a stream of {\em parallel
pairs} of readings from the two sensors.
We can write a strongly typed merge operator that produces this type,
given parallel streams of type $\texttt{Sensor1}^\star$ and
$\texttt{Sensor2}^\star$. This merge operator
is deterministic (indeed, all well-typed \core{} programs are, as we
show in Section~\ref{sec:determinism}); operationally, it waits for events to
arrive on both of its input streams before sending them along as a
pair.

\section{Kernel \core{}}
\label{sec:core-lang}
In this section, we define the most important constructors of
stream types and
the corresponding features of the term language; these form the ``kernel'' of
the \core{} calculus.
The rest of the types and terms of Full \core{} will layered on bit by
bit in Section~\ref{sec:full-core-lang}.

The \emph{concatenation}
constructor $\cdot$ describes
streams that \emph{vary} over time: if $s$ and $t$ are stream
types, then $s \cdot t$ describes a stream on which all the elements of
$s$ arrive first, followed by the elements of $t$. A producer of a stream of type
$s \cdot t$ must first produce a stream of type $s$ and then a stream of type $t$,
while a consumer can assume that the incoming data will first consist of data of type $s$
and then of type $t$. The transition point between the $s$ and $t$ parts is handled
automatically by \core{}'s semantics: the underlying data of a stream of type $s \cdot t$ includes
a \emph{punctuation marker} \cite{tuckerPunc}
indicating the cross-over. One consequence of this is that, unlike Kleene Star
for regular languages, streams of type $s^\star$ are distinguishable from
streams of type $s^\star \cdot s^\star$ because a transformer accepting the
latter can see when its input crosses from the first $s^\star$ to the second.

On the other hand, the \emph{parallel} stream type $s \| t$ describes a stream
with two parallel substreams of types $s$ and $t$.  Semantically, the $s$ and
$t$ components are produced and consumed independently: a transformer that
produces $s \| t$ may send out an entire $s$ first and then a $t$, or an entire
$t$ and then the $s$, or any interleaving of the two.  Conversely, a transformer
that accepts $s \| t$ must handle all these possibilities uniformly by
processing the $s$ and $t$ parts independently.  To enable this, each element in
the parallel stream is tagged to indicate which substream it belongs to.  This
means that streams of type $s \| t$ are isomorphic, but not identical, to
streams of type $t \| s$, and similarly
$\texttt{Int}^\star \| \texttt{Int}^\star$ is not the same as $\texttt{Int}^\star$.

Parallel types can be combined with concatenation types in interesting ways.
For example, a stream of type $(s \| t) \cdot r$ consists of a stream
of interleaved items from $s$ and $t$, followed (once all the $s$'s and
$t$'s have arrived) by a
stream of type $r$.  By contrast, a stream of type
$(s \cdot t) \| (s' \cdot t')$ has two interleaved components, one
a stream described by $s$ followed by a stream described by $t$ and the
other an $s'$ followed by a $t'$.  The fact that the parallel type is on the outside
means that the change-over points from $s$ to $t$ and $s'$ to $t'$ are
completely independent.

The base type $\onet$ describes
a stream containing just one data item, itself a unit
value.  The other base type is $\epst$,
the type of the empty stream
containing no data; it is the unit for both the $\cdot$ and $\|$
constructors---i.e.,
$s \cdot \epst$, $\epst \cdot s$, $\epst \| s$ and $s \| \epst$
are all equivalent to $s$, in the sense that there are
\core{} transformers that convert between them.

In summary, the Kernel \core{} stream types are given by the
grammar on the top left in Figure~\ref{fig:core-typing-rules}.
%
(So far, these types can only describe streams of
fixed, finite size.  In
Section~\ref{sec:star} we will enrich the kernel type system with
unbounded streams via the Kleene star type $s^\star$.)

What about terms? Recall that our goal is to develop a language of core terms $e$, typed by
stream types, where well-typed terms $x : s \vdash e : t$ are interpreted as
stream transformers accepting a stream described by $s$ and
producing one described by $t$.
A term $e$ runs by accepting some inputs as
described by $s$, producing some outputs as described by $t$ and then stepping
to a new term $e'$ with an updated type, ready to accept the rest of the input and produce the rest of the output. This process happens reactively: output
is only produced when an input arrives. The formal semantics of \core{} are
described in Section~\ref{subsec:incremental}.

To represent stream transformers
with multiple parallel and sequential inputs,
we draw upon results from proof theory for insight.  Both the
types $s \cdot t$ and $s \| t$ are \emph{product types}, in the sense that
a stream of either of these types contains both the data of a stream of type $s$
and a stream of type $t$---although the temporal structure differs between the two.
A standard observation from proof theory is that, in situations where
a logic or type theory includes two products, the corresponding typing
judgment requires a context with two {\em context formers}.%
\footnote{%
Such \emph{bunched} contexts were first introduced in
the logic of Bunched Implication~\cite{o1999logic}, the basis of
modern separation logic~\cite{reynolds2002separation}. Our bunched contexts differ from those of BI by
the choice of structural rules: our substructural type former is affine ordered, while the BI one is linear.
}

The first context former, written with a comma
($\commactx{\Gamma}{\Delta}$), describes inputs to a transformer
arriving in parallel, one component
structured according to $\Gamma$ and the other according to
$\Delta$. The second context former, written with a semicolon
($\semicctx{\Gamma}{\Delta}$) describes inputs that will first arrive from
the environment according
to $\Gamma$, then according to $\Delta$.

These interpretations are enforced by restricting the ways that these contexts can be manipulated using \emph{structural rules}.
Comma contexts can be manipulated in all the ways standard contexts can: they can be
reordered---from $\commactx{\Gamma}{\Delta}$ to $\commactx{\Delta}{\Gamma}$
--- duplicated, and dropped.
Semicolon contexts, on the other hand, are ordered and affine: a context $\semicctx{\Gamma}{\Delta}$ cannot be freely rewritten to
a context $\semicctx{\Delta}{\Gamma}$, and we cannot duplicate a context $\Gamma$ into one like $\semicctx{\Gamma}{\Gamma}$.
These restrictions enforce the meaning of $\semicctx{\Gamma}{\Delta}$ as data arriving according to $\Gamma$ and then $\Delta$:
to exchange them would be to allow a consumer to assume that the data is sent in the opposite order, and to duplicate is to
assume that the data input will be replayed.

In summary, our type system is {\em substructural}. The semicolon context former
is ordered (no exchange) and affine (no contraction), while the comma
context former is fully structural. Both context formers are associative, with
the empty context serving as a unit for each. The full list of structural rules can be found in Appendix~\ref{app:ctxsub}.
Formally,
stream contexts are drawn from the grammar at the top right of
Figure~\ref{fig:core-typing-rules}.





\subsection{Kernel Typing Rules}

\begin{figure}
    \begin{mathpar}
       \hspace*{.7in}
       s,t,r := 1 \mid \varepsilon \mid s \cdot t \mid s \| t
       \hspace*{.6in} %
       \Gamma ::= \cdot \mid
         {\Gamma},{\Gamma}
         \mid {\Gamma};{\Gamma} \mid x : s

     \vspace*{2ex}

     \ENDOFLINE

      \INFERRULE[T-Par-R]{
          \Gamma \vdash e_1 : s\\
          \Gamma \vdash e_2 : t
      }{
          \Gamma \vdash \parpairTm{e_1}{e_2} : s \| t
      }

      \INFERRULE[T-Par-L]{
        \Gamma(\commactx{x : s}{y : t}) \vdash e : r
      }{
        \Gamma(z : s \| t) \vdash \letparTm{x}{y}{z}{e} : r
      }

      \ENDOFLINE

      \INFERRULE[T-Cat-R]{
          \Gamma \vdash e_1 : s\\
          \Delta \vdash e_2 : t
      }{
        \semicctx{\Gamma}{\Delta} \vdash \catpairTm{e_1}{e_2} : s \cdot t
      }

      \INFERRULE[T-Cat-L]{
          \Gamma(\semicctx{x : s}{y : t}) \vdash e : r
      }{
          \Gamma(z : s \cdot t) \vdash \letcatTm{t}{x}{y}{z}{e} : r
      }

      \ENDOFLINE

      \INFERRULE[T-Eps-R]{ }{
        \Gamma \vdash \epsTm : \varepsilon
      }

      \INFERRULE[T-One-R]{ }{
        \Gamma \vdash \oneTm : 1
      }

      \INFERRULE[T-Var]{ }{\Gamma(x:s) \vdash x : s}

      \ENDOFLINE

      \INFERRULE[T-SubCtx]{
        \Gamma \leq \Gamma'\\
        \Gamma' \vdash e : s
      }{
        \Gamma \vdash e : s
      }

  \end{mathpar}
  \BEFORECAPTIONSPACE
  \caption{Kernel \core syntax and typing rules}
  \label{fig:core-typing-rules}
\end{figure}

The typing rules for Kernel \core{} are collected in
Figure~\ref{fig:core-typing-rules}. The typing judgment,
written $\Gamma \vdash e : s$, says that $e$ is a stream transformer
from a collection of streams structured like $\Gamma$, to a single stream structured like $s$.

The most straightforward typing rule is the right rule for parallel
(\ruleName{T-Par-R}). It says that, from a context $\Gamma$, we can
produce a stream of type $s \| t$
by producing $s$ and $t$ independently from $\Gamma$, using transformers $e_1$ and $e_2$. We write
the combined transformer as a ``parallel pair''
$\parpairTm{e_1}{e_2}$. Semantically, it
operates by copying the inputs arriving on $\Gamma$,
passing the copies to $e_1$ and $e_2$,
and pairing up the outputs into a parallel stream.
Similarly, the \ruleName{T-Cat-R} rule is used to produce a stream of type $s \cdot t$.  It
uses a similar pairing syntax---if term $e_1$ has type $s$ and $e_2$
has type $t$, then the ``sequential pair'' $\catpairTm{e_1}{e_2}$ has type $s \cdot t$---but the context in the conclusion differs. Since $e_1$ needs
to be run before $e_2$, the part of the input stream that $e_1$
depends on ($\Gamma$) must
arrive before the part that $e_2$ depends on ($\Delta$). Semantically,
this term will operate
by accepting data from the $\Gamma$ part of the context and running
$e_1$; once the $\Gamma$ part is used up
and the $\Delta$ part starts to arrive, it will switch to running $e_2$.

These right rules describe how to {\em produce} a stream of parallel
or concatenation type.  The corresponding left
rules describe how to {\em use} a variable of one of these types appearing
somewhere in the context.
Syntactically, the terms take the form of let-bindings that
deconstruct variables of type $s \cdot t$
(or $s \| t$) as pairs of variables of type $s$ and $t$,
connected by $;$ (or $,$).
We use the standard BI notation $\Gamma(-)$ for a context with a
hole, and
$\Gamma(\Delta)$ when this hole has been filled with the context
$\Delta$. In particular, $\Gamma(x : s)$ is a context with a distinguished
variable $x$.

The \ruleName{T-Par-L} rule says that if $z$ is a variable of type $s \| t$
somewhere in the context, we can replace its binding with with a pair
of bindings for variables $x$ and $y$
of types $s$ and $t$ and use these in a continuation term $e$ of final
type $r$.  When
typing $e$, the variables $x$ and $y$ appear in the same
position as the original variable $z$, but separated by a
comma---i.e., $x$ and $y$
are assumed to arrive in parallel.  Similarly, the rule
\ruleName{T-Cat-L} says that, if a
variable $z$ of type $s \cdot t$ appears somewhere in the context, it can be
let-bound to a pair of variables $x$ and $y$ of types $s$ and $t$ that are again
used in the continuation $e$.  This
time, though, $x$ and $y$ are separated by a semicolon---i.e., the
substream bound to $x$ will arrive and be processed first, followed by
the substream bound to $y$.

\ruleName{T-Eps-R} and \ruleName{T-One-R} are the right rules for the two base types,
witnessed by the terms $\epsTm$ and $\oneTm$.  Semantically, $\epsTm$ does
nothing: it accepts inputs on $\Gamma$ and produces no output. On the other hand,
$\oneTm$ emits a unit value as soon as it receives its first
input and never emits anything else.

The variable rule (\ruleName{T-Var}) says that, if $x : s$ is a variable
somewhere in the context, then
we can simply send it along the output stream. Semantically, it works by dropping everything
in the context except for the $s$-typed data for $x$, which it forwards along.

The rule \ruleName{T-SubCtx} bundles together all of the
structural rules as a
subtyping relation on contexts. For example, the weakening rule for
semicolon contexts is written, $\semicctx{\Gamma}{\Delta} \leq \Gamma$
and the comma exchange rule is $\commactx{\Gamma}{\Delta} \leq
\commactx{\Delta}{\Gamma}$.


\subsubsection*{Examples and Non-Examples}
To show the typing rules in action, here are two small examples
of transformers written in Kernel \core{}, as well as three examples of programs which are \emph{rejected}
by the type system.
The first example is a simple ``parallel-swap''
transformer, which accepts a stream $z$ of type $s \| t$, and outputs
a stream of type $t \| s$, swapping the ``positions'' of the parallel substreams:
$$z : s \| t \vdash \letparTm{x}{y}{z}{\parpairTm{y}{x}} : t \| s$$
It works by splitting the variable $z : s \| t$ into variables $x : s$
and $y : t$ and yielding a parallel pair with the order reversed.

\begin{wrapfigure}{R}{6.8cm}
\vspace*{-2ex}
\hspace*{-1cm}
\begin{minipage}{8cm}
\begin{mathpar}
  \INFERRULE[T-Cat-L]{
    \INFERRULE[T-Cat-R]{
      \INFERRULE[T-Var]{\lightning}{x : s \vdash y : t}\\
      \INFERRULE[T-Var]{\lightning}{y : t \vdash x : s}
    }{\semicctx{x : s}{y : t} \vdash \catpairTm{y}{x} : t \cdot s}
  }{
    z : s\cdot t \vdash \letcatTm{s}{x}{y}{z}{\catpairTm{y}{x}} : t \cdot s
  }
\end{mathpar}
\end{minipage}
\vspace*{-3ex}
\end{wrapfigure}
The first and most important \emph{non}-example is the lack of a corresponding
``cat-swap'' term, which would accept a stream $z$ of type $s \cdot t$, and
produce a stream of type $t \cdot s$. This program is undesirable because it is not
implementable without a space leak\jwc{Cite}. Implementing it requires the entire stream of type $s$
to be saved in memory to emit it after the stream of type $t$.%
\footnote{A program with this behavior is
implementable in \core{}, but requires a special program construct---see Section~\ref{sec:wait}---ensuring that leaky programs like
this one cannot be written accidentally.}
The natural term for this program would be $\letcatTm{s}{x}{y}{z}{\catpairTm{y}{x}}$,
but this does not typecheck, as attempting to build a derivation gets
stuck.
Applying the syntax-directed rules gets us to a point where we must show that $y$ has type $t$ in a context
with only $x$, and that $x$ has type $s$ in a context with only $y$. This is because
the \ruleName{T-Cat-R} splits the context, but the variables are listed in the opposite order than we'd need.
The lack of a structural rule to let us permute the $x$ and $y$ in the context means that there is nothing to do here, and so a typechecker will reject this program.

The second example is a ``broadcast'' transformer, which takes a variable $x : s$ and outputs a stream of type $s \| s$, duplicating the variable of type
$s$, and sending it out to two parallel outputs:
$x : s \vdash \parpairTm{x}{x} : s \| s$.

Another non-example is the ``replay'' transformer, which would take a variable $x : s$ and produce a stream $s \cdot s$
which repeats the input stream twice. This is the concat type equivalent of the broadcast transformer, but it is undesirable
for the same reason as the cat-swap program: it would require saving the entire incoming stream of type $s$ in order to replay it.
This time, the failure of the natural term $\catpairTm{x}{x}$ to typecheck is down to a lack of contraction rule for semicolon contexts:
we are not permitted to turn a context $x : s$ into a context $\semicctx{x : s}{x : s}$.

The last non-example is a ``tie-breaking'' transformer, which would
take a stream $z : \texttt{Int} \| \texttt{Int}$
of two ints in parallel and produce a stream of type $\texttt{Int}$,
forwarding along the \texttt{Int} that arrived first.
This program (like others that require inspecting the interleaving of
data in a stream of type $s \| t$) is
not expressible. In Section~\ref{sec:determinism}, we'll prove that a well-typed program cannot implement this behavior.

\subsection{Prefixes and Semantics}
\label{subsec:incremental}
Next we define the semantics of Kernel
\core{}.
The natural notion of ``values'' in this semantics is finite
prefixes of streams:
the meaning of a well-typed term $\Gamma \vdash e : s$ is
a function that accepts an environment mapping variables in $\Gamma$
to prefixes of streams and produces a prefix of a stream of type $s$.

Because the streams that \core{} programs operate over are more structured than
traditional homogeneous streams---including cross-over punctuation in streams
of type $s \cdot t$ and disambiguating tags in streams of type $s \| t$---the
prefixes are also more structured.  A prefix in \core{} is not a
simple sequence of data items, but a
structured value whose possible shapes are determined by its type.

\begin{figure}
  \begin{mathpar}
    \inferrule{ }{\prefixHasType{\epsp}{\epst}}

    \inferrule{ }{\prefixHasType{\onepA}{\onet}}

    \inferrule{ }{\prefixHasType{\onepB}{\onet}}

    \ENDOFLINE

    \inferrule{\prefixHasType{p}{s} \\\\ \prefixHasType{p'}{t}}{\prefixHasType{\parp{p}{p'}}{s \| t}}

    \hspace*{-2em}

    \inferrule{\prefixHasType{p}{s}}{\prefixHasType{\catpA{p}}{s \cdot t}}

    \inferrule{\prefixHasType{p'}{t} \\\\ \prefixHasType{p}{s} \\ \isMaximal{p}}{\prefixHasType{\catpB{p}{p'}}{s \cdot t}}

  \end{mathpar}
  \BEFORECAPTIONSPACE
  \caption{Prefixes for Types}
  \label{fig:prefix-rules}
\vspace*{-2ex}
\end{figure}

For example, there are two prefixes of a stream of type
$\onet$: the empty prefix, written $\onepA$, and the prefix
containing the single element (), written $\onepB$.  Similarly, the unique
stream of type $\epst$ has a single prefix, the empty prefix, which we write
$\epsp$.

What about $s \| t$?  A parallel stream of
type $s \| t$ is conceptually a pair of independent streams of type $s$ and $t$,
so
a prefix of a parallel stream should be a pair $\parp{p_1}{p_2}$, where $p_1$ is a
prefix of a stream of type $s$, and $p_2$ is a prefix of a stream of type $t$.
Crucially, this definition encodes no information
about any interleaving of $p_1$ and $p_2$:
the prefix
$\parp{p_1}{p_2}$ equally represents a situation where all of $p_1$ arrived first and
then all of $p_2$, one where $p_2$ arrived before $p_1$, and many others where
the elements of $p_1$ and $p_2$ arrived in some interleaved order.
In a nutshell, this definition is what guarantees deterministic processing.
By representing all possible interleavings using the same prefix
value, we ensure that a transformer that
operates on these values cannot possibly depend on ordering information that isn't
present in the type.

Finally, let's consider the prefixes of streams of type $s \cdot
t$.  One case is a
prefix that only includes data from $s$ because it cuts off before reaching the point
where the
$s\cdot t$ stream stops carrying elements of $s$ and starts on $t$. We write such a
prefix as $\catpA{p}$, with $p$ a prefix of type $s$. The other case
is where the prefix does include the crossover point---i.e., it consists
of a ``completed'' prefix of $s$ plus a prefix of $t$.
We write this as $\catpB{p}{p'}$, with $p$ a prefix of $s$ and $p'$
a prefix of $t$. The requirement that $p$ be completed is formalized by the judgment $\isMaximal{p}$,
which ensures that the prefix $p$ describes an entire completed stream
(see Appendix~\ref{app:maximal}).
We formalize all these possibilities as a judgment
$\prefixHasType{p}{s}$, shown in Figure~\ref{fig:prefix-rules}.

Every type $s$ has a distinguished \emph{empty prefix},
written $\emp{s}$ and defined by straightforward recursion on $s$
(see Appendix~\ref{app:emp}).
We lift the idea of prefixes from types to contexts, defining an
\emph{environment} $\eta$ for a
context $\Gamma$ to be a mapping from the variables $x : s$ in $\Gamma$ to
prefixes of the corresponding types $s$.  We write this with a judgment $\envHasType{\eta}{\Gamma}$
(Figure~\ref{fig:env-rules}). Along with ensuring that $\eta$ has well-typed
bindings for all variables, the judgment ensures that the prefixes respect the
order structure of the context. In particular, an environment $\eta$
for a semicolon context $\semicctx{\Gamma}{\Delta}$ must assign prefixes
\emph{in order}: the prefixes for $\Gamma$, the earlier part of the context, must all be maximal before the prefixes $\Delta$ can begin.
In other words, either $\eta$ assigns maximal prefixes to every variable in $\Gamma$---which we write $\maximalOn{\eta}{\Gamma}$---or $\eta$ assigns empty prefixes to every variable in $\Delta$---which we write $\emptyOn{\eta}{\Delta}$.

\begin{figure}
  \begin{mathpar}
    \inferrule{ }{\envHasType{\eta}{\cdot}}

    \inferrule{
      \eta(x) \mapsto p\\
      \prefixHasType{p}{s}
    }{
      \envHasType{\eta}{x : s}
    }

    \inferrule{
      \envHasType{\eta}{\Gamma}\\
      \envHasType{\eta}{\Delta}\\
    }{
      \envHasType{\eta}{\commactx{\Gamma}{\Delta}}
    }

    \ENDOFLINE

    \inferrule{
      \envHasType{\eta}{\Gamma}\\
      \envHasType{\eta}{\Delta}\\
      \left(\maximalOn{\eta}{\Gamma}\right) \vee \left(\emptyOn{\eta}{\Delta}\right)
    }{
      \envHasType{\eta}{\semicctx{\Gamma}{\Delta}}
    }
  \end{mathpar}
  \BEFORECAPTIONSPACE
  \caption{Environments for Contexts}
  \label{fig:env-rules}
\vspace*{-2ex}
\end{figure}

One might worry that these structured stream prefixes might be
incompatible with a future distributed implementation atop an
existing stream processing substrate. Fortunately, they are not: by viewing a
\core{} stream as a series of single-event prefixes, each
consisting of a data item plus some extra tag bits, we recover
the traditional homogeneous view. Moreover, this wire representation
incurs only a constant overhead: the maximum
size of the tag bits on a stream element of type $s$ is bounded by the
syntactic depth of $s$ (See Appendix~\ref{app:eventsApp}).


\subsubsection*{Semantics}
We describe how well-typed \core{} terms execute with an operational semantics.
Given a well-typed term $\Gamma \vdash e : s$ and an input environment
$\envHasType{\eta}{\Gamma}$, the semantics describes how to run $e$ with $\eta$
to produce an output prefix $p$. It also describes how to produce a
``resultant'' term $e'$, whose purpose is to continue the computation once further
data on the input stream data arrives.
Formally the semantics is given by a judgment
$\prefixstep{\eta}{e}{ }{e'}{p}$, which we pronounce ``running the core term $e$
on the input environment $\eta$ yields the output prefix $p$ and
steps to $e'$.''
The rules for this judgment are gathered in
Figure~\ref{fig:prefix-sem-selected-rules} and described below; the
full set of rules for all of \core{} can be found in
Appendix~\ref{app:semantics}.

The following theorem establishes the soundness of the Kernel \core{} semantics,
formalizing the intuitive description given above:
If we run a well-typed core term $e$ on an environment $\eta$ of the context type, it will
return a prefix $p$ with the result type $s$, and step to a term $e'$ which is well typed
in context ``the rest of'' $\Gamma$ after $\eta$ and has type ``the rest of'' $s$
after $p$. The ``rest'' of a type/context after a prefix/environment is, intuitively, its {\em derivative}
with respect to the prefix/environment,
in the sense of standard Brzozowski derivatives of regular expressions
\cite{brzozowski1964derivatives} --- we make this formal in Section~\ref{sec:deriv}.
Most critically, the types of the variables in $e$ and $e'$ are different: if $x$ has type $s$ in $e$,
then $x$ has type $\delta_{\eta(x)} s$ in $e'$, having already consumed $\eta(x)$.

\begin{theorem}[Soundness of the Kernel \core{} Semantics]
  \label{thm:prefix-sem-correct-core}
  Suppose: $\Gamma \vdash e : s$ and $\envHasType{\eta}{\Gamma}$.
  Then, there are $p$ and $e'$ such that
  $\prefixstep{\eta}{e}{ }{e'}{p}$, with
  $\prefixHasType{p}{s}$ and $\deriv{\eta}{\Gamma} \vdash e' : \deriv{p}{s}$
\end{theorem}

Appendix~\ref{app:incr-sem-thms} presents the proof of the soundness
theorem for Full \core{} (see Section~\ref{sec:full-core-lang}).

In light of this theorem, our operational semantics can be thought
of as defining a reactive state machine. Well typed terms $\Gamma \vdash e : s$ are the states, while
the semantics judgment defines the  transition function: when new inputs $\eta$
arrive, we step the semantics to produce an output prefix $p$, and step to a new
state $\deriv{\eta}{\Gamma} \vdash e' : \deriv{p}{s}$. This form of semantics --- a state machine with terms themselves as states, typed by derivatives ---
predates our work, having been pioneered by the Esterel programming language \cite{esterel}.

\begin{figure}
  \begin{mathpar}

        \INFERRULE[S-Var]{
            \eta(x) \mapsto p
        }{
            \prefixstep{\eta}{x}{ }{x}{p}
        }

        \ENDOFLINE

        \INFERRULE[S-Par-R]{
            \prefixstep{\eta}{e_1}{}{e_1'}{p_1}\\
            \prefixstep{\eta}{e_2}{}{e_2'}{p_2}\\
        }{
            \prefixstep{\eta}{\parpairTm{e_1}{e_2}}{ }{\parpairTm{e_1'}{e_2'}}{\parp{p_1}{p_2}}
        }

        \ENDOFLINE

        \INFERRULE[S-Par-L]{
            \eta(z) \mapsto \parp{p_1}{p_2}\\
            \prefixstep{\eta[x \mapsto p_1,y\mapsto p_2]}{e}{ }{e'}{p'}
        }{
            \prefixstep{\eta}{\letparTm{x}{y}{z}{e}}{ }{\letparTm{x}{y}{z}{e'}}{p'}
        }

        \ENDOFLINE

        \INFERRULE[S-Cat-R-1]{
            \prefixstep{\eta}{e_1}{ }{e_1'}{p}\\
            \neg\left(\isMaximal{p}\right)
        }{
            \prefixstep{\eta}{\catpairTm{e_1}{e_2}}{ }{\catpairTm{e_1'}{e_2}}{\catpA{p}}
        }

        \ENDOFLINE

        \INFERRULE[S-Cat-R-2]{
            \prefixstep{\eta}{e_1}{ }{e_1'}{p}\\
            \isMaximal{p}\\
            \prefixstep{\eta}{e_2}{ }{e_2'}{p'}
        }{
            \prefixstep{\eta}{\catpairTm{e_1}{e_2}}{ }{e_2'}{\catpB{p}{p'}}
        }

        \ENDOFLINE

        \INFERRULE[S-Cat-L-1]{
          \eta(z) \mapsto \catpA{p}\\
          \prefixstep{\eta[x \mapsto p, y \mapsto \emp{t}]}{e}{ }{e'}{p'}
      }{
          \prefixstep{\eta}{\letcatTm{t}{x}{y}{z}{e}}{ }{\letcatTm{t}{x}{y}{z}{e'}}{p'}
      }

        \ENDOFLINE

      \INFERRULE[S-Cat-L-2]{
          \eta(z) \mapsto \catpB{p}{p'}\\
          \prefixstep{\eta[x \mapsto p,y\mapsto p']}{e}{ }{e'}{p''}
      }{
          \prefixstep{\eta}{\letcatTm{t}{x}{y}{z}{e}}{ }{\cutTm{x}{\sinkTm{p}}{e'[z/y]}}{p''}
      }

        \ENDOFLINE

        \INFERRULE[S-Eps-R]{ }{
          \prefixstep{\eta}{\epsTm}{ }{\epsTm}{\epsp}
        }

        \ENDOFLINE

        \INFERRULE[S-One-R]{ }{
          \prefixstep{\eta}{\oneTm}{ }{\epsTm}{\onepB}
        }

        \ENDOFLINE

  \end{mathpar}
  \BEFORECAPTIONSPACE
  \caption{Incremental semantics of Kernel \core{}}
  \label{fig:prefix-sem-selected-rules}
\end{figure}

\subsubsection*{Semantics of the Right Rules}
\cutifneeded{The right rules for
parallel and concatenation are the simplest to understand.}
For \ruleName{S-Par-R}, we accept an environment $\eta$ and use it to
run the component terms $e_1$ and $e_2$, independently producing outputs
$p_1$ and $p_2$ and stepping to new terms $e_1'$ and $e_2'$.
The pair term $\parpairTm{e_1}{e_2}$ then steps to
$\parpairTm{e_1'}{e_2'}$ and produces the output $\parp{p_1}{p_2}$.

There are two rules, \ruleName{S-Cat-R-1} and \ruleName{S-Cat-R-2}, for running the concatenation pair $\catpairTm{e_1}{e_2} : s \cdot
t$. In either case, we begin by running $e_1$ with the environment $\eta$,
producing a prefix $p$ and term $e_1'$. If $p$ is not maximal, we stop
there:
future inputs will allow the first component to produce the rest of $s$, so it is
not yet time to start running $e_2$ to produce $t$. This case is
handled by \ruleName{S-Cat-R-1},
where the resulting term is $\catpairTm{e_1'}{e_2}$ and the output prefix is $\catpA{p}$.

On the other hand, if $p$ {is} maximal, then we run $e_2$,
which steps to $e_2'$ and produces a prefix $p'$ using rule \ruleName{S-Cat-R-2}, where the entire term then outputs $\catpB{p}{p'}$,
and steps to $e_2'$. Note that the pair is eliminated in the process: we step from
$\catpairTm{e_1}{e_2}$ to just $e_2'$. This is because we are done producing the $s$ part of the $s \cdot t$,
and so a subsequent step of evaluation only has to run $e_2'$ to produce the rest of the $t$.

\subsubsection*{Semantics of Variables}
The variable semantics S-Var is a simple variable lookup. We look up the prefix bound to the variable $x$
in the environment, return it, and then step to $x$ itself.

\subsubsection*{Semantics of Left Rules}
The semantics of the left rules for
concatenation and parallel are similar, both accepting an environment $\eta$ with a binding for $z : s
\otimes t$, where $\otimes $ is one of the two products, binding variables $x$ and
$y$ of type $s$ and $t$ to the two components of the product, and using the updated environment to run
the continuation term.

In the case of the semantics of the left rule for parallel (\ruleName{S-Par-L}),
looking up $z$ of type $s \| t$ will always yield a prefix $\parp{p_1}{p_2}$.
The rule binds $p_1$ to $x$ and $p_2$ to $y$ and runs the continuation term,
stepping to $e'$ and producing the output prefix $p$. Then the whole term
steps to $\letparTm{x}{y}{z}{e'}$
and produces $p$.

The left rule for concatenation has two cases, depending on what kind of prefix comes back from the
lookup for $z$. If the lookup yields is $\catpA{p}$, the rule
\ruleName{S-Cat-L-1} applies. Since no data for $y$ has arrived, we
bind $y$ to $\emp{t}$, the empty prefix of type $t$, and run the continuation \footnote{This need to compute $\emp{t}$ at runtime to bind to $y$ is the reason that the term for \texttt{T-Cat-L}, $\letcatTm{t}{x}{y}{z}{e}$, includes a $t$ in the syntax. In Section~\ref{sec:full-core-lang}, the case analysis expressions for star types and sum types will have similar annotations for the same reason.}.
If the result comes back as $\catpB{p}{p'}$, then the rule \ruleName{S-Cat-L-2}
applies, so we run the continuation with $p$ and $p'$ bound to $x$ and
$y$.

Both rules output the prefix resulting from running the continuation, but they
step to different resulting terms.
If $\eta(z) = \catpA{p}$, then the resulting term must be another use of Cat-L:
the variable $z$ still expects to get some more of the first component of the
concatenation, and then the second component. If $\eta(z) = \catpB{p}{p'}$ on
the other hand, the $z$ stream has crossed over to the second part. In this
case, we close over the (now not-needed) $x$ variable in $e'$, and connect $z$
to the $y$ input of $e'$ by substituting $y$ for $z$.

\subsubsection{Derivatives}
\label{sec:deriv}
When $\prefixHasType{p}{s}$,
we write $\deriv{p}{s}$ for the derivative \cite{brzozowski1964derivatives} of $s$ by $p$: the type of streams that result after a prefix
of type $p$ has been ``chopped off'' the beginning of a stream of type $s$.
Because this operation is partial---$\deriv{p}{s}$ is only defined when $\prefixHasType{p}{s}$---we formally define this as a
a 3-place relation, written as $\derivrel{p}{s}{s'}$ and
pronounced as ``the
derivative of $s$ with respect to $p$ is $s'$'' (see
Figure~\ref{fig:deriv-rules}).

The derivative of the type $\onet$ with respect to the empty prefix
$\onepA$ is
$\onet$ (the rest of the stream is the entire stream), and its derivative
with respect to the full prefix $\onepB$ is $\epst$ (there is no more
stream left
after the unit element has arrived). For parallel, the derivative is taken component-wise.
The interesting cases are those for the concatenation type. If the prefix
has the form $\catpA{p}$, the derivative $\deriv{\catpA{p}}{s \cdot t}$ is
$\left(\deriv{p}{s}\right) \cdot t$, i.e., some of the $s$ has gone
by but not all, and once it does we
still expect $t$ to come after it. On the other hand, if the prefix
has the form $\catpB{p}{p'}$, the derivative $\deriv{\catpB{p}{p'}}{s
  \cdot t}$ is just $\deriv{p'}{t}$, i.e.,
the $s$ component is complete, and the rest of the stream is just the part of $t$ after $p'$.

\begin{figure}
  \begin{mathpar}
      \inferrule{ }{\derivrel{\epsp}{\epst}{\epst}}

      \inferrule{ }{\derivrel{\onepA}{\onet}{\onet}}

      \inferrule{ }{\derivrel{\onepB}{\onet}{\epst}}

      \inferrule{\derivrel{p}{s}{s'}}{\derivrel{\catpA{p}}{s \cdot t}{s' \cdot t}}

      \ENDOFLINE

      \inferrule{\derivrel{p'}{t}{t'}}{\derivrel{\catpB{p}{'p}}{s \cdot t}{t'}}

      \inferrule{\derivrel{p}{s}{s'} \\ \derivrel{p'}{t}{t'}}{\derivrel{\parp{p}{p'}}{s \| t}{s' \| t'}}
  \end{mathpar}
  \BEFORECAPTIONSPACE
  \caption{Derivatives}
  \label{fig:deriv-rules}
\end{figure}

This definition gets lifted to contexts and environments pointwise: if $x : s$ is a variable in $\Gamma$, the derivative of $\deriv{\eta}{\Gamma}$
has $x : \deriv{\eta(x)}{s}$ in the same location.

\subsection{The Homomorphism Property and Determinism}
The semantics is designed to run a stream
transformer on an ``input chunk'' of any size, from individual input
events one at a time all the way up to the entire stream at once.
The cost of this flexibility is that it raises the question of {\em
  coherence}---i.e., whether
we are guaranteed to arrive at the {\em same} final output depending on how we
carve up a transformer's input into a series of prefixes.
Fortunately, this is indeed guaranteed.

Coherence is a corollary of our main technical result, a {\em
  homomorphism theorem} that
says running a term $e$ on an environment $\eta$ and then running the resulting
term $e'$ on an environment $\eta'$ of appropriate type produces the same
end result as running $e$ on the combined environment.

\begin{theorem}[Homomorphism Theorem]
  \label{thm:hom-prop}
  Suppose
    (1) $\Gamma \vdash e : s$,
    (2) $\envHasType{\eta}{\Gamma}$,
    (3) $\envHasType{\eta'}{\deriv{\eta}{\Gamma}}$,
    (4) $\prefixHasType{p}{s}$,
    (5) $\prefixHasType{p'}{\deriv{p}{s}}$,
    (6) $\prefixstep{\eta}{e}{ }{e'}{p}$, and
    (7) $\prefixstep{\eta'}{e'}{ }{e''}{p'}$.
  Then, if $\prefixstep{\prefixConcat{\eta}{\eta'}}{e}{ }{e'''}{p''}$,
  we have $p'' = \prefixConcat{p}{p'}$, and $e''' = e'$
\end{theorem}
(The proof goes
  by induction on the derivation of $\prefixstep{\eta}{e}{ }{e'}{p}$,
  inverting everything in sight. See Appendix~\ref{app:incr-sem-thms} for full details.)
The operation $\prefixConcat{p}{p'}$ here is
{\em prefix concatenation},
which takes a prefix $p$ of type $s$ and a prefix $p'$ of type $\deriv{p}{s}$
and produces the prefix of type $s$ that is first $p$ and then $p'$.
Formally, this is defined as a 4-place partial inductive relation $\prefixConcatRel{p}{p'}{p''}$,
which is defined when $p$ and $p'$ have types $s$ and $\deriv{p}{s}$, respectively.
The operation $\prefixConcatRel{\eta}{\eta'}{\eta''}$ does the same for environments.
See Appendix~\ref{app:concatenation}.


\label{sec:determinism}
The homomorphism theorem not only justifies running the semantics of prefixes of any size; it also implies deterministic processing of parallel streams.
Intuitively, determinism states that the results of
a stream transformer do not depend on the particular order in which parallel data
arrives. We formalize this through the following scenario.  Suppose
$\commactx{\Gamma}{\Gamma'} \vdash e : s$ is a term with two parallel contexts
serving as its input, and suppose that $\eta$ is an environment for $\commactx{\Gamma}{\Gamma'}$.
Write $\eta_1 = \eta|_{\Gamma}$ and $\eta_2 = \eta|_{\Gamma'}$, for the
restrictions of $\eta$ to the variables in $\Gamma$ and $\Gamma'$,
respectively.
Now, there are two
different ways of running $e$ on this data.  One is to first run $e$ on
$\eta_1 \cup \emp{\Gamma'}$ (which has $\eta_1$ bindings for $\Gamma$ and then the empty prefix for everything in $\Gamma'$) and then
run the resulting term on $\eta_2 \cup \emp{\Gamma}$ (with an empty prefixes for
$\Gamma$).  The other does the opposite, first running $e$ on
$\eta_2 \cup \emp{\Gamma}$ and then running the resulting term
on $\eta_1 \cup \emp{\Gamma'}$.  Determinism says that these strategies produce equal
results. It is proved in Appendix~\ref{app:determinism} by
observing that the homomorphism theorem guarantees that each of these
options is equivalent to running $e$ on $\eta$.

\begin{theorem}[Determinism]
  Suppose
    (1) $\commactx{\Gamma}{\Gamma'} \vdash e : s$,
    (2) $\prefixHasType{\eta}{\commactx{\Gamma}{\Gamma'}}$,
    (3) $\prefixstep{\eta|_{\Gamma} \cup \emp{\Gamma'}}{e}{
    }{e_1}{p_1}$ and $\prefixstep{\eta|_{\Gamma'} \cup
      \emp{\Gamma}}{e_1}{ }{e_2}{p_2}$,
    (4) $\prefixstep{\eta|_{\Gamma'} \cup \emp{\Gamma}}{e}{ }{e_1'}{p_1'}$ and $\prefixstep{\eta|_{\Gamma} \cup \emp{\Gamma'}}{e_1'}{ }{e_2'}{p_2'}$.
  Then $e_2 = e_2'$ and $\prefixConcat{p_1}{p_2} = \prefixConcat{p_1'}{p_2'}$.
\end{theorem}

To intuitively see how this theorem follows from homomorphism, note that prefixes are
canonical representatives of equivalence classes of sequences of stream elements, up to the possible
reorderings defined by their type \cite{stanford2022safe}. The homomorphism
theorem then guarantees that these normal forms are processed
compositionally, and so are
independent of the actual temporal ordering of parallel data---it suffices to
compute on the combined normal forms from the two steps.

\section{Full \core{}}
\label{sec:full-core-lang}
We now sketch the remaining types and terms of
\core{} that are not part of Kernel \core{}.

\subsection{Sums}
Sum types in \core{}, written $s + t$,
are {\em tagged} unions: a stream of type $s + t$ is either a stream of type $s$ or a stream of type $t$,
and a consumer can tell which. Streams of type $s$ are not the same as
streams of type $s + s$,
and streams of type $s + t$ are isomorphic to, but not identical to, streams of type $t + s$.
Operationally, a producer of a sum stream sends a tag bit before
sending the rest of the stream, to tell downstream consumers which
side to expect. Conversely,
a consumer of $s + t$ first reads the bit to learn which it is getting
next.

%
A prefix of $s + t$ can be a prefix of one of $s$ or one of $t$, written $\sumpA{p}$
or $\sumpB{p}$, or it can be $\sumpEmp$,
the empty prefix of type $s + t$, which does not even include the initial tag bit.
The derivatives with respect to these prefixes are defined as follows: (a) the empty prefix takes nothing off the type ($\deriv{\sumpEmp}{s + t} = s + t$)
and (b) the two injections reduce to taking the derivative of the corresponding branch of the sum ($\deriv{\sumpA{p}}{s + t} = \deriv{p}{s}$ and $\deriv{\sumpB{p}}{s + t} = \deriv{p}{t}$).



\begin{wrapfigure}{R}{11.2cm}
\vspace*{-3ex}
$
  \INFERRULE[T-Sum-R-1]{
    \Gamma \vdash e : s
  }{
    \Gamma \vdash \sumInlTm{e} : s + t
  }
  \hspace{.4in}
  \INFERRULE[T-Sum-L-Surf]{
    \Gamma(x : s) \vdash e_1 : r\\
    \Gamma(y : t) \vdash e_2 : r\\
  }{
    \Gamma(z : s + t) \vdash \texttt{case}_{r}(z,x.e_1,y.e_2) : r
  }
$
\vspace*{-3ex}
\end{wrapfigure}
The typing rules for sums are the
normal injections on the right (\ruleName{T-Sum-R-1} and a
symmetric rule \ruleName{T-Sum-R-2})
and a case analysis rule on the left (\ruleName{T-Sum-L-Surf}).
The right rules operate by prepending their respective tags and then
running the embedded terms. The left rule does
case analysis: if the incoming stream $z$ comes from the left
of the sum, it is processed with
$e_1$; if from the right, $e_2$.
To run a sum case term, the semantics must dispatch
on the tag that says if the stream $z$ being destructed is a left or a right. But the prefix $z$ might not include a
tag, if only data from the surrounding context has arrived. In this case, $z$ will map to $\sumpEmp$, and we have
no way of determining which branch to run. The solution is to
run neither!
Instead, we hold on to the environment, saving \emph{all} incoming data to the program
until the tag arrives. Once we get a prefix that includes the tag, we continue by running the corresponding branch
with the accumulated inputs.
Note that this buffering is necessarily
a blocking operation.%
\jwc{Fix wording here --- we don't define blocking really.}
\footnote{Depending on the rest of the context, it could also require unbounded memory! Fortunately, we can easily detect this, and flag
it as a warning to the user: running a case on $z : s + t$ in a context $\Gamma(z : s + t)$ could require buffering all variables to the left of $z$ or in parallel with $z$ in the context.
Unbounded memory is required if and only if any of those variables have star type.
}

\begin{wrapfigure}{R}{7.1cm}
\vspace*{-3ex}
$
 \INFERRULE[T-Sum-L]{
    \envHasType{\eta}{\Gamma(z : s + t)}\\
    \Gamma(x : s) \vdash e_1 : r\\
    \Gamma(y : t) \vdash e_2 : r
  }{
    \deriv{\eta}{\Gamma(z : s + t)} \vdash {\sumcaseTm{r}{\eta}{z}{x}{e_1}{y}{e_2}} : r
  }
$
\vspace*{-3ex}
\end{wrapfigure}
All this requires a slightly generalized typing
rule (\ruleName{T-Sum-L}) that includes a {\em buffer environment} $\envHasType{\eta}{\Gamma(z : s + t)}$ of the context
type in the term. This buffer holds all of the input data we've seen so far. As prefixes arrive, we append to this buffer until we
get the tag. Accordingly, the context in this rule is $\deriv{\eta}{\Gamma(z : s + t)}$: the term is typed in the context consisting of
everything after the part of the stream that has so far been buffered.

Fortunately, the only typing rule that a \core{} programmer needs to concern themselves with is \ruleName{T-Sum-L-Surf}.
While writing the program, and before it runs, the buffer is empty ($\eta = \emp{\Gamma(z : s + t)}$).
In this case, the $\deriv{\eta}{\Gamma(z : s + t)} = \Gamma(z : s + t)$, and so the generalized rule
\ruleName{T-Sum-L} simplifies to the ``surface'' rule, \ruleName{T-Sum-L-Surf}.
Full details 
can be found in Appendix~\ref{app:semantics}.

\subsection{Star}
\label{sec:star}
Full \core{} also includes a type constructor for unbounded streams,
written $s^\star$ because it is inspired by the Kleene star from the
theory of regular languages.  (We do not need to distinguish between unbounded
finite streams and ``truly infinite'' ones, because our
operational semantics is based on prefixes: we're always only
operating on ``the first part'' of the input stream, and it doesn't
matter whether the part we haven't seen yet is finite or infinite.)
The type $s^\star$ describes a stream that consists of zero or more
sub-streams of type $s$, in sequence.

In ordinary regular languages, $r^\star$ is equal to $\varepsilon + r \cdot
r^\star$.  In the language of stream types, this equation says that
a stream of type $s^\star$ is either empty ($\varepsilon$) or a stream of type
$s$ followed by another stream of type $s^\star$---i.e., $s^\star$
can be understood as the least fixpoint of the stream type operator $x
\mapsto \varepsilon + s \cdot x$.
The definitions of prefixes and typing rules for star all
follow from this perspective.

In particular,
$\texttt{prefix}(s^\star) = \texttt{prefix}(\varepsilon + s \cdot
s^\star)$.
The empty prefix of
type $s^\star$, written $\stpEmp$, is effectively the empty prefix of the sum
that makes up $s^\star$. The second form of prefix---the ``done''
prefix of type $s^\star$---is written $\stpDone$.  It corresponds to the left injection of the
sum, and receiving it means that the stream has ended.  Note that, despite
containing no $s$ data, this prefix is not \emph{empty}: it conveys the information that
the stream is complete. The final two cases correspond to the right injection of
the sum, i.e., a prefix of type $s \cdot s^\star$.  This is either $\stpA{p}$, with $p$
a prefix of $s$, or $\stpB{p}{p'}$, with $p$ a maximal prefix of type $s$ and $p'$ another
prefix of $s^\star$.

For derivatives, the empty prefix leaves the type as-is
($\deriv{\stpEmp}{s^\star} = s^\star$).  Because no data will arrive after the
done prefix, the derivative of $s^\star$ with respect to $\stpDone$ is
$\varepsilon$. In the case for $\stpA{p}$, after
some of an $s$ has been received, the remainder of $s^\star$ looks like the
remainder of the first $s$ followed by some more $s^\star$, so the derivative is
defined as $\deriv{\stpA{p}}{s^\star} = \left(\deriv{p}{s}\right) \cdot s^\star$.
Finally, $\deriv{\stpB{p}{p'}}{s^\star} = \deriv{p'}{s^\star}$.

%

\begin{wrapfigure}{R}{8.6cm}
\vspace*{-3ex}
$
  \INFERRULE[T-Star-R-1]{ }{
    \Gamma \vdash \nilTm : s^\star
  }
$
\hspace*{2em}
$
  \INFERRULE[T-Star-R-2]{
    \Gamma \vdash e_1 : s\\
    \Delta \vdash e_2 : s^\star\\
  }{
    \Gamma ; \Delta \vdash \consTm{e_1}{e_2} : s^\star
  }
$
\vspace*{-2ex}
\end{wrapfigure}
The typing rules for star are again motivated by the analogy with
lists. There are right rules for \texttt{nil} and
\texttt{cons} and a case analysis principle for the left rule.
The ``nil'' rule \ruleName{T-Star-R-1} corresponds to the left injection into the
sum $s^\star = \varepsilon + s \cdot s^\star$:
from any context, we can produce $s^\star$ by simply ending the
stream.
The ``cons''
rule \ruleName{T-Star-R-2} is the right injection: from a context
$\semicctxcompact{\Gamma}{\Delta}$, we can produce an $s^\star$ by
producing one $s$ from $\Gamma$ and the remaining $s^\star$ from $\Delta$.
Operationally,
this should run the same way as the \ruleName{T-Cat-R} rule: by first running $e_1$, and if an entire $s$
is produced, continuing by running $e_2$ to produce some prefix of the tail.
\begin{wrapfigure}{R}{7.2cm}
\vspace*{-3ex}
$
  \INFERRULE[T-Star-L]{
    \Gamma(\cdot) \vdash e_1 : r\\
    \Gamma(x : s ; xs : s^\star) \vdash e_2 : r\\
    \envHasType{\eta}{\Gamma(z : s^\star)}\\
  }{
    \deriv{\eta}{\Gamma(z : s^\star)} \vdash
       \starcaseTm{s,r}{p}{z}{e_1}{x}{xs}{e_2} : r
  }
$
\vspace*{-3ex}
\end{wrapfigure}
The \ruleName{T-Star-L} rule is a case analysis principle for
streams of star type: either such a stream is empty, or else
it comprises one $s$ followed by an $s^\star$. The fact that the head
$s$ will come first and the tail $s^\star$ later tells us that
the variables $x : s$ and $xs : s^\star$  should be separated by a
semicolon in the context. Like \ruleName{T-Sum-L}, this rule
includes a buffer, collecting input environments until the prefix bound to $z$
is enough to make the decision for which branch of the case to run.

The semantics of the right rules are straightforward: the rules for \ruleName{T-Star-R-1} are like those for \ruleName{T-Eps-R},
while the rules for \ruleName{T-Star-R-2} are like those for \ruleName{T-Cat-R}.
The semantics of \ruleName{T-Star-L} is just like \ruleName{T-Sum-L}, buffering input prefixes until
either (a) we get $z \mapsto \stpDone$, at which point we run $e_1$, or (b) we get $z \mapsto \stpA{p}$ or $z \mapsto \stpB{p}{p'}$,
in which case we run $e_2$.
For full details, see Appendix~\ref{app:semantics}.

\subsection{Let-Binding}
\label{sec:let-binding}
Full \core{} also allows for more general let-binding.
Given a transformer $e$ whose output is used in the input of another
term $e'$, we can compose
them to form a single term $\cutTm{x}{e}{e'}$ that operates as
the sequential composition of $e$ followed by $e'$. The rules for this construct are in Figure~\ref{fig:let-rules}.
Note that this sequencing is not the same kind of sequencing as in a concat-pair $\catpairTm{e}{e'}$.
The latter produces data that follows the sequential pattern $s \cdot t$, while the former is sequential composition of code.
When a let binding is run, both terms are evaluated, and the output of the first is passed to the input of the second.
An important point to note is that this semantics is non-blocking:
even if $e$ produces the empty prefix, we still run $e'$, potentially producing output.

\begin{wrapfigure}{R}{8cm}
\vspace*{-3px}
$
      \INFERRULE[T-Let]{
        \Delta \vdash e : s\\
        \Gamma(x : s) \vdash e' : t\\
        \inert{e}
      }{
        \Gamma(\Delta) \vdash \cutTm{x}{e}{e'} : t
      }
$
$
      \INFERRULE[S-Let]{
          \prefixstep{\eta}{e_1}{ }{e_1'}{p}\\
          \prefixstep{\eta[x\mapsto p]}{e_2}{ }{e_2'}{p'}
        }{
          \prefixstep{\eta}{\cutTm{x}{e_1}{e_2}}{ }{\cutTm{x}{e_1'}{e_2'}}{p'}
        }
$

  \caption{Rules for Let-Bindings}
  \label{fig:let-rules}
\vspace*{-10px}
\end{wrapfigure}

The semantic rule \ruleName{S-Let} for let-binding (in Figure~\ref{fig:let-rules}) is a straightforward encoding of this behavior.
Given the input environment $\eta$,
we run the term $e$, bind the resulting prefix $p$ to $x$, and run the continuation $e'$, returning its output.
The resultant term is another let-binding between the resultant terms of $e$ and $e'$.

The typing rule \ruleName{T-Let} says that if $e$ has type $s$ in context $\Delta$ and $e'$ has type $t$ in a
context $\Gamma(x : s)$ with a variable of type $s$, we can form the let-binding
term $\cutTm{x}{e}{e'}$, which has type $t$ in context
$\Gamma(\Delta)$. 
The soundness of the semantics rule \ruleName{S-Let} depends on a subtle requirement:
$e$ must not produce nonempty output until $e'$ is ready to accept it.
This is enforced by the third premise of the \ruleName{T-Let} rule, which
states that $e$ must be \emph{inert}: it only produces nonempty output when
given nonempty input.
This restriction rules out let-bindings such as $\cutTm{x}{\oneTm}{e'}$,
since the semantics of $\oneTm$ always produces nonempty output (namely $\onepB$), even when given
an environment mapping every variable to an empty prefix\footnote{Because such let-bindings are essentially
trivial, we expect that they can be eliminated --- see Section~\ref{sec:future-work} for more discussion.}.
In actuality, inertness is not a purely syntactic condition on terms, but depends also
on typing information. To this end, inertness is tracked like an effect through the type system:
see Appendix~\ref{app:inertness} for details.

\subsection{Recursion}
\label{sec:recursion}
To write interesting transformers over $s^\star$ streams,
we provide a way to define transformers recursively.
Adding a traditional general recursion operator $\texttt{fix}(x.e)$ does not work in our context,
as arrow types are required to define functions this way. We instead
add explicit term-level recursion
and recursive call operators. The program $\fixTm{e_\texttt{args}}{e}$ defines
a recursive
transformer with body $e$ and initial arguments $e_\texttt{args}$. Recursive calls are made inside the body $e$ with a term
$\texttt{rec}\left(e_\texttt{args}\right)$, which calls the function being defined with arguments $e_\texttt{args}$.
This back-reference works in the same way that uses of the variable $x$ in the body of a traditional fix point $\texttt{fix}(x.e)$ refer to the term $\texttt{fix}(x.e)$ itself.
This function-free approach is
approach is inspired by the concept of {\em
cyclic proofs}~\cite{brotherston_cyclic_2005, fortier_cuts_2013, derakhshan_session-typed_2021} from proof theory, where derivations may refer back to themselves.
Alternatively, one can think of this construction as defining our terms and proof trees
as infinite {\em coinductive} trees; then the term-level $\texttt{fix}$
operator defines terms as {\em co}fixpoints.

Full details of the typing rules and semantics of fixpoints can be found in
Appendices~\ref{app:type-system} and~\ref{app:semantics}.
In brief, to typecheck a fixpoint term, we simply type its body $e$, assuming that all
instances of the \texttt{rec} in $e$ have the same type as the fixpoint itself.
Then, to run a fixpoint term $\fixTm{e_\texttt{args}}{e}$, the rule unfolds the recursion one
step by substituting the body $e$ for instances of \texttt{rec} in itself,
then runs the resulting term, binding all of the arguments to their variables.

Naturally, this can lead to non-termination, as $\texttt{fix}(\texttt{rec})$ unfolds to itself.%
\footnote{
Cyclic proof systems
usually ensure soundness by imposing a guardedness condition
\cite{brotherston_cyclic_2005} which
requires certain rules be applied before a back-edge can be inserted in the
derivation tree. Because we are not primarily concerned with \core{} as a logic
at the moment, we leave a guardedness condition to future work.}
To bound the depth of evaluation, we
{\em step index} both semantic judgments by adding a fuel parameter
that decreases when we unfold a \texttt{fix}.
The semantic judgment then looks like
$\prefixstep{\eta}{e}{n}{e'}{p}$:
when we run $e$ on $\eta$, it steps to $e'$ producing $p$
and unfolding at most $n$ uses of \texttt{fix} along the way.

\subsection{Stateful Transformers}
\label{sec:wait}
In the \core{} typing judgment $\Gamma \vdash e : s$,
the variables in $\Gamma$ range over \emph{future values} that have yet to
arrive at the transformer $e$.
The ordered nature of semicolon contexts means
that variables further to the right
in $\Gamma$ correspond to data that will arrive further in the
future. This imposes a strong restriction on
programming: if earlier values in the stream are used at all, they must be
used {\em before} later values; once a value in the stream has ``gone
by,'' there is no way to refer to it again.
By using variables from the $\Gamma$
context, a term $e$ can refer to values that will arrive in the future; but it has no way
of referring to values that have arrived in the \emph{past}.
This limitation is by design:
from a programming perspective, referring to variables from the past
requires {\em memory}, which is a resource to be carefully managed in streaming contexts.
Of course, while some important streaming functions
(e.g., map and filter) can get by without state, but many others (e.g., ``running
sums'') require it. In this section, we add support for stateful
stream transformers.

To maintain state from the past,
we extend the typing judgment of
\core{} to include a second context, $\Omega$, called the
\emph{historical context}, which gives types to variables bound to values stored in memory. We write
$\Omega \mid \Gamma \vdash e : s$ to mean ``$e$ has type $s$ in context $\Gamma$ and historical context $\Omega$''.

What types do variables in the historical context have?
Once a complete stream
of type $\left(\texttt{Int}^\star \| \texttt{Int}^\star\right) \cdot
\texttt{Int}^\star$ has been received and is stored in memory, we may as well regard the data
as a value of the standard type
$\left(\texttt{list(Int)} \times \texttt{list(Int)}\right) \times
\texttt{list(Int)}$ from the simply typed lambda-calculus (STLC).
In other words, parts of streams that {\em will} arrive in the future have
stream types, parts of streams that {\em have} arrived in the past
can be given standard STLC types.
The ``flattening'' operation $\flatten{s}$
transforms stream types into STLC types. The interesting cases of its
definition are
$\flatten{s \cdot t} = \flatten{s \| t} = \flatten{s} \times \flatten{t}$ and
$\flatten{s^\star} = \texttt{list}\left(\flatten{s}\right)$.

The historical context is a fully structural: $\Omega \, ::= \, \cdot \,
|\, \Omega, x : A$, where the types $A$ are drawn from some set of
conventional lambda-calculus types
including at least products, sums, a unit, and a list type.
Operationally, the historical context behaves like a standard context
in a functional programming language: at the top level, terms
to be run must be typed in an empty
historical context; at runtime, historical variables get their values
by substitution.

\begin{wrapfigure}{R}{4.1cm}
\vspace*{-3ex}
$
  \INFERRULE[T-HistPgm]{
    \Omega \vdash M : \flatten{s}
  }{
    \Omega \mid \Gamma \vdash \histPgmTm{M}{s} : s
  }
$
\vspace*{-3ex}
\end{wrapfigure}
Rather than giving a specific set of ad-hoc rules for manipulating values from
the historical context, we parameterize the \core{} calculus over an arbitrary
language with terms $M$, typing judgment $\Omega \vdash M : A$, and
big-step semantics $M \downarrow v$. We call any such fixed choice of language
the \emph{history language}.  Programs from the history language can be embedded
in \core{} programs using the \ruleName{T-HistPgm} rule, which says
that a historical program $M$ of type $\Omega \vdash M :
\flatten{s}$ with access the historical context can be used in place
of a \core{} term of type $s$.
Operationally, as soon as any prefix of the input arrives, we
run the historical program to completion and yield the result
as its stream output (after converting it into a value of type $s$).

How does information get added to the historical context?
Intuitively, a variable in $\Gamma$ (a stream that will arrive in the
future) can be moved to $\Omega$, where streams that have arrived in the past are
saved, by waiting for the future to become the past!
Formally, we define an operation called ``wait,'' which allows the
programmer to specify
part of the incoming context and block this subcomputation until
that part of the input stream has
arrived in full. Once it has,
we can bind it to the variables in the historical context and continue
by running $e$.

\begin{wrapfigure}{R}{6cm}
\vspace*{-3ex}
$
  \INFERRULE[T-Wait-Surf]{
    \Omega, x : \flatten{s} \mid \Gamma(\cdot) \vdash_\Sigma e : s
  }{
    \Omega \mid \Gamma(x:s) \vdash_\Sigma \waitTmNOBUFFER{s}{x}{e} : s
  }
$
\vspace*{-3ex}
\end{wrapfigure}
The \ruleName{T-Wait-Surf} rule
encodes the typing content of this behavior. It allows us to
specify a variable $x$ of the input, flatten its type, and then move it to the historical context,
so that the continuation $e$ can refer to it in historical terms.
Semantically, this works by
buffering in environments until a maximal prefix for $x$ has arrived.
Once we have a full prefix for $x$, we substitute it into $e$ and continue running
the resulting term.\footnote{%
 The semantics of the \ruleName{T-Wait} rule is reminiscent of the ``blocking
 reads'' of Kahn Process Networks, where every read from a parallel stream
 blocks all other reads to ensure determinism. Here, we choose a variable and
 block the rest of the program until it is complete and in memory.
}
This buffering is implemented the same way as in the left rules for plus and star, by generalizing
the typing rule \ruleName{T-Wait-Surf} to a rule \ruleName{T-Wait} which includes an explicit prefix buffer.
As with plus and star, the generalized rule simplifies to the surface rule when the buffer is empty.
The generalized rule and the semantics of both the wait and historical program constructs
can be found in Appendix~\ref{app:type-system} and Appendix~\ref{app:semantics}. The remaining typing rules in \core{} change only by adding an
$\Omega$ to the typing judgment everywhere.




\subsubsection*{Updated Soundness Theorems}
Adding recursion and the historical context requires us to update
to the soundness theorem from that of Kernel \core{} to Full \core{}. If a well typed term has (a) closed historical context, and (b) no unbound recursive calls,  takes a step on
a well-typed input {\em using some amount of gas}, then the output and resulting
term are also well typed.
(The proof is by a large but routine induction, first on the derivation of $\prefixstep{\eta}{e}{ }{e'}{p}$,
and then on the derivation of $\cdot \mid \Gamma \vdash_\emptyset e : s$. See Appendix~\ref{app:incr-sem-thms} for cases.)

\begin{theorem}[Soundness of the \core{} Semantics]
  If $\cdot \mid \Gamma \vdash_\emptyset e : s$, and $\envHasType{\eta}{\Gamma}$,
  and $\prefixstep{\eta}{e}{n}{e'}{p}$, then $\prefixHasType{p}{s}$ and $\cdot \mid \deriv{\eta}{\Gamma} \vdash_\emptyset e' : \deriv{p}{s}$
\end{theorem}

A similarly updated statement of the homomorphism
theorem can be found in Appendix~\ref{app:semantics}.

\section{\lang{}}
\label{sec:examples}

\bcp{General comment: I still find ``delta'' as a proper name with a
  mandatory lc first letter to be quite awkward to read.  Even after
  all this time, my autocorrector flags every single instance.  Would
  it really be so bad to capitalize it? :-)}

We next show how \core{} addresses the problems that we identified in Section~\ref{sec:extended-motivation} of
(a) type-safe programming with
temporal patterns and (b) deterministic processing of parallel data. We also show how some other
characteristic streaming idioms can be expressed elegantly in \core{}.

The examples in this section are written in \lang{}\footnote{\lang{} is available at \url{http://www.github.com/alpha-convert/delta}},
an experimental language design based on \core{}.\lang{}
proposes a high-level functional syntax that, after typechecking,
elaborates to \core{} terms. It supports some features not included
in the \core{} calculus that we expect will be required in full-blown
language designs based on \core{}.

\subsection*{\lang{} Features}
While the proof terms of \core{} allow elimination forms (such as \texttt{let (x,y) = z in e}) to only be applied to variables (an artifact of the sequent calculus formalism),
\lang{}'s syntax is a standard one where elimination forms can be applied
to arbitrary expressions. \lang{} also includes more types than \core{}, adding base types \texttt{Int} and \texttt{Bool}.

\subsubsection*{Functions and Macros}
Top-level functions in \lang{} are simply open terms: a function definition \texttt{fun f(x : Int*) : Int* = e}
elaborates and typechecks to a core term $e$ which satisfies the typing judgment $x : \texttt{Int}^\star \vdash \texttt{e} : \texttt{Int}^\star$.
Higher order functions in \lang{} are implemented as \emph{macros}. A function written as \texttt{fun g<f : Int -> Int>(x : Int*) : Int* = e}
is a macro which takes another function \texttt{f : Int* -> Int*} as a
parameter. Calls to \texttt{g} in other functions then look like \texttt{g<f'>}, where \texttt{f'} is either (a)
another function defined at top level, or (b) a call to yet another macro.
If the macro \texttt{g} is recursive, its recursive calls do not receive a macro argument --- all recursive usages
of a macro get passed the initial macro parameter \texttt{f}.
This discipline ensures that the macro usage does not depend on runtime data, and so higher-order functions can be fully resolved to \core{} terms statically.

Neither of these features --- standard top-level functions and higher-order
macros --- require the use of first-class function types, which \core{} does not
currently support. Defining true higher-order functions would allow for streams of \emph{functions}, such as $\left(s \to t\right)^\star$. We hope to investigate these in future work; see
Section~\ref{sec:future-work}.

Functions in \lang{} can also be (prenex-) polymorphic \cite{milner78}.
Polymorphic functions definitions are annotated with an list of their type arguments, like \texttt{fun f[s,t](x : s*) : t* = e}. When such a
function is called, the type arguments must be passed explicitly
like \texttt{f[Int,Bool]}.

\subsubsection*{Historical Arguments and Generalized Wait}
Functions in \lang{} can also take arguments for their historical contexts: a function \texttt{fun f\{acc : Int\}(xs : Bool*) : Int* = e}
takes an in-memory \texttt{Int} argument, and elaborates to a core term that satisfies the typing judgment
$\texttt{acc} : \texttt{Int} \mid \texttt{xs} : \texttt{Bool}^\star \vdash e : \texttt{Int}^\star$, where \texttt{acc}
is in the historical context.
When \texttt{f} is called, the \texttt{acc} argument must be passed a historical program. For example, if \texttt{u : Int} is in the current historical
context (and \texttt{ys : Bool*} in the regular one), \texttt{f\{u + 1\}(ys)} is an acceptable call to \texttt{f}.

The \texttt{wait} construct is also slightly more general in \lang{}. Instead of just \texttt{wait}ing
on variables, programmers may \texttt{wait} on the result of some expression,
and then save its result into memory: this is accomplished with \texttt{wait e as x do e' end}.

\subsubsection*{\lang{} Implementation}
The implementation begins by elaborating a high-level surface syntax down to an
``elaborated syntax'', which eliminates shadowing, resolves function calls, and
transforms the syntax into the sequent calculus representation by introducing
intermediate variables for subexpressions.

The elaborated syntax is then typechecked, producing \emph{templates} of \core{} terms.
These templates serve two roles. First, they are monomorphizers. Since \core{} is a monomorphic calculus, typechecking produces a map from closed types (to plug in for type variables)
to raw terms. Second, the templates implement macro expansion.

The typechecker uses a (we believe novel)
algorithm for checking our variant of ordered \& bunched terms. While we have tested the typechecker
with many terms, we have not proved that the algorithm is sound and complete with respect to the declarative
type system presented in Appendix~\ref{app:technical}. The interpreter, on the other hand, is very straightforward: it is a direct translation of the \core{} semantics into code.

More details about the project structure of the \lang{} prototype can be found in Appendix~\ref{app:artifact}.

\subsection*{Examples}
Besides its type system, \lang's design differs from that of most
stream processing languages in another important respect.
In languages like Flink \cite{Flink},
Beam\cite{Beam}, and Spark \cite{SparkStreaming}, streaming programs must be
written using
a handful of provided combinators like \texttt{map} \texttt{filter}
and \texttt{fold} (or
possibly as SQL-style queries, in languages derived from CQL
\cite{arasu2003cql}).
By contrast, \lang{} programs are written in the style of functional
list processors.  Instead of
working to cram complex program behaviors into \texttt{map}s,
\texttt{filter}s, and \texttt{fold}s, programmers
can express their intent more directly in the form of more general recursive
functional programs. Of
course, this does not preclude the use of the aforementioned
combinators: they are directly implementable in \lang{}.






\subsubsection*{Map}
\begin{wrapfigure}{R}{5.9cm}
\begin{lstlisting}[style=wrapfig, aboveskip=-3.3ex, belowskip=-2ex]
fun map [s,t] <f : s -> t> (xs : s*) : t* =
  case xs of
    nil => nil
  | y :: ys => f(y) :: map(ys)
\end{lstlisting}
\end{wrapfigure}

Given a transformer from \texttt{s} to \texttt{t}, we can lift it to a transformer from \texttt{s$\superstar$}
to \texttt{t$\superstar$} with a \texttt{map} operation. The code for this function
is essentially identical to the familiar functional
program, but
its type is more general than the standard \texttt{map} function
on homogeneous streams, which has type $(\texttt{a} \to \texttt{b}) \to \left(\texttt{Stream a} \to \texttt{Stream b}\right)$:
the types $s$ and $t$ here can be arbitrary stream types: they need not be
singletons.

\begin{wrapfigure}{R}{7.5cm}
\begin{lstlisting}[style=wrapfig, aboveskip=-3.3ex, belowskip=-1ex]
fun mapMaybe[s,t]<f : s -> t + Eps> (xs : s*) : t*=
  case xs of
    nil => nil
  | y :: ys => case f y of
               | inl(t) => t :: mapMaybe(ys)
               | inr(_) => mapMaybe(ys)

fun liftP[s]<f : {s}(Eps) -> Bool>(x : s) : s + Eps =
  wait x, f{x}(sink) as b do
      if {b} then inl({x}) else inr(sink)
  end

fun filter<f : {s}(Eps) -> Bool>(xs : s*) : s* =
  mapMaybe[s,s]<liftPred<f>>(xs)
\end{lstlisting}
\end{wrapfigure}
\subsubsection*{Filter}
\label{sec:filter}

Similarly, given a ``predicate'' function $f$ from $s$ to $t + \varepsilon$ (the streaming version of $t \; \texttt{option}$),
we can transform an incoming stream of $s^\star$ to include just the transformed elements which pass the filter.

We can then recover a traditional predicate-based filter by lifting a predicate \texttt{f} that takes an in-memory \texttt{s} to \texttt{Bool}
to a streaming function \texttt{s -> s + Eps} with \texttt{liftP}. This program simply waits for its argument to arrive, then
applies the predicate to the in-memory \texttt{s}.





\begin{wrapfigure}{R}{8cm}
\begin{lstlisting}[style=wrapfig, aboveskip=-3ex, belowskip=-2ex]
fun fold [s,t] <f : {t}(s) -> t>{acc : t}(xs : s*) : t =
  case xs of
    nil => {acc}
  | y :: ys => wait f{acc}(y) as acc' do
                  fold{acc'}(ys)
               end
\end{lstlisting}
\end{wrapfigure}
\subsubsection*{Fold}
\core{} can express both {\em running folds}, which output a
stream of all their intermediate
states, and {\em functional folds}, which output only the final state.
Since functional folds that return only the final state
cannot be given this rich type in traditional stream processing languages (for the same reason as the \texttt{head} and \texttt{tail} functions),
we present one here. See Appendix~\ref{app:running-fold} for
discussion of a running fold; the code is similar except that it
outputs \texttt{y} at every step. \bcp{If we have a bit of space left,
I'd advocate using it to show the running fold.}

The \texttt{fold} transducer maintains an in-memory accumulator of type $t$; this gets
updated by a streaming
step function $\texttt{f : \{t\}(s) -> t}$ that takes the state $t$ and
the new element $s$ and produces a $t$. The whole fold takes a stream
\texttt{xs} of type $s^\star$ and an initial accumulator value $y :
t$, and it eventually
produces the final state $t$.  As for \texttt{map}, the code for \texttt{fold} is
very similar to the traditional functional program: the only distinction is
the inclusion of \texttt{wait}s to marshal data into memory.

\begin{wrapfigure}{R}{5.5cm}
\begin{lstlisting}[style=wrapfig, aboveskip=-3ex, belowskip=-3ex]
fun head [s] (xs : s*) : Eps + s =
  case xs of nil => inl(sink)
           | y :: _ => inr(y)
\end{lstlisting}
\end{wrapfigure}
\subsubsection*{Singletons, Head, Tail}
In the homogeneous model, stream types are always conceptually unbounded.
But in many practical situations, a stream will only be expected to contain
a single element---a constraint that cannot be expressed with homogeneous
streams. Using stream types, we can write stream transformers that
are statically known to only produce a single output. For example, the
``head'' function
is trivially expressible in the same manner as head on lists, as shown
on the right.

(Exercise:
try writing the term for \texttt{tail} on star streams.  This
 requires a use of \texttt{wait} and an accumulator argument like in
 \texttt{fold}. \jwc{Dependent on the version of fold we use})

\begin{wrapfigure}{R}{6.3cm}
\begin{lstlisting}[style=wrapfig, aboveskip=-3ex, belowskip=-3ex]

fun averageSingle (run : Int . Int*) : Int =
  let (x;xs) = run in
  let (sm,len) = (sum(xs), length(xs)) in
  wait x,sm,len do
    {(x + sm) / (1 + len)}
  end

fun averageAbove{t : Int}(xs : Int*) : Int* =
  map<averageSingle>(thresh{t}(xs))
\end{lstlisting}
\end{wrapfigure}

\subsubsection*{Brightness Levels}

The structured communication protocol from the brightness-levels example in Section~\ref{sec:extended-motivation}
can be encoded as the type $\left(\texttt{Int} \cdot \texttt{Int}^\star\right)^\star$:
a stream of nonempty sequences of \texttt{Int}s, representing ``runs'' of
light levels greater than some threshold.
Given such a stream, writing a program to compute the averages is easy:
we just \texttt{map} an average operation---taking $\texttt{Int} \cdot
\texttt{Int}^\star$ to \texttt{Int}---over the incoming stream to produce a
stream $\texttt{Int}^\star$ of averages. The per-run average
operation, \texttt{averageSingle}, is defined by computing its sum and length
in parallel, waiting for the results, then dividing the sum (plus the
first element) by the length (plus one).

Notice that, since each run is statically known to have at
least one element, \texttt{averageSingle} can omit the error handling that, with a
homogeneous stream type, would be needed to avoid a potential divide by zero.
By contrast, with a homogeneous stream type like $\left(\texttt{Start} + \texttt{Int} +
\texttt{End}\right)^\star$, this operation would need to be written in a
low-level, more stateful manner, remembering the current run of
\texttt{Int}s
until an \texttt{End} event arrives, averaging, and handling the divide-by-zero
error which could in principle (although not in practice) occur if no \texttt{Int}s
arrived between a \texttt{Start} and an \texttt{End}.

\jwc{Edit this patter.}

\begin{wrapfigure}{R}{8cm}
\begin{lstlisting}[style=wrapfig, aboveskip=-3ex, belowskip=-2ex]
fun thresh{t : Int}(xs : Int*) : (Int . Int*)* =
  case xs of
    nil => nil
  | y :: ys => wait y do
                 if {y > t} then
                   let (run;rest) = spanGt{t}(ys) in
                   ({y};run) :: thresh{t}(rest)
                 else
                   thresh{t}(ys)
               end
\end{lstlisting}
\end{wrapfigure}
The thresholding operation \texttt{thresh}, which
takes $\texttt{Int}^\star$ and
produces the runs of elements above the threshold $\left(\texttt{Int} \cdot
\texttt{Int}^\star\right)^\star$, is straightforward.
Whenever the incoming stream goes above the threshold \texttt{t}, we collect all
of the subsequent elements into a run, emit it, and recurse down the rest of the stream.
This uses an operation \texttt{spanGt : \{Int\} (Int*) -> Int . Int*} that returns the initial ``span'' of
elements above \texttt{t}, followed by the rest of the stream.
It's important to note that this program is completely
non-blocking: as soon as the first element above \texttt{t} arrives, it is
forwarded along, as are all subsequent elements until the stream drops below \texttt{t}.
By contrast with homogeneously typed streaming languages, \lang{}'s
type safety guarantees
that \texttt{thresh} {\em does in fact} output a stream that adheres to the protocol, and (2)
any downstream transformer does not have to replicate this parsing logic.

The complete program, first calling \texttt{thresh}, and then mapping \texttt{averageSingle} over the stream of runs,
is \texttt{averageAbove}.

\subsubsection*{Side Outputs \& Error Handling}
A common streaming idiom is the use of ``side outputs'' for reporting errors.
In languages that support this idiom, operations include extra output
streams where error messages are sent as they
arise at runtime.  These side outputs are always a second-class
mechanism: the error streams cannot be transformed or used in a manner other
than dumping them to a log somewhere.  \core{} provides a first-class account of
side outputs, encoding them as a parallel output type. A function $s \to t$
that may produce errors of type $e$ can have type $s \to t \| e^\star$.
Alternately, errors can be handled inline in the traditional functional
way, using a sum type $s + e$.

\subsubsection*{Partitioning and Merging}
Partitioning is a crucial streaming idiom where a single stream of
data is split into two or more parallel streams to be
routed to different downstream processing nodes, thus
exposing parallelism and increasing potential throughput.
Appendix~\ref{app:partitioning} shows how two different partitioning
strategies can be implemented in \core{}.
First is a {\em round-robin partitioner},
which fairly partitions an incoming stream of type
$s^\star$ into a parallel pair of streams $s^\star \| s^\star$
by sending the first element to the left branch, the second to the
right, and so on.

Second is a {\em decision-based partitioner}, which routes a stream of type $s^\star$
one direction or another into an output stream of type $t^\star \| r^\star$ based on the result
of a function from $s$ to $t + r$.

\jwc{TODO: merge}


\subsubsection*{Windowing and Punctuation}
{\em Windowing} is another core concept in stream processing systems, where
aggregation operations like
moving averages or sums are defined over ``windows''---groupings of consecutive
events, gathered together into a set.  In \core{}, these transformers are just
\texttt{map}s over a stream whose elements are windows.  Given a per-window
aggregation transformer \texttt{f} from an individual window
$s^\star$ to a result type
$t$, plus a ``windowing strategy'' \texttt{win} which takes a stream $r^\star$
and turns it into a stream of windows $s^{\star\star}$, we can write
the windowed operation as \texttt{xs : r$^\star$ |-
map<f>(win(xs)) : t$^\star$}.  Appendix~\ref{app:windowing} defines both
sliding and tumbling size-based window operators,
as well as punctuation-based windowing, where windows are delimited by
punctuation marks inserted into the stream.

\section{Related Work}
\label{sec:related-work}

\noindent
Streams as a programming abstraction have their sources in early work
in the programming languages~\cite{gilles1974semantics,
  burge1975stream, stephens1997survey, thies2002streamit}
and database~\cite{Aurora, Borealis, TelegraphCQ, CACQ, STREAM2004,
  arasu2003cql, arasu2006cql}
communities.
Though streams have mostly been viewed as homogeneous
sequences, more interesting treatments have also been proposed.
%
For example, streams in the database literature are sometimes viewed as time-varying relations,
while the PL community has produced formalisms like process calculi and functional reactive programming.
To our knowledge, ours is first type system for
streams capturing both (1) heterogeneous patterns of
events over time and (2) combinations of parallel and sequential
data.

\SIMPLEHEADER{Sequential, homogeneous streams and dataflow programs}
Traditionally, streams have been viewed in the PL community as coinductive sequences~\cite{burge1975stream}: a stream of \texttt{A} has a single (co)constructor, $\texttt{cocons}: \texttt{Stream} A \to (A \times \texttt{Stream } A)$ and acts as a lazily evaluated infinite list.
In particular, this is the setting of traditional \emph{dataflow programming}~\cite{stephens1997survey}.
One major challenge in reasoning about dataflow over sequential
streams is the nondeterminism arising from operators
whose output may depend on the order in
which events arrive on multiple input streams.
Kahn's seminal ``process networks''~\cite{gilles1974semantics}
(including their restriction to synchronous networks~\cite{lee1987synchronous,thies2002streamit,BCEHlGdS2003SL})
avoid this problem by  allowing only \emph{blocking reads} of
messages on FIFO queues.
In contrast, the semantics of \core{} leverages its type structure to
guarantee deterministic parallel processing {\em without} blocking in many cases.
For example, in the context of a \ruleName{T-Let} rule,
if the type system can detect statically that a transformer is using
two parallel streams safely, it can read from them
simultaneously.

\SIMPLEHEADER{Partitioned streams}
Building on streams as homogeneous sequences,
modern stream processing systems such as
Flink~\cite{Flink,Flink2015},
Spark Streaming~\cite{SparkStreaming,Spark2013},
Samza~\cite{Samza,Samza2017},
Arc~\cite{arc-lang},
and Storm~\cite{Storm}
support {\em dynamic partitioning}: a stream type can define one stream
with many parallel substreams,
where the number of substreams and assignment of data to substreams is determined at runtime.
The type \texttt{Stream t} in these systems
is implicitly a parallel composition of homogeneous streams: $\texttt{t}^\star \| \cdots \| \texttt{t}^\star$.
Unlike in \core{}, these parallel substreams cannot have more general types.

Some which papers which attempt to build very general compile targets for stream
processing support parallelism in only restricted ways. For example,
Brooklet~\cite{soule2010universal} and the DON Calculus~\cite{dexter2022essence}
support data parallelism only as an optimization pass in limited cases.
This is because stream partitioning does not in general
preserve the semantics of the source program and can introduce
undesirable nondeterminism~\cite{mamouras2019data,schneider2013safe,hirzel2014catalog}.
While \core{} does not support dynamic partitioning, we hope to address it in future work; see Section~\ref{sec:future-work}.


\hyphenation{time-stamps}

\SIMPLEHEADER{Streams as time-varying relations}
In the database literature,
streams are often viewed as relations (sets of tuples) that vary over time.
Stream management systems in the early 2000s pioneered this paradigm,
including
Aurora~\cite{Aurora} and
Borealis~\cite{Borealis},
TelegraphCQ~\cite{TelegraphCQ} and
CACQ~\cite{CACQ},
and STREAM~\cite{STREAM2004}.
A time-varying relation can be viewed as either a function
from timestamps to finite relations
or an infinite set of timestamped values;
this correspondence was elegantly exploited by early streaming query languages such as CQL~\cite{arasu2003cql,arasu2006cql}
and remains popular today~\cite{jain2008towards, begoli2019one}.
Time-varying relations can be expressed in \core{} using
Kleene star and concatenation:
a relation of tuples of type \texttt{T} timestamped by \texttt{Time} can be expressed as
$\left(\texttt{T}^\star \cdot \texttt{\texttt{Time}}\right)^\star$.
We can also express the common pattern where parallel streams
are synchronized by a single timestamp (again, modulo dynamic
partitioning) with types like
$\left(\left(\texttt{T}^\star \|
\texttt{T}^\star\right) \cdot
\texttt{\texttt{Time}}\right)^\star$.
Each \texttt{Time} event is a punctuation mark containing the
timestamp of the prior set of
tuples~\cite{punctuation2003, heartbeats2005}.
Traditional systems include separate APIs for
operations that
modify punctuation (e.g., a \emph{delay} function that increments timestamps);
whereas in our system they are ordinary stream operators and
punctuation markers are ordinary events.

\SIMPLEHEADER{Streams as pomsets}
A sweet spot between the homogeneous sequential and relational
viewpoints is found in prior work treating streams as
{\em pomsets} (partially ordered multisets)~\cite{synch-schemas, mamouras2019data,
  kallas2020diffstream, kallas2022flumina, kappePomset},
inspired by work in concurrency
theory~\cite{mazurkiewicz1986trace, DiekertR1995}.
In a pomset,
data items may be completely ordered (a sequence), completely unordered (a bag), or somewhere in between.
%
Some recent works have proposed pomset-based types for streams~\cite{mamouras2019data,
synch-schemas}, but their types do not support concatenation
and do not come with type {\em systems}---programs must be shown to be
well typed semantically, rather than via syntactic typing
rules.

\emph{Functional reactive programming (FRP)}~\cite{elliott1997functional}
treats programs as incremental, reactive state machines
written using functional combinators.
The fundamental abstraction is a ``signal'':
a time-varying value \texttt{Sig(A) = Time -> A}.
Work on type systems for FRP has used modal and substructural types
~\cite{simplyratt, diamondsnotforever, fairreactive,
krishnaswami13}
to guarantee properties like causality, productivity,
and space leak freedom. While our type system is not
{\em designed} to address these issues, it does incidentally have bearing on
them. For one, our incremental semantics
demonstrates that \core{}'s type system enforces causality: since outputs that have been incrementally emitted cannot be retracted or changed,
the type system must ensure that past outputs cannot depend on future inputs. Similarly,
potential space leaks can be detected
statically by checking that only bounded-sized types are buffered using
\texttt{wait} or the buffering built into the left rules for sums and star.
Our current calculus does not guarantee productivity (new inputs must
eventually produce new outputs),
but in Section~\ref{sec:future-work}
we discuss how to remedy this by imposing guardedness conditions on
recursive calls.

\citet{ltlfrp} permits the type of a signal to vary over time, using dependent types
inspired by Linear Temporal Logic
\cite{pnueli1977LTL}.
This system includes an {\em until} type that behaves like our
concatenation type: a signal of type $A \,U\, B$ is a signal of type $A$,
followed by a signal of type $B$.
However, unlike parallel streams in our setting, time updates in
steps, discretely; i.e., parallel signals all present new values together, at the same time.
Concurrently with our work, \citet{bahr2023asynchronous} proposes a modal type system to weaken the synchronicity assumption;
however, it still treats signals as homogeneous: the type of data cannot change over time.
Lastly, \citet{event-driven} develop a modal type system which expresses low-level event handlers.
These are also purely synchronous, and the programs are written as event handlers
as opposed to high-level ``batch'' processors.

\emph{Stream Runtime Verification (SRV)} aims, broadly, to monitor streams
at runtime and provide boolean or numerical ``triggers'' that fire when they
satisfy some specification.  Many RV projects like LOLA
\cite{lola}, HLola \cite{hlola}, RTLola \cite{rtlola}, Striver \cite{striver},
HStriver \cite{hstriver} also provide high-level, declarative specification languages
for writing such monitors. Because these languages often use regular expressions or LTL
as a formalism, they often bear a resemblance to our stream types.
Despite this similarity, our goals and methods are quite different. Unlike
the dynamically-checked specifications of SRV, the types in \lang{} are static guarantees:
a stream program of type $s$ necessarily produces a stream of type $s$.

\SIMPLEHEADER{Streaming with Laziness}
It is folklore in the Haskell community that a ``sufficiently lazy'' list program
can be run as a streaming program using a clever trick with lazy IO
\cite{iteratees} \cite{jmct-chat}. This
``sufficient laziness'' condition is syntactically brittle, and requires an expert Haskell
programmer to carefully ensure that all functions involved are lazy in the just the
right way. Indeed, many Haskell programmers instead reach for combinator libraries like Pipes \cite{pipes}
FoldL \cite{foldl}, Conduit \cite{conduit}, Streamly \cite{streamly}, and others to ensure their programs have a streaming semantics.
In \lang{}, the type system takes care of this for you: all well-typed programs can be given a streaming
semantics. Moreover, the \core{} semantics gives a direct account of how pure
functions execute incrementally as state machines, as opposed to
the way that Haskell's non-strict semantics incidentally yields streaming behavior when combined with Lazy IO.


\SIMPLEHEADER{Session types and process calculi}
Another large body of work with similar vision is
session types for process calculi~\cite{honda2008multiparty}, where types describe complex
sequential protocols between communicating processes as they
evolve through time. A main difference from our work is that
the session type of a process describes the {\em protocol} for its
communications with other processes---i.e., the sequence of sends and
receives on different channels---while the stream type of a \core{}
program describes only the data that it communicates.
Indeed, a stream transformer might display many
patterns of communication with downstream transformers:
it can run in ``batch mode''---sending exactly one output after accepting
all available input---or in a sequence smaller steps, sending along
partial outputs as it receives partial inputs.
Also, a single channel in a process calculus
cannot carry parallel substreams: all events in a channel are ordered relative to each other.
Recently, \citet{frumin22session} proposed a session-types
interpretation of BI
that uses the bunched structure very
differently from \core.
In particular, processes of type $A * B$ and $A \wedge B$ both behave
semantically like a process of type $A$ in parallel with a process of
type $B$, while, in
\core{}, $s \cdot t$ and $s \| t$ describe very different streams.

\SIMPLEHEADER{Concurrent Kleene Algebras and regular expression types}
Stream types are partly inspired by Concurrent Kleene Algebras
(CKAs)~\cite{hoare2009concurrent} and related syntaxes for pomset
languages~\cite{kappePomset}, but we are apparently the first to use these formalisms as {\em types} in a programming
language rather than as a tool for reasoning about concurrency.
In particular, traditional applications of Kleene algebra such as NetKAT~\cite{anderson2014netkat} and Concurrent NetKAT~\cite{wagemaker2022concurrent}
use KA to model \emph{programs},
whereas
in \core{} we use the KA structure to describe the {\em data} that programs exchange,
while the programs themselves are written in a separate language.
We have also taken
inspiration from languages for programming with XML data
\cite[etc.]{HosoyaVouillonPierce2000,CDuce03,Castagna2002} using types
based on regular expressions.

\section{Conclusions and Future Work}
\label{sec:future-work}
We have proposed a new static type system for stream
programming, motivated by a novel variant of BI logic and
able to capture both complex temporal patterns and deterministic
parallel processing.

In the future, we hope to add more types to \core{}. Adding a support
for {\emph bags} --- unbounded parallelism, the parallel analog of Kleene star ---
would enable dynamic partitioning. \core{} also lacks function types.
The proof theory of BI would imply that there should be two (one for each context former),
but we have yet to investigate what these functions might mean in the streaming setting.

Further theoretical investigations include (1) alternate semantics for stream types, including
a denotational semantics as pomset morphisms, Kahn Process Networks \cite{gilles1974semantics}, or some category of state machines,
(2) eliminating the inertness restriction on let-bindings,
and (3) adding a {\em guardedness} condition on recursive calls
to ensure termination and hence productivity.

On the applied side, we plan to build a distributed implementation of \lang{}
by compiling \core{} terms to programs for an existing stream processing system like
Apache Storm \cite{Storm}, thus inheriting its desirable fault-tolerance and
delivery guarantees.  We hope
to build such a compiler and use it as a platform for experimenting with
type-enabled optimizations and resource usage analysis.




\begin{acks}
  We thank the reviewers for their feedback.
  We also thank Justin Lubin for feedback on drafts of this paper, and Alex
  Kavvos, Andrew Hirsch, Mae Milano, and Michael Arntzenius for helpful
  discussions about early versions of this work.
\end{acks}

\bibliographystyle{ACM-Reference-Format}
\bibliography{refs}

\newpage
\appendix
\section{\lang{} Implementation}
\label{app:artifact}

The \lang{} implementation is available at \url{http://www.github.com/anonymous-github-user/delta},
and has been tested with GHC version 9.2.7 and Stack version 2.9.3.

\begin{table}[h]
    \caption{Overview of the \lang{} Implementation}
    \begin{tabular}{lll}
        Name        & Description & Location \\\hline
        Var         & Various kinds of variables &  \texttt{Var.hs} \\
        Values      & Prefixes and environments &  \texttt{Values.hs} \\
        Types       & Types and contexts &  \texttt{Types.hs} \\
        HistPgm     & Types and semantics for historical programs &  \texttt{HistPgm.hs} \\
        CoreSyntax  & Syntax of \core{} terms and operations on them &  \texttt{CoreSyntax.hs} \\
        SurfaceSyntax & ASTs for the surface syntax &  \texttt{Frontend/SurfaceSyntax.hs} \\
        Parser      & Parser for the surface syntax &  \texttt{Frontend/Parser.y} \\
        ElabSyntax  & Elaborated syntax, and the elaborator code &  \texttt{Frontend/ElabSyntax.hs} \\
        Typechecker & Typechecker from elab syntax to core terms &  \texttt{Frontend/Typecheck.hs} \\
        Template    & Macros and Monomorphization for \core{} terms &  \texttt{Backend/Monomorphizer.hs} \\
        EnvSemantics   & Implementation of the \core{} semantics &  \texttt{Backend/EnvSemantics.hs} \\
        ErrUtil     & Error handling utilities &  \texttt{Util/ErrUtil.hs} \\
        PartialOrder & A partial order data structure &  \texttt{Util/PartialOrder.hs} \\
        PrettyPrint & Pretty printer typeclass &  \texttt{Util/PrettyPrint.hs} \\
    \end{tabular}
\end{table}

\newpage

\section{Examples}
\label{app:examples}
This appendix collects some additional examples of programming with \lang{}

\subsection{Running Fold}
\label{app:running-fold}

We can also define a {\em running} fold operation on star streams,
which outputs its partial results as it goes.

\begin{lstlisting}
fun runningFold[s,t]<f : {t}(s) -> t>{acc : t} (xs : s*) : t* =
  case xs of
    nil => nil
  | y :: ys => wait f{acc}(y) as acc' do
                 {acc'} :: runningFold{acc'}(ys)
               end
\end{lstlisting}

\subsection{Partitioning}
\label{app:partitioning}

A crucial streaming idiom is partitioning, where a homogeneous stream of
data is split into two or more parallel streams, which are then
routed to different downstream nodes in the dataflow graph.
The purpose of partitioning is to expose parallelism: the different
downstream operators can be run separately, potentially on different physical machines.
Depending on the situation, a programmer may choose to use different partitioning strategies.
In \core{}, some common partitioning strategies are implementable as transformers.

\subsubsection*{Round Robin Partitioning}
A round-robin partitioner fairly distributes an incoming stream of type
$s^\star$ into a parallel pair of streams $s^\star \| s^\star$.
It does this by sending the first element to the left branch, the second to the right, the third to the left, and so on.
In \core{}, we write this by maintaining a Boolean accumulator,
and negating after each item. If the Boolean is true, we send the element
left, if it's false, we send it right.

\begin{lstlisting}
fun roundRobin[s]{b : Bool}(xs : s*) : s* || s* =
  case xs of
    nil => (nil,nil)
  | y::ys =>
      let (zs,ws) = roundRobin{!b}(ys) in
      if {b} then (y::zs,ws) else (zs,y::ws)
\end{lstlisting}

\subsubsection*{Decision-Based Partitioning}
A decision-based partitioner routes stream elements based on the result of a predicate.

\begin{lstlisting}
fun decPartition[s,t,r]<f : (s) -> t + r>(xs : s*) : t* || r* =
  case xs of
    nil => (nil,nil)
  | y::ys =>
          let (ts,rs) = decPartition(ys) in
          case f(y) of
               inl t => (t::ts,rs)
             | inr r => (ts,r::rs)
\end{lstlisting}

\subsection{Windowing and Punctuation}
\label{app:windowing}
Many kinds of windows have been considered in the literature.
The most common windows are event-based --- windows defined by the number of
elements they'll contain --- and time-based --- windows which contain
all the events from a fixed length of time. Windows can also be tumbling ---
the next window starts after the previous ends --- or sliding ---
every event could begin a new window.

In \core{}, windowed operators are just \texttt{map}s over a stream whose
elements are windows. Given a per-window stream transformer \texttt{f}
which takes windows $s^\star$ to a result type $t$, and a ``windowing strategy'' \texttt{win}
which takes a stream $r^\star$ and turns it into a stream of windows $\left(s^\star\right)^\star$,
we can write a windowed operation of type $r^\star \to t^\star$ as follows:
\texttt{xs : r$^\star$ |- map(f)(win(xs)) : t$^\star$}.

For example, if we wanted to compute a size-3 sliding sum of a stream of
\texttt{Int}s, we would use a windower \texttt{win} which takes $\texttt{Int}^\star$
to $\left(\texttt{Int}^\star\right)^\star$ where the inner streams are the windows,
and \texttt{f} from $\texttt{Int}^\star$ to \texttt{Int} is the sum operation.

Every per-window function commonly used in stream processing practice operates on entire windows at once,
which is accomplished in \core{} by \texttt{wait}-ing on
the whole window, and then
aggregating it with an embedded historical program.
For this reason, we focus primarily on the window construction aspect.

\subsubsection*{Fixed-Size Tumbling Windows}
The $k$-size tumbling windower creates windows of size $k$,
where each new window starts immediately after the last window ended.
For instance when $k=2$, a stream $1,2,4,7,3,8,\dots$ turns into a stream
$\langle 1,2 \rangle, \langle 4,7 \rangle, \langle 3,8 \rangle,\dots$.
The code for a fixed-size tumbling window is exactly the functional code for
computing $k$-strides of a list, by grouping together the first $k$ elements,
and recursing down the rest of the stream.

\begin{lstlisting}
fun firstN[s]{n : Int}(xs : s*) : s* . s* =
  case xs of
    nil => (nil;nil)
   | y::ys => if {n > 0} then
                let (predN;rest) = firstN{n-1}(ys) in
                (y::predN;rest)
              else (nil;y::ys)
  
fun tumble[s]{k : Int}(xs : s*) : s** =
  let (first;rest) = firstN[s]{k}(xs) in first :: tumble{k}(rest)

\end{lstlisting}

$k$-size window transformers can actually have the even stronger output
type $\left(s^k\right)^\star$, where $s^k$ is the $k$-fold concatenation of $s$. If the window function being used
requires that the windows all have exactly size $k$ (like taking pairwise differences for $k = 2$),
this type can be used instead. The following program implements size-2 windows with this stronger type
by casing two-deep into the stream at a time, and pairing up elements into concatenation pairs.

\begin{lstlisting}
fun parsepairs[s](xs : s*) : (s . s)* =
  case xs of
    nil => nil
  | y :: ys => case ys of
                 nil => nil
               | z :: zs => (y;z) :: parsepairs(zs)
\end{lstlisting}

\subsubsection*{Fixed-Size Sliding Windows}
A $k$-sized sliding windower produces a new window for each
new element, including both the new element and the $k-1$ previous ones.
The code for this windower keeps the current window under construction in
memory.  When each new stream element arrives, we emit the current window.  For
the first $k$ elements, we only add to the window. After $k$, we start evicting
from the window.
\begin{lstlisting}
fun slidingWindower(acc : s$\superstar$; xs : s$\superstar$) : s$^{\star\star}$ =
  case xs of
    nil => acc :: nil
  | y::ys => wait y do
               let next = {if |acc| < k then y :: acc else y :: (init acc)} in
               next :: slidingWindower(next;xs)
             end
\end{lstlisting}

\subsubsection*{Punctuation-Based Windows}
Time-based windows are commonly implemented
by way of {\em punctuation}: unit elements
inserted into a stream to authoritatively mark
that a period of time has ended.
This is required because in the presence
of network delays, it's impossible to know
if a time period is over (and so a window can be emitted)
or if there are more elements in the period to arrive.
A punctuated stream has type $(1 + s)^\star$,
where the punctuation events mark the end of each time period.

The following code computes a windowed stream $\left(s^\star\right)^\star$ from
a punctuated stream $(\varepsilon+s)^\star$ by emitting windows which are the (potentially empty) runs
of $s$s between punctuation marks.

\begin{lstlisting}
fun tilFirstPunc[s](xs : (Eps + s)*) : s* . (Eps + s)* =
  case xs of
    nil => (nil;nil)
  | y::ys => case y of
               inl _ => (nil;ys)
             | inr s => let (cur;rest) = tilFirstPunc(ys) in
                        (s::cur;rest)

fun puncWindow[s](xs : (Eps + s)*) : s** =
  let (run;rest) = tilFirstPunc[s](xs) in
  run :: puncWindow(rest)
\end{lstlisting}

\subsubsection*{Merging Streams and Synchronizing Punctuation}
\label{app:merge}
Parallel streams of star type can be {\em synchronized}, pairing off one element
from one stream with one element of another. Given a stream of type $s^\star \| t^\star$,
we can produce a stream of type $\left(s \| t\right)^\star$.
This type's similarity to the standard functional program \texttt{zip} is more than just surface level:
the program below has essentially the same code.

\begin{lstlisting}
    fun sync[s,t](xs : s*, ys : t*) : (s || t)* =
    case xs of
      nil => nil
    | x'::xs' => case ys of
                   nil => nil
                 | y'::ys' => wait x',y' do
                                {(x',y')} :: sync(xs',ys')
                              end

\end{lstlisting}

Semantically, this program \texttt{wait}s until a full element from each of the
parallel input streams has arrived, sends them both out, and then continues
with zipping the two tails. This is necessarily blocking: the output type
guarantees that exactly one $s$ and $t$ will be produced before the next pair
begins, and so we must wait for both to arrive before sending the other out.
The upshot is that because this program is well typed in \core{}, it is
necessarily deterministic. This gives us the deterministic merge operation that
was needed to prevent the bug when averaging data from a pair of sensors in
Section~\ref{sec:extended-motivation}.

Moreover, for parallel streams of windows, synchronization enables
databases-style {\em streaming joins}.  Given parallel
streams $\left(s^\star\right)^\star$ and $\left(t^\star\right)^\star$, we can
synchronize to get $\left(s^\star \| t^\star\right)^\star$, and then apply a
join operation to each parallel pair of windows.

\newpage

\section{Technicalities}
\label{app:technical}

This appendix collects technical definitions that did not fit in the
main body of the paper.

\subsection{Basics}

Stream types are defined by the following grammar. The base types included are the unit type $\onet$ which types streams that contain
exactly one unit element, the type of the empty stream $\epst$, and the type of streams consisting of a single integer, $\intt$.
Larger types include the concatenation type $s \cdot t$, the sum type $s + t$, the parallel stream type $s \| t$, and the star type $s^\star$.

$$
\begin{array}{ll}
    s,t,r :=& 1 \mid \varepsilon \mid \intt \mid s \cdot t \mid s + t \mid s \| t \mid s^\star
\end{array}
$$

Contexts in the stream types calculus system have a bunched structure. The context former $\commactx{\Gamma}{\Delta}$
corresponds to the parallel type, while the context former $\semicctx{\Gamma}{\Delta}$ corresponds to the concatenation type.
The two context formers share a unit, written as ``$\cdot$''.

$$
\begin{array}{ll}
    \Gamma ::= & \cdot \mid \commactx{\Gamma}{\Gamma} \mid \semicctx{\Gamma}{\Gamma} \mid x : s \\
\end{array}
$$

A stream type is \emph{null} if it includes no data.
Null types are parallel combinations of $\varepsilon$s.

\begin{definition}[Nullable]
    \label{def:nullable}
    We define a judgment $\nullable{s}$ as follows:

    \begin{mathpar}
      \inferrule{ }{\nullable{\varepsilon}}

      \inferrule{\nullable{s} \\ \nullable{t}}{\nullable{s \| t}}
    \end{mathpar}

    We extend to contexts pointwise.
    \begin{mathpar}
      \inferrule{ }{\nullable{\cdot}}

      \inferrule{\nullable{s}}{\nullable{x : s}}

      \inferrule{\nullable{\Gamma} \\ \nullable{\Gamma'}}{\nullable{\commactx{\Gamma}{\Gamma'}}}

      \inferrule{\nullable{\Gamma} \\ \nullable{\Gamma'}}{\nullable{\semicctx{\Gamma}{\Gamma'}}}
    \end{mathpar}
\end{definition}

Prefixes are also like in the main paper,
with a definition $\isMaximal{p}$ for ``complete'' prefixes, and a typing relation $\prefixHasType{p}{s}$.

\begin{definition}[Prefixes]
  \label{def:prefix}
  The grammar of prefixes is given by:
    \begin{equation*}
        \begin{split}
        p ::= \onepA \ |\  \onepB \ |\  \epsp \ |\  \parp{p}{p'} \\
                        \ |\  \catpA{p} \ |\  \catpB{p}{p'} \\
                        \ |\  \sumpEmp \ |\  \sumpA{p} \ |\  \sumpB{p} \\
                        \ |\  \stpEmp \ |\  \stpDone \\
                        \ |\  \stpA{p} \ |\  \stpB{p}{p'}
        \end{split}
\end{equation*}
\end{definition}

\begin{definition}[Maximal Prefix]
    \label{app:maximal}
    \begin{mathpar}
        \infer{ }{
            \isMaximal{\epsp}
        }

        \infer{ }{
            \isMaximal{\onepB}
        }

        \infer{
            \isMaximal{p_1}\\
            \isMaximal{p_2}
        }{
            \isMaximal{\parp{p_1}{p_2}}
        }

        \ENDOFLINE

        \infer{
            \isMaximal{p_1}\\
            \isMaximal{p_2}
        }{
            \isMaximal{\catpB{p_1}{p_2}}
        }

        \infer{
            \isMaximal{p}
        }{
            \isMaximal{\sumpA{p}}
        }

        \ENDOFLINE

        \infer{
            \isMaximal{p}
        }{
            \isMaximal{\sumpB{p}}
        }

        \infer{ }{\isMaximal{\stpDone}}

        \infer{
            \isMaximal{p}\\
            \isMaximal{p'}
        }{
            \isMaximal{\stpB{p}{p'}}
        }
    \end{mathpar}
\end{definition}

\begin{definition}[Well-Typed Prefixes]
\begin{mathpar}
    \infer{ }{
        \prefixHasType{\epsp}{\epst}
    }

    \infer{ }{
        \prefixHasType{\onepA}{\onet}
    }

    \infer{ }{
        \prefixHasType{\onepB}{\onet}
    }

        \ENDOFLINE

    \infer{
        \prefixHasType{p_1}{s}\\
        \prefixHasType{p_2}{t}\\
    }{
        \prefixHasType{\parp{p_1}{p_2}}{s \| t}
    }

    \infer{
        \prefixHasType{p}{s}
    }{
        \prefixHasType{\catpA{p}}{s \cdot t}
    }

        \ENDOFLINE

    \infer{
        \prefixHasType{p_1}{s}\\
        \isMaximal{p_1}\\
        \prefixHasType{p_2}{t}
    }{
        \prefixHasType{\catpB{p_1}{p_2}}{s \cdot t}
    }

    \infer{ }{
        \prefixHasType{\sumpEmp}{s+t}
    }

    \infer{
        \prefixHasType{p}{s}
    }{
        \prefixHasType{\sumpA{p}}{s+t}
    }

    \infer{
        \prefixHasType{p}{t}
    }{
        \prefixHasType{\sumpB{p}}{s+t}
    }

    \infer{ }{
        \prefixHasType{\stpEmp}{s^\star}
    }

    \infer{ }{
        \prefixHasType{\stpDone}{s^\star}
    }

    \infer{
        \prefixHasType{p}{s}
    }{
        \prefixHasType{\stpA{p}}{s^\star}
    }

    \infer{
        \prefixHasType{p}{s}\\
        \isMaximal{p}\\
        \prefixHasType{p'}{s^\star}
    }{
        \prefixHasType{\stpB{p}{p'}}{s^\star}
    }
\end{mathpar}
\end{definition}

For each type $s$, we define the ``empty'' prefix $\emp{s}$ inductively on the structure of $s$.

\label{app:emp}
\begin{definition}[Empty Prefix]
    \label{def:emptyPrefix}
    The empty prefix is defined as follows:
    \begin{quote}
    \begin{itemize}
        \item[($\epst$)] $\emp{\epst} = \epsp$
        \item[($\onet$)] $\emp{\onet} = \onepA$
        \item[($s \| t$)] $\emp{s \| t} = \parp{\emp{s}}{\emp{t}}$
        \item[($s + t$)] $\emp{s + t} = \sumpEmp$
        \item[($s \cdot t$)] $\emp{s \cdot t} = \catpA{\emp{s}}$
        \item[($s^\star$)] $\emp{s^\star} = \stpEmp$
    \end{itemize}
    \end{quote}

    We lift this to contexts in the natural way,
    with $\emp{\cdot} = \epsp$, and
    $\emp{\semicctx{\Gamma}{\Delta}} = \catpA{\emp{\Gamma}}$, and
    $\emp{\commactx{\Gamma}{\Delta}} = \parp{\emp{\Gamma}}{\emp{\Delta}}$.
\end{definition}

\begin{theorem}[Empty Prefix is Well-Typed]
\label{thm:empty-prefix-correct}
$\prefixHasType{\emp{s}}{s}$
\end{theorem}

\begin{definition}[Prefix is Empty]
    \begin{mathpar}

        \infer{ }{
            \isEmpty{\epsp}
        }

        \infer{ }{
            \isEmpty{\onepA}
        }

        \infer{
            \isEmpty{p_1}\\
            \isEmpty{p_2}
        }{
            \isEmpty{\parp{p_1}{p_2}}
        }

        \ENDOFLINE

        \infer{
            \isEmpty{p}
        }{
            \isEmpty{\catpA{p}}
        }

        \infer{ }{
            \isEmpty{\sumpEmp}
        }

        \infer{ }{
            \isEmpty{\stpEmp}
        }
    \end{mathpar}
\end{definition}

\begin{theorem}[Empty Prefix Is Empty]
    \label{thm:emp-is-empty}
    $\isEmpty{\emp{s}}$
\end{theorem}
\begin{proof}
    Induction on $s$.
\end{proof}

\begin{theorem}[Empty And Maximal Means Nullable]
    \label{thm:empty-and-maximal-imply-nullable}
    If $\prefixHasType{p}{s}$, and simultaneously $\isEmpty{p}$ and $\isMaximal{p}$, then $\nullable{s}$.
\end{theorem}
\begin{proof}
    By induction on $\prefixHasType{p}{s}$
\end{proof}

\subsection{Derivatives}

We define a 3-place relation $\derivrel{p}{s}{s'}$ between a prefix and two types.

\begin{definition}[Derivatives]
    \label{tdef:derivrel}
    \begin{mathpar}
        \infer{ }{
            \derivrel{\epsp}{\epst}{\epst}
        }

        \infer{ }{
            \derivrel{\onepA}{\onet}{\onet}
        }

        \infer{ }{
            \derivrel{\onepB}{\onet}{\epst}
        }

        \infer{
            \derivrel{p_1}{s}{s'}\\
            \derivrel{p_2}{t}{t'}
        }{
            \derivrel{\parp{p_1}{p_2}}{s \| t}{s' \| t'}
        }

        \ENDOFLINE

        \infer{
            \derivrel{p}{s}{s'}
        }{
            \derivrel{\catpA{p}}{s \cdot t}{s' \cdot t}
        }

        \infer{
            \derivrel{p_2}{t}{t'}
        }{
            \derivrel{\catpB{p_1}{p_2}}{s \cdot t}{t'}
        }

        \infer{ }{
            \derivrel{\sumpEmp}{s+t}{s+t}
        }

        \ENDOFLINE

        \infer{
            \derivrel{p}{s}{s'}
        }{
            \derivrel{\sumpA{p}}{s+t}{s'}
        }

        \infer{
            \derivrel{p}{t}{t'}
        }{
            \derivrel{\sumpA{p}}{s+t}{t'}
        }

        \infer{ }{
            \derivrel{\stpEmp}{s^\star}{s^\star}
        }

        \infer{ }{
            \derivrel{\stpDone}{s^\star}{\varepsilon}
        }

        \ENDOFLINE

        \infer{
            \derivrel{p}{s}{s'}
        }{
            \derivrel{\stpA{p}}{s^\star}{s' \cdot s^\star}
        }

        \infer{
            \derivrel{p'}{s^\star}{s'}
        }{
            \derivrel{\stpB{p}{p'}}{s^\star}{s'}
        }

    \end{mathpar}
\end{definition}

\begin{definition}[Context Derivatives]
    \begin{mathpar}
        \infer{ }{
            \derivrel{\eta}{\cdot}{\cdot}
        }

        \infer{
            \eta(x) \mapsto p\\
            \derivrel{p}{s}{s'}
        }{
            \derivrel{\eta}{x : s}{x : s'}
        }

        \infer{
            \derivrel{\eta}{\Gamma}{\Gamma'}\\
            \derivrel{\eta}{\Delta}{\Delta'}
        }{
            \derivrel{\eta}{\commactx{\Gamma}{\Delta}}{\commactx{\Gamma'}{\Delta'}}
        }

        \ENDOFLINE

        \infer{
            \derivrel{\eta}{\Gamma}{\Gamma'}\\
            \derivrel{\eta}{\Delta}{\Delta'}
        }{
            \derivrel{\eta}{\semicctx{\Gamma}{\Delta}}{\semicctx{\Gamma'}{\Delta'}}
        }
    \end{mathpar}
\end{definition}

Derivatives are functions defined when the prefix input is well-typed.
\begin{theorem}[Derivative Function]
   \label{thm:derivrel-fun}
   For any $p$ and $s$, there is at most one $s'$ such that $\derivrel{p}{s}{s'}$.
   If $\prefixHasType{p}{s}$, then such an $s'$ exists.
\end{theorem}
\begin{proof}
   Induction on the derivation of $\derivrel{p}{s}{s'}$ for uniqueness, and $\prefixHasType{p}{s}$ for existence.
\end{proof}
When it's guaranteed to exist, we write this $s'$ simply as $\deriv{p}{s}$.
The empty prefix is the identity for the derivative operator.

\begin{theorem}[Empty Prefix Derivative]
    \label{thm:derivrel-emp}
    $\deriv{\emp{s}}{s} = s$.
\end{theorem}

\begin{theorem}[Empty Context Derivative]
    \label{thm:empty-context-deriv}
    If $\emptyOn{\eta}{\Gamma}$ and $\envHasType{\eta}{\Gamma}$ then
    $\deriv{\eta}{\Gamma} = \Gamma$.
\end{theorem}

\begin{theorem}[Context Derivatives Function]
    \label{thm:derivrel-env-fun}
    There is at most one $\Gamma'$ such that $\derivrel{\eta}{\Gamma}{\Gamma'}$,
    and the $\Gamma'$ exists when $\envHasType{\eta}{\Gamma}$.
\end{theorem}
\begin{proof}
    Uniqueness by induction on the derivation of $\derivrel{\eta}{\Gamma}{\Gamma'}$,
    existence by induction on the derivation of $\envHasType{\eta}{\Gamma}$.
\end{proof}

\begin{theorem}[Maximal Derivative iff Nullable]
    \label{thm:maximal-deriv-null}
    If $\derivrel{p}{s}{s'}$ then $\isMaximal{p}$ if and only if $\nullable{s'}$
\end{theorem}

\begin{theorem}[Only Prefix of a Null Type is Empty]
    If $\prefixHasType{p}{s}$ and $\nullable{s}$, then $p = \emp{s}$
\end{theorem}

\subsection{Environments}

\begin{definition}[Environments and Typing]
    An environment is a partial map $\eta : \texttt{Var} \to \texttt{Prefix}$.
    We write $\envHasType{\eta}{\Gamma}$ to mean that $\eta$ is a well-typed environment for $\Gamma$.

    \begin{mathpar}
        \infer{ }{
            \envHasType{\eta}{\cdot}
        }

        \infer{
            \eta(x) \mapsto p\\
            \prefixHasType{p}{s}
        }{
            \envHasType{\eta}{x : s}
        }

        \infer{
            \envHasType{\eta}{\Gamma}\\
            \envHasType{\eta}{\Delta}\\
        }{
            \envHasType{\eta}{\commactx{\Gamma}{\Delta}}
        }

        \ENDOFLINE

        \infer{
            \envHasType{\eta}{\Gamma}\\
            \envHasType{\eta}{\Delta}\\
            \emptyOn{\eta}{\Delta} \vee \maximalOn{\eta}{\Gamma}
        }{
            \envHasType{\eta}{\semicctx{\Gamma}{\Delta}}
        }
    \end{mathpar}
\end{definition}

\begin{definition}[All Maximal, All Empty, Agreement]
    \label{def:agree}
    For a set $S$, we say $\emptyOn{\eta}{S}$ if for all $x \in S$, there is some $p$ such that $\eta(x) \mapsto p$, and $\isEmpty{p}$.
    We say $\maximalOn{\eta}{S}$ if for all $x \in S$, there is some $p$ such that $\eta(x) \mapsto p$, and $\isMaximal{p}$.
    We write $\emptyOn{\eta}{\Gamma}$ and $\maximalOn{\eta}{\Gamma}$ to mean $\emptyOn{\eta}{\dom{\Gamma}}$ and $\maximalOn{\eta}{\dom{\Gamma}}$, respectively.
    We also write $\emptyOn{\eta}{e}$ and $\maximalOn{\eta}{e}$ to mean $\emptyOn{\eta}{\text{fv}(e)}$ and $\maximalOn{\eta}{\text{fv}(e)}$, respectively.

    Finally, we say that $\eta$ and $\eta'$ agree on $\Delta$ and $\Delta'$, written $\agree{\eta}{\eta'}{\Delta}{\Delta'}$ if
    $\maximalOn{\eta}{\Delta} \implies \maximalOn{\eta'}{\Delta'}$, and $\emptyOn{\eta}{\Delta} \implies \maximalOn{\eta'}{\Delta'}$
\end{definition}

An environment is also an environment for every subcontext.
\begin{theorem}[Environment Subcontext Lookup]
    \label{thm:env-subctx-lookup}
    If $\envHasType{\eta}{\Gamma(\Delta)}$, then $\envHasType{\eta}{\Delta}$
\end{theorem}
\begin{proof}
    Induction on $\Gamma(-)$.
\end{proof}

Moreover, replacing a the environment $\eta|_\Delta$ for a subcontext $\Delta$ with another environment $\eta'$ for another context $\Delta'$
yields a well-typed context, so long as $\eta$ and $\eta'$ agree on $\Delta$ and $\Delta'$. If $\eta$ was maximal (on $\Delta$) then $\eta'$ must also be
(on $\Delta'$), and if $\eta$ was empty (on $\Delta$), then $\eta'$ must also be empty (on $\Delta'$).
\begin{theorem}[Environment Subcontext Bind]
    \label{thm:env-subctx-bind}
    If $\envHasType{\eta}{\Gamma(\Delta)}$ and $\envHasType{\eta'}{\Delta'}$
    such that $\agree{\eta}{\eta'}{\Delta}{\Delta'}$
    then $\envHasType{\eta \cdot \eta'}{\Gamma(\Delta')}$
\end{theorem}
\begin{proof}
    Induction on the structure of $\Gamma(-)$, inverting everything in sight.
\end{proof}

\begin{theorem}[Environment Par Bind]
    \label{thm:env-par-bind}
    If $\envHasType{\eta}{\Gamma(z : s \| t)}$ and $\eta(z) \mapsto \parp{p_1}{p_2}$ 
    then $\envHasType{\eta[x \mapsto p_1, y \mapsto p_2]}{\Gamma(\commactx{x:s}{y:t})}$
\end{theorem}
\begin{proof}
    By Theorem~\ref{thm:env-subctx-bind}.
\end{proof}

\begin{theorem}[Environment Cat Bind 1]
    \label{thm:env-cat-bind-1}
    If $\envHasType{\eta}{\Gamma(z : s \cdot t)}$ and $\eta(z) \mapsto \catpA{p}$ 
    then $\envHasType{\eta[x \mapsto p, y \mapsto \emp{t}]}{\Gamma(\semicctx{x:s}{y:t})}$
\end{theorem}
\begin{proof}
    By Theorem~\ref{thm:env-subctx-bind}.
\end{proof}

\begin{theorem}[Environment Cat Bind 2]
    \label{thm:env-cat-bind-2}
    If $\envHasType{\eta}{\Gamma(z : s \cdot t)}$ and $\eta(z) \mapsto \catpB{p_1}{p_2}$
    then $\envHasType{\eta[x \mapsto p_1, y \mapsto p_2]}{\Gamma(\semicctx{x:s}{y:t})}$
\end{theorem}
\begin{proof}
    By Theorem~\ref{thm:env-subctx-bind}.
\end{proof}

Lastly, the structure of the above subcontext replacement operation is compatible with derivatives.
Taking the derivative of $\Gamma(\Delta)$ by $\eta$ yields $\Gamma'(\deriv{\eta}{\Delta})$ for some $\Gamma'(-)$,
and for \emph{any} other filler $\Delta_0$ and environment $\envHasType{\eta_0}{\Delta_0}$,
the outer derivative bit of the derivative remains unchanged: $\deriv{\eta \cup \eta_0}{\Gamma(\Delta_0)}$ is $\Gamma'(\deriv{\eta_0}{\Delta_0})$
\begin{theorem}[Environment Subcontext Bind Derivative]
    \label{thm:env-subctx-bind-deriv}
    If $\derivrel{\eta}{\Gamma(\Delta)}{\Gamma_0}$ then there is some $\Gamma'(-)$ such that
    for all $\Delta'$ and $\Delta''$ and $\eta'$, if $\derivrel{\eta'}{\Delta'}{\Delta''}$ and $\agree{\eta}{\eta'}{\Delta}{\Delta'}$
    then $\derivrel{\eta \cup \eta'}{\Gamma(\Delta')}{\Gamma'(\Delta'')}$
\end{theorem}
\begin{proof}
    Induction on $\Gamma(-)$.
\end{proof}

\begin{theorem}[Environment Par Derivative]
    \label{thm:env-par-bind-deriv}
    If $\derivrel{\eta}{\Gamma(z : s \| t)}{\Gamma'(z : s' \| t')}$
    and $\eta(z) \mapsto \parp{p_1}{p_2}$ 
    then $\derivrel{\eta[x \mapsto p_1, y \mapsto p_2]}{\Gamma(\commactx{x:s}{y:t})}{\Gamma'(\commactx{x : s'}{y : t'})}$
\end{theorem}
\begin{proof}
    By Theorem~\ref{thm:env-subctx-bind-deriv}.
\end{proof}

\begin{theorem}[Environment Cat Derivative 1]
    \label{thm:env-cat-bind-deriv-1}
    If $\derivrel{\eta}{\Gamma(z : s \cdot t)}{\Gamma'(z : s' \cdot t)}$
    and $\eta(z) \mapsto \catpA{p}$ 
    then $\derivrel{\eta[x \mapsto p, y \mapsto \emp{t}]}{\Gamma(\semicctx{x:s}{y:t})}{\Gamma'(\semicctx{x : s'}{y : t})}$
\end{theorem}
\begin{proof}
    By Theorem~\ref{thm:env-subctx-bind-deriv}.
\end{proof}

\begin{theorem}[Environment Cat Derivative 2]
    \label{thm:env-cat-bind-deriv-2}
    If $\derivrel{\eta}{\Gamma(z : s \cdot t)}{\Gamma'(z : t')}$
    and $\eta(z) \mapsto \catpB{p_1}{p_2}$ 
    then $\derivrel{\eta[x \mapsto p_1, y \mapsto p_2]}{\Gamma(\semicctx{x:s}{y:t})}{\Gamma'(\semicctx{x : s'}{y : t'})}$
\end{theorem}
\begin{proof}
    By Theorem~\ref{thm:env-subctx-bind-deriv}.
\end{proof}

\begin{theorem}[Environment Lookup]
    \label{thm:env-lookup}
    For any $\eta$, there is at most one $p$ so that $\eta(x) \mapsto p$.
    When $\envHasType{\eta}{\Gamma(x:s)}$, this $p$ exists, and $\prefixHasType{p}{s}$.
\end{theorem}
\begin{proof}
    The ``at most one'' $p$ is immediate from the fact that $\eta$ is a deterministic partial function.
    If $\envHasType{\eta}{\Gamma(x:s)}$ then $\envHasType{\eta}{x : s}$ by Theorem~\ref{thm:env-subctx-lookup}.
    By inversion, there is some $\prefixHasType{p}{s}$ such that $\eta(x) \mapsto p$.
\end{proof}

\begin{theorem}[Environment Lookup Derivative]
    \label{thm:env-lookup-deriv}
    Suppose:
    \begin{enumerate}
        \item $\eta(x) = p$
        \item $\envHasType{\eta}{\Gamma(x:s)}$
        \item $\derivrel{p}{s}{s'}$
        \item $\derivrel{\eta}{\Gamma(x:s)}{\Gamma_0}$
    \end{enumerate}
    Then there is some $\Gamma'(-)$ such that $\Gamma_0 = \Gamma'(x : s')$.
\end{theorem}
\begin{proof}
    Immediate by Theorem~\ref{thm:env-subctx-bind-deriv}
\end{proof}

\subsection{Concatenation}
\label{app:concatenation}

More generally, we often want to concatenate a prefix $p$ of $s$
with a prefix $p'$ of $\deriv{p}{s}$. This is defined with another 3-place,
type-indexed relation.

\begin{definition}[Prefix Concatenation]
    \label{def:prefix-concat}
    We define a relation $\prefixConcatRel{p}{p'}{p''}$.

    \begin{mathpar}
        \inferrule{ }{\prefixConcatRel{\epsp}{\epsp}{\epsp}}

        \ENDOFLINE

        \inferrule{\prefixHasType{p}{\onet}}{\prefixConcatRel{\onepA}{p}{p}}

        \inferrule{ }{\prefixConcatRel{\onepB}{\epsp}{\onepB}}

        \ENDOFLINE

        \inferrule{
            \prefixConcatRel{p_1}{p_1'}{p_1''}\\
            \prefixConcatRel{p_2}{p_2'}{p_2''}\\
        }{\prefixConcatRel{\parp{p_1}{p_2}}{\parp{p_1'}{p_2'}}{\parp{p_1''}{p_2''}}}

        \ENDOFLINE

        \inferrule{
            \prefixConcatRel{p}{p'}{p''}
        }{
            \prefixConcatRel{\catpA{p}}{\catpA{p'}}{\catpA{p''}}
        }

        \inferrule{
            \prefixConcatRel{p}{p'}{p''}
        }{
            \prefixConcatRel{\catpA{p}}{\catpB{p'}{p'''}}{\catpB{p''}{p'''}}
        }

        \ENDOFLINE

        \inferrule{
            \prefixConcatRel{p'}{p''}{p'''}
        }{
            \prefixConcatRel{\catpB{p}{p'}}{p''}{\catpB{p}{p'''}}
        }

        \ENDOFLINE

        \inferrule{ }{
            \prefixConcatRel{\sumpEmp}{p}{p}
        }

        \inferrule{
            \prefixConcatRel{p'}{p''}{p''}
        }{
            \prefixConcatRel{\sumpA{p}}{p'}{\sumpA{p''}}
        }

        \inferrule{
            \prefixConcatRel{p}{p'}{p''}
        }{
            \prefixConcatRel{\sumpB{p}}{p'}{\sumpB{p''}}
        }

        \ENDOFLINE

        \inferrule{ }{
            \prefixConcatRel{\stpEmp}{p}{p}
        }

        \inferrule{ }{
            \prefixConcatRel{\stpDone}{\epsp}{\stpDone}
        }

        \ENDOFLINE

        \inferrule{
            \prefixConcatRel{p}{p'}{p''}
        }{
            \prefixConcatRel{\stpA{p}}{\catpA{p'}}{\stpA{p''}}
        }

        \ENDOFLINE

        \inferrule{
            \prefixConcatRel{p}{p'}{p''}
        }{
            \prefixConcatRel{\stpA{p}}{\catpB{p'}{p'''}}{\stpB{p''}{p'''}}
        }

        \ENDOFLINE

        \inferrule{
            \prefixConcatRel{p'}{p''}{p'''}
        }{
            \prefixConcatRel{\stpB{p}{p'}}{p''}{\stpB{p}{p'''}}
        }

    \end{mathpar}

\end{definition}

This relation is a function when the inputs are well-typed.
Because of this, when $\prefixHasType{p}{s}$ and $\prefixHasType{p'}{\deriv{p}{s}}$,
we write $\prefixConcat{p}{p'}$ for the unique $p''$ that the following theorem guarantees.

\begin{theorem}[Prefix Concatenation Function]
    \label{thm:prefix-concat-fun}
    For all $p,p'$ and $s$, there is at most one $p''$ such that $\prefixConcatRel{p}{p'}{p''}$.
    If $\prefixHasType{p}{s}$ and $\prefixHasType{p'}{\deriv{p}{s}}$, then
    such a $p''$ exists, and satisfies:
    \begin{enumerate}
        \item $\prefixHasType{p''}{s}$
        \item $\deriv{p''}{s} = \deriv{p'}{\deriv{p}{s}}$
    \end{enumerate}
\end{theorem}
\begin{proof}
    Existence, (1), and (2) follow by induction on the derivation of $\prefixHasType{p}{s}$. Uniqueness is immediate by the fact that the relation is a function.
\end{proof}

\begin{theorem}[Prefix Concatenation Empty]
    \label{thm:prefix-concat-emp}
    If $\prefixHasType{p}{s}$, then $\prefixConcatRel{\emp{s}}{p}{p}$ and $\prefixConcatRel{p}{\emp{\deriv{p}{s}}}{p}$
\end{theorem}
\begin{proof}
    Induction on the derivation of $\prefixHasType{p}{s}$.
\end{proof}

\begin{theorem}[Maximal Prefix Concatenation]
    \label{thm:concat-maximal}
    Suppose $\prefixConcatRel{p}{p'}{p''}$.
    If $p''$ is maximal, then $p'$ is maximal. If $p$ or $p'$ is maximal, then $p''$ is maximal.
    Moreover, if $p$ is maximal, then $p'' = p$.
\end{theorem}
\begin{proof}
    Induction on the derivation of $\prefixConcatRel{p}{p'}{p''}$.
\end{proof}

\begin{theorem}[Prefix Concatenation Associativity]
    \label{thm:prefix-concat-assoc}
    $\prefixConcat{p}{\left(\prefixConcat{p'}{p''}\right)} = \prefixConcat{\left(\prefixConcat{p}{p'}\right)}{p''}$,
    when defined.
\end{theorem}
\begin{proof}
    Induction on derivations of concatenation.
\end{proof}

\begin{definition}[Environment Concatenation]
    \label{def:env-cat}
    We write $\prefixConcatRel{\eta}{\eta'}{\eta''}$ to mean that
    $\eta''$ is the function defined on the largest subset $S$ of $\dom{\eta} \cap \dom{\eta'}$
    such that for all $x \in S$, the prefix concatenation $\prefixConcatRel{\eta(x)}{\eta'(x)}{p}$ exists, and $\eta''(x) = p$, for all $x \in S$.
\end{definition}

\begin{theorem}[Environment Concatenation Function]
    For any $\eta$ and $\eta'$, there is at most one $\eta''$ such that $\prefixConcatRel{\eta}{\eta'}{\eta''}$,
    and such an $\eta''$ exists when $\envHasType{\eta}{\Gamma}$ and $\derivrel{\eta}{\Gamma}{\Gamma'}$ and $\prefixHasType{\eta'}{\Gamma'}$.
\end{theorem}
\begin{proof}
    Uniqueness by the "greatest" property, existence by Theorem~\ref{thm:prefix-concat-fun}.
\end{proof}

\begin{theorem}[Environment Concatenation Correctness]
    \label{thm:env-concat-correct}
    If $\envHasType{\eta}{\Gamma}$ and $\derivrel{\eta}{\Gamma}{\Gamma'}$ and $\envHasType{\eta'}{\Gamma'}$, and $\prefixConcatRel{\eta}{\eta'}{\eta''}$,
    then $\envHasType{\eta''}{\Gamma}$, and if $\derivrel{\eta'}{\Gamma'}{\Gamma''}$ then $\derivrel{\eta''}{\Gamma}{\Gamma''}$.
\end{theorem}

\begin{theorem}[Environment Concatenation Empty]
    \label{thm:env-concat-emp}
    If $\envHasType{\eta}{\Gamma}$ and $\prefixConcatRel{\eta}{\eta'}{\eta''}$, then:
    \begin{itemize}
        \item If $\emptyOn{\eta}{S}$ and then $\eta''|_{S} = \eta'|_{S}$
        \item If $\emptyOn{\eta'}{S}$, then $\eta''|_{S} = \eta|_{S}$
    \end{itemize}
\end{theorem}
\begin{proof}
    Induction on the derivation of $\envHasType{\eta}{\Gamma}$, using Theorem~\ref{thm:prefix-concat-emp}.
\end{proof}

\begin{theorem}[Maximal Environment Concatenation]
    \label{thm:env-concat-maximal}
    If $\prefixConcatRel{\eta}{\eta'}{\eta''}$, then $\maximalOn{\eta}{S}$ or $\maximalOn{\eta'}{S}$ if and only if $\maximalOn{\eta''}{S}$.
\end{theorem}
\begin{proof}
    Immediate corollary of Theorem~\ref{thm:concat-maximal}
\end{proof}

\begin{theorem}[Prefix Concatenation Associativity]
    \label{thm:env-concat-assoc}
    $\prefixConcat{\eta}{\left(\prefixConcat{\eta'}{\eta''}\right)} = \prefixConcat{\left(\prefixConcat{\eta}{\eta'}\right)}{\eta''}$,
    when defined.
\end{theorem}
\begin{proof}
    Corollary of Theorem~\ref{thm:prefix-concat-assoc}
\end{proof}

\subsection{Historical Contexts}
\label{app:histctx}

\begin{definition}[Historical Context]
    Contexts $\Omega := \cdot \,|\, \Omega,x:A$ are fully structural contexts, where the $A$ are STLC types.
\end{definition}

A stream type is ``flattened'' into an STLC type by turning concatenations and parallels
into products, and stars into lists.

\begin{definition}[Type and Context Flatten]
    \label{def:type-ctx-flatten}
    For $s$ a stream type, we define its flattening into an STLC type, denoted $\flatten{s}$, inductively:
    \begin{itemize}
      \item $\flatten{1} = 1$
      \item $\flatten{\varepsilon} = 1$
      \item $\flatten{s \cdot t} = \flatten{s} \times \flatten{t}$
      \item $\flatten{s \| t} = \flatten{s} \times \flatten{t}$
      \item $\flatten{s + t} = \flatten{s} + \flatten{t}$
      \item $\flatten{s^\star} = \texttt{list}\left(\flatten{s}\right)$
    \end{itemize}

    For $\Gamma$ a bunched context, we define its flattening to a standard context, $\flatten{\Gamma}$ inductively:

    \begin{itemize}
        \item $\flatten{\cdot} = \cdot$
        \item $\flatten{x : s} = x : \flatten{s}$
        \item $\flatten{\Gamma ; \Gamma'} = \flatten{\Gamma},\flatten{\Gamma'}$
        \item $\flatten{\Gamma , \Gamma'} = \flatten{\Gamma},\flatten{\Gamma'}$
    \end{itemize}

    For an STLC value $v : \flatten{s}$, we write $\histValToPrefix{v}{s}$ for the maximal prefix of type $s$
    that it corresponds to. Dually, for a maximal prefix $\prefixHasType{p}{s}$, we write $\flatten{p}$ for the STLC value of type
    $\flatten{s}$ it corresponds to.
\end{definition}

\begin{definition}[Historical Programs and Substitutions]
    \label{def:histpgm}
    Fix a language of terms $M$, with type system $\lctck{\Omega}{M}{A}$. Write its semantics as $M \downarrow v$.
    We assume that this relation is a decidable partial function, in the sense that $M$ evaluates to at most one $v$, and it is
    decidable whether or not such a $v$ exists. We write substitutions $\histsubsttck{\theta}{\Omega'}{\Omega}$. Substitutions have a contravariant action
    on terms, written $M[\theta]$: if $\lctck{\Omega}{M}{A}$, then $\lctck{\Omega'}{M[\theta]}{A}$.
    We lift this substitution action to \core{} terms compositionally, substituting into all historical terms.
    We write a list of such terms as $\overline{M}$, and lift the typing relation and semantics to
    lists of terms, written $\lctck{\Omega}{\overline{M}}{\overline{A}}$ and $\overline{M} \downarrow \theta$.
\end{definition}

\subsection{Context Subtyping}
\label{app:ctxsub}

The following is a full listing of subtyping rules.

\begin{definition}[Subtyping]
    \label{def:ctxsubty}
    \begin{mathpar}
        \infer[Sub-Cong]{
            \subty{\Delta}{\Delta'}
        }{
            \subty{\Gamma(\Delta)}{\Gamma(\Delta')}
        }

        \infer[Sub-Refl]{ }{
            \subty{\Gamma}{\Gamma}
        }

        \infer[Sub-Comma-Exc]{ }{
            \subty{\commactx{\Gamma}{\Delta}}{\commactx{\Delta}{\Gamma}}
        }

        \ENDOFLINE

        \infer[Sub-Sng-Wkn]{ }{
            \subty{x : s}{\cdot}
        }

        \infer[Sub-Comma-Wkn]{ }{
            \subty{\commactx{\Gamma}{\Delta}}{\Gamma}
        }

        \infer[Sub-Semic-Wkn-1]{ }{
            \subty{\semicctx{\Gamma}{\Delta}}{\Gamma}
        }

        \infer[Sub-Semic-Wkn-2]{ }{
            \subty{\semicctx{\Gamma}{\Delta}}{\Delta}
        }

        \ENDOFLINE

        \infer[Sub-Comma-Unit]{ }{\subty{\Gamma}{\commactx{\Gamma}{\cdot}}}

        \infer[Sub-Semic-Unit-1]{ }{\subty{\Gamma}{\semicctx{\Gamma}{\cdot}}}

        \infer[Sub-Semic-Unit-2]{ }{\subty{\Gamma}{\semicctx{\cdot}{\Gamma}}}
    \end{mathpar}
\end{definition}

Environment typing is preserved by subtyping, and derivatives preserve subtyping relations between contexts.

\begin{theorem}[Subtyping Preserves Environments]
    \label{thm:subty-env}
    If $\envHasType{\eta}{\Gamma}$ and $\subty{\Gamma}{\Delta}$ then $\envHasType{\eta}{\Delta}$
\end{theorem}
\begin{proof}
    By induction on $\subty{\Gamma}{\Delta}$, and inversion on $\envHasType{\eta}{\Gamma}$.
\end{proof}

\begin{theorem}[Derivatives Preserve Subtyping]
    \label{thm:deriv-subty}
    Suppose $\subty{\Gamma}{\Delta}$ and $\derivrel{\eta}{\Gamma}{\Gamma'}$ and $\derivrel{\eta}{\Delta}{\Delta'}$.
    Then, $\subty{\Gamma'}{\Delta'}$.
\end{theorem}
\begin{proof}
    By cases on $\subty{\Gamma}{\Delta}$, inverting the derivations of
    $\derivrel{\eta}{\Gamma}{\Gamma'}$ and
    $\derivrel{\eta}{\Delta}{\Delta'}$, and using the determinism of the derivative relation.
\end{proof}

\subsection{Type System}
\label{app:type-system}

\subsubsection{Inertness}
\label{app:inertness}
Most terms, like variables or case expressions, require some non-empty amount of
input to arrive for them to produce a non-empty output. However, this is not
true of all terms: constants like $\oneTm$ and $\nilTm$, (some) sequential terms
like $\catpairTm{\oneTm}{e}$ and $\consTm{\epsTm}{e}$, and sum terms $\inl{e}$
produce nonempty output even when given an entirely empty input prefix. Terms
like these contain ``information'' that they are always ready to produce, even
if there is no input to drive them forward. We call terms that are always ready to produce
output \emph{jumpy}, and terms that are not \emph{inert}.

As described in Section~\ref{sec:let-binding}, the type system requires that
let-bound terms are always inert to guarantee soundness of the semantics.
In particular, inertness is what guarantees that the ``agreement'' (Definition~\ref{def:agree}) requirement
in Theorem~\ref{thm:env-subctx-bind} to hold in the soundness case for \ruleName{T-Let}.
For arbitrary terms, the maximality component of agreement always holds (this by Lemma~\ref{lemma:batch-sem-aux}),
but the emptiness component of agreement requires inertness.

To enforce that the bodies of let-bindings are inert, we track a syntactic over-approximation of inertness with the type system, essentially as an
\emph{effect}. This is accomplished by giving every typing judgment an
\emph{inertness annotation}, $i ::= \I \mid \J$, and we ensure that if $e$ is
typed with annotation $\texttt{Inert}$, then $e$ produces empty output when
given an empty input. This invariant is proved as an additional consequence to
the soundness theorem. 

We note that the choice to include inertness in the type
system itself, as opposed to a predicate on (typed) terms, is essentially an
arbitrary one: we choose the former to minimize the number of assumptions
running around in our proofs.

For the most part, the inertness analysis is straightforward. Constants like
$\oneTm$ and $\nilTm$, and injections like $\sumInlTm{e}$, $\sumInrTm{e}$, and
$\consTm{e_1}{e_2}$ (secretly the right injection into $\epst + s\cdot s^\star$)
all have annotation $\J$.  Non-buffering elimination forms have the same
inertness as their bodies, and variables and $\epsTm$ are inert.  The most
important ones are in the rules \ruleName{T-Cat-R} and \ruleName{T-Plus-L} (and
the similar ones in \ruleName{T-Star-L} and \ruleName{T-Plus-L}).

The inertness requirement for \ruleName{T-Cat-R} says that if the resulting term
$\catpairTm{e_1}{e_2}$ is to be typed as inert, $e_1$ must be inert, and the type of $e_1$
must not be \texttt{null}. Otherwise, $\catpairTm{e_1}{e_2}$ could produce a maximal .

The rule for \ruleName{T-Plus-L} says that it is inert when the buffer
environment does not yet include a decision for which way to go ($\eta(z) =
\sumpEmp$). Note that in practice, this is always satisfied.  At the beginning of
execution, $\eta$ maps all variables to empty prefixes, and as soon as $\eta(z)$
gets either $\sumpA{p}$ or $\sumpB{p}$, we step to the corresponding branch. In
fact, the result of \emph{every} step is inert: otherwise we would've output a
larger prefix in that step!

\begin{definition}[Typing Rules] Figure~\ref{fig:typingrules-1} and Figure~\ref{fig:typingrules-2} present
the full typing rules.
\end{definition}

\begin{definition}[Recursion Signature]
A recursion signature $\Sigma$ is either empty (signaling that typechecking is not in the body of a recursive function),
or the signature $\Omega \mid \Gamma \to s \, @ \, i$ of a sequent which defines the recursive function we are currently checking the body of.
$\Sigma ::= \emptyset \mid \left(\Omega \mid \Gamma \to s \, @ \, i\right)$
\end{definition}

\begin{figure}[t]
\centering
    \begin{mathpar}
        \infer[T-Eps-R]{ }{
            \tck{\Omega}{\Gamma}{\Sigma}{\epsTm}{\epst}{i}
        }

        \infer[T-One-R]{ }{
            \tck{\Omega}{\Gamma}{\Sigma}{\oneTm}{\onet}{\J}
        }

        \infer[T-Var]{ }{
            \tck{\Omega}{\Gamma(x:s)}{\Sigma}{x}{s}{i}
        }

        \infer[T-Sub]{
            \tck{\Omega}{\Delta}{\Sigma}{e}{s}{i}\\
            \subty{\Gamma}{\Delta}
        }{
            \tck{\Omega}{\Gamma}{\Sigma}{e}{s}{i}
        }

        \ENDOFLINE

        \infer[T-Par-R]{
            \tck{\Omega}{\Gamma}{\Sigma}{e_1}{s}{i}\\
            \tck{\Omega}{\Gamma}{\Sigma}{e_2}{t}{i}
        }{
            \tck{\Omega}{\Gamma}{\Sigma}{\parpairTm{e_1}{e_2}}{s \| t}{i}
        }

        \infer[T-Cat-R]{
            \tck{\Omega}{\Gamma}{\Sigma}{e_1}{s}{i_1}\\
            \tck{\Omega}{\Delta}{\Sigma}{e_2}{t}{i_2}\\
            i_3 = \I \implies i_1 = \I \wedge \neg\left(\nullable{s}\right)
        }{
            \tck{\Omega}{\semicctx{\Gamma}{\Delta}}{\Sigma}{\catpairTm{e_1}{e_2}}{s \cdot t}{i_3}
        }

        \ENDOFLINE

        \infer[T-Par-L]{
            \tck{\Omega}{\Gamma(\commactx{x : s}{y : t})}{\Sigma}{e}{r}{i}
        }{
            \tck{\Omega}{\Gamma(z : s \| t)}{\Sigma}{\letparTm{x}{y}{z}{e}}{r}{i}
        }

        \infer[T-Cat-L]{
            \tck{\Omega}{\Gamma(\semicctx{x : s}{y : t})}{\Sigma}{e}{r}{i}
        }{
            \tck{\Omega}{\Gamma(z : s \cdot t)}{\Sigma}{\letcatTm{t}{x}{y}{z}{e}}{r}{i}
        }

        \ENDOFLINE

        \infer[T-Plus-R-1]{
            \tck{\Omega}{\Gamma}{\Sigma}{e}{s}{i}
        }{
            \tck{\Omega}{\Gamma}{\Sigma}{\inl{e}}{s+t}{\J}
        }

        \infer[T-Plus-R-2]{
            \tck{\Omega}{\Gamma}{\Sigma}{e}{t}{i}
        }{
            \tck{\Omega}{\Gamma}{\Sigma}{\inr{e}}{s+t}{i}
        }

        \ENDOFLINE

        \infer[T-Plus-L]{
            \envHasType{\eta}{\Gamma(z:s+t)}\\
            \derivrel{\eta}{\Gamma(z : s+t)}{\Gamma'}\\
            \tck{\Omega}{\Gamma(x:s)}{\Sigma}{e_1}{r}{i_1}\\
            \tck{\Omega}{\Gamma(y:t)}{\Sigma}{e_2}{r}{i_2}\\
            i = \I \implies \eta(z) = \sumpEmp
        }{
            \tck{\Omega}{\Gamma'}{\Sigma}{\sumcaseTm{r}{\eta}{z}{x}{e_1}{y}{e_2}}{r}{i}
        }

    \end{mathpar}
\caption{Full Typing Rules (Part 1)}
\label{fig:typingrules-1}
\end{figure}

\begin{figure}

\begin{mathpar}
        \inferrule[T-Star-R-1]{ }{
            \tck{\Omega}{\Gamma}{\Sigma}{\nilTm}{s^\star}{\J}
        }

        \inferrule[T-Star-R-2]{
            \tck{\Omega}{\Gamma}{\Sigma}{e_1}{s}{i_1}\\
            \tck{\Omega}{\Delta}{\Sigma}{e_2}{s^\star}{i_2}
        }{
            \tck{\Omega}{\semicctx{\Gamma}{\Delta}}{\Sigma}{\consTm{e_1}{e_2}}{s^\star}{\J}
        }

        \ENDOFLINE

        \inferrule[T-Star-L]{
            \envHasType{\eta}{\Gamma(z:s^\star)}\\
            \derivrel{\eta}{\Gamma(z : s^\star)}{\Gamma'}\\
            \tck{\Omega}{\Gamma(\cdot)}{\Sigma}{e_1}{r}{i_1}\\
            \tck{\Omega}{\Gamma(x:s;xs:s^\star)}{\Sigma}{e_2}{r}{i_2}\\
            i = \I \implies \eta(z) = \stpEmp
        }{
            \tck{\Omega}{\Gamma'}{\Sigma}{\starcaseTm{s,r}{\eta}{z}{e_1}{x}{xs}{e_2}}{r}{i}
        }

        \ENDOFLINE

        \inferrule[T-HistPgm]
        {
            \lctck{\Omega}{M}{\flatten{s}}
        }{
            \tck{\Omega}{\Gamma}{\Sigma}{\histPgmTm{M}{s}}{s}{\J}
        }

        \inferrule[T-Wait]{
            \envHasType{\eta}{\Gamma(x : s)}\\
            \derivrel{\eta}{\Gamma(x:s)}{\Gamma'}\\
            \tck{\Omega, x : \flatten{s}}{\Gamma(\cdot)}{\Sigma}{e}{t}{i}\\
            i' = \I \implies \neg\left(\isMaximal{\eta(z)}\right) \wedge \neg\left(\nullable{s}\right)
        }{
            \tck{\Omega}{\Gamma'}{\Sigma}{\waitTm{\eta}{t}{x}{e}}{t}{i'}
        }

        \ENDOFLINE

        \infer[T-Let]{
            \tck{\Omega}{\Delta}{\Sigma}{e_1}{s}{\I}\\
            \tck{\Omega}{\Gamma(x:s)}{\Sigma}{e_2}{t}{i}
        }{
            \tck{\Omega}{\Gamma(\Delta)}{\Sigma}{\cutTm{x}{e_1}{e_2}}{t}{i}
        }

        \infer[T-Rec]{
            \tck{\Omega'}{\Gamma'}{\Omega \mid \Gamma \to s \, @ \, i}{A}{\Gamma}{i}\\
            \lctck{\Omega'}{\overline{M}}{\Omega}
        }{
            \tck{\Omega'}{\Gamma'}{\Omega \mid \Gamma \to s \, @ \, i}{\recTm{\overline{M}}{A}}{s}{i}
        }

        \ENDOFLINE

        \infer[T-Fix]{
            \tck{\Omega}{\Gamma}{\Omega \mid \Gamma \to s \, @ \, i}{e}{s}{i}\\
            \lctck{\Omega'}{\overline{M}}{\Omega}\\
            \tck{\Omega'}{\Gamma'}{\Sigma}{A}{\Gamma}{i}
        }{
            \tck{\Omega'}{\Gamma'}{\Sigma}{\fixTmHistargs{\overline{M}}{A}{e}}{s}{i}
        }

        \infer[T-ArgsLet]{
            \tck{\Omega}{\Gamma'}{\Sigma}{A}{\Gamma}{i}\\
            \tck{\Omega}{\Gamma}{\Sigma}{e}{s}{i}
        }{
            \tck{\Omega}{\Gamma'}{\Sigma}{\cutTm{\Gamma}{A}{e}}{s}{i}
        }
\end{mathpar}
    
\caption{Full Typing Rules (Part 2)}
\label{fig:typingrules-2}
\end{figure}

These typing rules are mutually defined with another typing judgment
$\tck{\Omega}{\Gamma}{\Sigma}{A}{\Gamma'}{i}$, meaning that $A$
is a well-typed set of arguments (hence $A$) for a recursive call to a function accepting
inputs $\Gamma'$. Here, $A$ is an \emph{tree} of terms, with either
comma or semicolon nodes. This judgment ensures that $e_{\Gamma'}$ has
well-typed bindings for every variable $x : s$ in $\Gamma'$, and that the
variables that $e_{\Gamma'}$ uses are used in accordance with $\Gamma$, its
context.

\begin{definition}[Recursive Argument Typing]
    $$
    A ::= \cdot \mid e \mid \parpairArgsTm{A}{A'} \mid \catpairArgsTmA{A}{A'} \mid \catpairArgsTmB{A}
    $$

    \begin{mathpar}
        \infer[T-Args-Emp]{ }{
            \tck{\Omega}{\Gamma}{\Sigma}{\cdot}{\cdot}{i}
        }

        \infer[T-Args-Sng]{
            \tck{\Omega}{\Gamma}{\Sigma}{e}{s}{i}
        }{
            \tck{\Omega}{\Gamma}{\Sigma}{e}{\left(x : s\right)}{i}
        }

        \infer[T-Args-Semic-1]{
            \tck{\Omega}{\Gamma}{\Sigma}{A}{\Delta}{i_1}\\
            \tck{\Omega}{\Gamma'}{\Sigma}{A'}{\Delta'}{i_2}
        }{
            \tck{\Omega}{\semicctx{\Gamma}{\Gamma'}}{\Sigma}{\catpairArgsTmA{A}{A'}}{\semicctx{\Delta}{\Delta'}}{i_3}
        }

        \infer[T-Args-Semic-2]{
            \tck{\Omega}{\Gamma'}{\Sigma}{A}{\Delta'}{i}\\
            \nullable{\Delta}
        }{
            \tck{\Omega}{\semicctx{\Gamma}{\Gamma'}}{\Sigma}{\catpairArgsTmB{A}}{\semicctx{\Delta}{\Delta'}}{i}
        }

        \infer[T-Args-Comma]{
            \tck{\Omega}{\Gamma}{\Sigma}{A}{\Delta}{i}\\
            \tck{\Omega}{\Gamma}{\Sigma}{A'}{\Delta'}{i}
        }{
            \tck{\Omega}{\Gamma}{\Sigma}{\parpairArgsTm{A}{A'}}{\commactx{\Delta}{\Delta'}}{i}
        }

    \end{mathpar}
\end{definition}

\subsubsection*{Buffering Rules}
The left rules for star and sums, as well as Wait, include a {\em
buffer} in the term: a prefix of the input context, where we store inputs until
we have received enough to run the term.  For example, the \ruleName{Wait} rule
has this buffer $\eta$, which we gather until it includes a maximal prefix of $x : s$.

\begin{mathpar}

    \inferrule[T-Wait]{
        \envHasType{\eta}{\Gamma(x : s)}\\
        \derivrel{\eta}{\Gamma(x:s)}{\Gamma'}\\
        \tck{\Omega, x : \flatten{s}}{\Gamma(\cdot)}{\Sigma}{e}{s}{i}
        i' = \I \implies \neg\left(\isMaximal{\eta(z)}\right) \wedge \neg\left(\nullable{s}\right)
    }{
        \tck{\Omega}{\Gamma'}{\Sigma}{\waitTm{\eta}{t}{x}{e}}{t}{i'}
    }
\end{mathpar}

The buffer is included in the syntax of the term. Additionally, the context in the conclusion is $\deriv{p}{\Gamma(\Delta)}$.
If we've buffered $\eta$ of the input, the term is expecting the rest of the context.
Users of the calculus need not worry about this detail: when writing programs and when the program starts running,
the buffer is empty: $\eta = \emp{\Gamma(\Delta)}$, and since
$\deriv{\eta}{\Gamma(\Delta)} = \Gamma(\Delta)$, this returns \ruleName{Wait} to the expected rule presented
in the body of the paper. The other rules that include buffers are \ruleName{Plus-L} and \ruleName{Star-L}.

\subsection{Sink Terms}
\label{app:sink-term}

Once we have produced an entire maximal prefix $\prefixHasType{p}{s}$, a program $e$ of type $s$
needs to transition to a program emitting nothing: we compute this term from $p$ with $\sinkTm{p}$.

\begin{definition}[Sink Terms]
    We define a term $\sinkTm{p}$ by induction on $p$.
    \begin{itemize}
        \item $\sinkTm{\epsp} = \epsTm$
        \item $\sinkTm{\onepA} = \epsTm$
        \item $\sinkTm{\onepB} = \epsTm$
        \item $\sinkTm{\parp{p_1}{p_2}} = \parpairTm{\sinkTm{p_1}}{\sinkTm{p_2}}$
        \item $\sinkTm{\catpA{p}} = \sinkTm{p}$
        \item $\sinkTm{\catpB{p_1}{p_2}} = \sinkTm{p_2}$
        \item $\sinkTm{\sumpEmp} = \epsTm$
        \item $\sinkTm{\sumpA{p}} = \sinkTm{p}$
        \item $\sinkTm{\sumpB{p}} = \sinkTm{p}$
        \item $\sinkTm{\stpEmp} = \epsTm$
        \item $\sinkTm{\stpDone} = \epsTm$
        \item $\sinkTm{\stpA{p}} = \sinkTm{p}$
        \item $\sinkTm{\stpB{p}{p'}} = \sinkTm{p'}$
    \end{itemize}
\end{definition}

Note that (because it's easier to have this be a function rather than a relation) sink terms are defined for \emph{all} prefixes
rather than just the maximal ones.

Sink terms are closed, and have the type we expect for a stream transformer that has just emitted an maximal $p$ of type $s$.

\begin{theorem}[Sink Terms Typing]
    \label{thm:sink-typing}
    If $\isMaximal{p}$ and $\prefixHasType{p}{s}$ and $\derivrel{p}{s}{s'}$, then
    $\tck{\cdot}{\Gamma}{\emptyset}{\sinkTm{p}}{s'}{\I}$
\end{theorem}

The relevant concatenation property of sink terms is that they only depend on the the shape of
the type $s$ \emph{after} the prefix has been emitted, so adding more to the beginning does not change anything.

\begin{theorem}[Sink Term Concatenation]
    \label{thm:sink-concat}
    If $\prefixConcatRel{p}{p'}{p''}$, then
    $\sinkTm{p'} = \sinkTm{p''}$.
\end{theorem}
\begin{proof}
    By induction on $\prefixConcatRel{p}{p'}{p''}$.
\end{proof}

\begin{theorem}[Fixpoint Substitution]
    \label{thm:fixsubst}
    For $e,e'$ terms, we define $\fixsubst{e}{e'}$ compositionally over the structure of $e$, with
    the only two interesting cases being:
    $$
    \fixsubst{\left(\recTm{\overline{M}}{A}\right)}{e'} = \fixTmHistargs{\overline{M}}{\fixsubst{A}{e'}}{e'}
    $$
    and
    $$
    \fixsubst{\left(\fixTmHistargs{\overline{M}}{A}{e}\right)}{e'} = \fixTmHistargs{\overline{M}}{\fixsubst{A}{e'}}{e}
    $$
    We define this mutually with a substitution for arguments $A$, with $\fixsubst{A}{e'}$ defined compositionally over the structure of $A$.

    Then if $\tck{\Omega}{\Gamma}{\Omega | \Gamma \to s \, @ \, i'}{e'}{s}{i'}$,
    we have:
    \begin{enumerate}
        \item If $\tck{\Omega'}{\Delta}{\Omega|\Gamma \to s \, @ \, i'}{e}{t}{i}$ then $\tck{\Omega'}{\Delta}{\cdot}{\fixsubst{e}{e'}}{t}{i}$
        \item If $\tck{\Omega'}{\Delta}{\Omega | \Gamma \to s \, @ \, i'}{A}{\Gamma'}{i}$, then $\tck{\Omega'}{\Delta}{\cdot}{\fixsubst{A}{e'}}{\Gamma'}{i}$
    \end{enumerate}
\end{theorem}
\begin{proof}

    (1) and (2) are proved by a routine simultaneous induction on typing derivations.
\end{proof}

\subsection{Semantics}
\label{app:semantics}

\begin{definition}[Semantics]
    We define relations $\prefixstep{p}{e}{n}{e'}{p'}$ (as shown in
    Figures~\ref{fig:semantics1} and \ref{fig:semantics2}), and 
    $\argsstep{\eta}{A}{\Gamma}{n}{A'}{\eta'}$ (as shown in Figure~\ref{fig:args-semantics}).
    \end{definition}

    \begin{figure}[t]
    \centering

    \begin{mathpar}

    \inferrule[S-Eps-R]{ }{
        \prefixstep{\eta}{\epsTm}{n}{\epsTm}{\epsp}
    }

    \inferrule[S-One-R]{ }{
        \prefixstep{\eta}{\oneTm}{n}{\epsTm}{\onepB}
    }

    \infer[S-Var]{
        \eta(x) \mapsto p
    }{
        \prefixstep{\eta}{x}{n}{x}{p}
    }

        \ENDOFLINE

    \infer[S-Par-R]{
        \prefixstep{\eta}{e_1}{n_1}{e_1'}{p_1}\\
        \prefixstep{\eta}{e_2}{n_2}{e_2'}{p_2}
    }{
        \prefixstep{\eta}{\parpairTm{e_1}{e_2}}{n_1+n_2}{\parpairTm{e_1'}{e_2'}}{\parp{p_1}{p_2}}
    }

    \infer[S-Cat-R-1]{
        \prefixstep{\eta}{e_1}{n}{e_1'}{p}\\
        \neg \left(\isMaximal{p}\right)
    }{
        \prefixstep{\eta}{\catpairTm{e_1}{e_2}}{n}{\catpairTm{e_1'}{e_2}}{\catpA{p}}
    }

        \ENDOFLINE

    \infer[S-Cat-R-2]{
        \prefixstep{\eta}{e_1}{n_1}{e_1'}{p_1}\\
        \isMaximal{p_1}\\
        \prefixstep{\eta}{e_2}{n_2}{e_2'}{p_2}\\
    }{
        \prefixstep{\eta}{\catpairTm{e_1}{e_2}}{n_1+n_2}{e_2'}{\catpB{p_1}{p_2}}
    }

        \ENDOFLINE

    \infer[S-Par-L]{
        \eta(z) \mapsto \parp{p_1}{p_2}\\
        \prefixstep{\eta[x\mapsto p_1,y\mapsto p_2]}{e}{n}{e'}{p'}
    }{
        \prefixstep{\eta}{\letparTm{x}{y}{z}{e}}{n}{\letparTm{x}{y}{z}{e}}{p}
    }

        \ENDOFLINE

    \infer[S-Cat-L-1]{
        \eta(z) \mapsto \catpA{p}\\
        \prefixstep{\eta[x\mapsto p,y\mapsto \emp{t}]}{e}{n}{e'}{p'}
    }{
        \prefixstep{\eta}{\letcatTm{t}{x}{y}{z}{e}}{n}{\letcatTm{t}{x}{y}{z}{e'}}{p'}
    }

        \ENDOFLINE

    \infer[S-Cat-L-2]{
        \eta(z) \mapsto \catpB{p_1}{p_2}\\
        \prefixstep{\eta[x\mapsto p_1,y\mapsto p_2]}{e}{n}{e'}{p}
    }{
        \prefixstep{\eta}{\letcatTm{t}{x}{y}{z}{e}}{n}{\cutTm{x}{\sinkTm{p_1}}{e'[z/y]}}{p}
    }

    \infer[S-Plus-R-1]{
        \prefixstep{\eta}{e}{n}{e'}{p}
    }{
        \prefixstep{\eta}{\inl{e}}{n}{e'}{\inl{p}}
    }

    \infer[S-Plus-R-2]{
        \prefixstep{\eta}{e}{n}{e'}{p}
    }{
        \prefixstep{\eta}{\inr{e}}{n}{e'}{\inr{p}}
    }

        \ENDOFLINE
    \infer[S-Plus-L-1]{
        \prefixConcatRel{\eta'}{\eta}{\eta''}\\
        \eta''(z) = \sumpEmp
    }{
        \prefixstep{\eta}{\sumcaseTm{r}{\eta'}{z}{x}{e_1}{y}{e_2}}{n}{\sumcaseTm{r}{\eta''}{z}{x}{e_1}{y}{e_2}}{\emp{r}}
    }

        \ENDOFLINE
    \infer[S-Plus-L-2]{
        \prefixConcatRel{\eta'}{\eta}{\eta''}\\
        \eta''(z) = \sumpA{p}\\
        \prefixstep{\eta''[x \mapsto p]}{e_1}{n}{e_1'}{p'}
    }{
        \prefixstep{\eta}{\sumcaseTm{r}{\eta'}{z}{x}{e_1}{y}{e_2}}{n}{e_1'[z/x]}{p'}
    }

        \ENDOFLINE
    \infer[S-Plus-L-3]{
        \prefixConcatRel{\eta'}{\eta}{\eta''}\\
        \eta''(z) = \sumpB{p}\\
        \prefixstep{\eta''[y \mapsto p]}{e_2}{n}{e_2'}{p'}
    }{
        \prefixstep{\eta}{\sumcaseTm{r}{\eta'}{z}{x}{e_1}{y}{e_2}}{n}{e_2'[z/y]}{p'}
    }
    \end{mathpar}

    \caption{Semantics (part 1)}
    \label{fig:semantics1}
    \end{figure}

    \begin{figure}[t]
    \centering

    \begin{mathpar}
    \inferrule[S-Star-R-1]{ }{
        \prefixstep{\eta}{\nilTm}{n}{\epsTm}{\stpDone}
    }

    \infer[S-Star-R-2-1]{
        \prefixstep{\eta}{e_1}{n}{e_1'}{p}\\
        \neg \left(\isMaximal{p}\right)
    }{
        \prefixstep{\eta}{\consTm{e_1}{e_2}}{n}{\catpairTm{e_1'}{e_2}}{\stpA{p}}
    }

        \ENDOFLINE
    \infer[S-Star-R-2-2]{
        \prefixstep{\eta}{e_1}{n_1}{e_1'}{p_1}\\
        \isMaximal{p_1}\\
        \prefixstep{\eta}{e_2}{n_2}{e_2'}{p_2}\\
    }{
        \prefixstep{\eta}{\consTm{e_1}{e_2}}{n_1+n_2}{e_2'}{\stpB{p_1}{p_2}}
    }

        \ENDOFLINE
    \infer[S-Star-L-1]{
        \prefixConcatRel{\eta'}{\eta}{\eta''}\\
        \eta''(z) = \stpEmp
    }{
        \prefixstep{\eta}{\starcaseTm{s,r}{\eta'}{z}{e_1}{x}{xs}{e_2}}{n}{\starcaseTm{s,r}{\eta''}{z}{e_1}{x}{xs}{e_2}}{\emp{r}}
    }

        \ENDOFLINE
    \infer[S-Star-L-2]{
        \prefixConcatRel{\eta'}{\eta}{\eta''}\\
        \eta''(z) = \stpDone\\
        \prefixstep{\eta''}{e_1}{n}{e_1'}{p}
    }{
        \prefixstep{\eta}{\starcaseTm{s,r}{\eta'}{z}{e_1}{x}{xs}{e_2}}{n}{e_1'}{p}
    }

        \ENDOFLINE
    \infer[S-Star-L-3]{
        \prefixConcatRel{\eta'}{\eta}{\eta''}\\
        \eta''(z) = \stpA{p}\\
        \prefixstep{\eta''[x \mapsto p,y\mapsto \emp{s^\star}]}{e_2}{n}{e_2'}{p'}
    }{
        \prefixstep{\eta}{\starcaseTm{s,r}{\eta'}{z}{e_1}{x}{xs}{e_2}}{n}{\letcatTm{s^\star}{x}{y}{z}{e_2'}}{p'}
    }

        \ENDOFLINE
    \infer[S-Star-L-4]{
        \prefixConcatRel{\eta'}{\eta}{\eta''}\\
        \eta''(z) = \stpB{p}{p'}\\
        \prefixstep{\eta''[x \mapsto p,y\mapsto p']}{e_2}{n}{e_2'}{p''}
    }{
        \prefixstep{\eta}{\starcaseTm{s,r}{\eta'}{z}{e_1}{x}{xs}{e_2}}{n}{\cutTm{x}{\sinkTm{p}}{e_2'[z/xs]}}{p''}
    }

        \ENDOFLINE
    \infer[S-Let]{
        \prefixstep{\eta}{e_1}{n_1}{e_1'}{p}\\
        \prefixstep{\eta[x\mapsto p]}{e_2}{n_2}{e_2'}{p'}
    }{
        \prefixstep{\eta}{\cutTm{x}{e_1}{e_2}}{n_1+n_2}{\cutTm{x}{e_1'}{e_2'}}{p'}
    }

    \inferrule[S-HistPgm]{
        M \downarrow v\\
        p = \histValToPrefix{v}{s}
    }{
        \prefixstep{\eta}{\histPgmTm{M}{s}}{n}{\sinkTm{p}}{p}
    }

        \ENDOFLINE
    \infer[S-Wait-1]{
        \prefixConcatRel{\eta'}{\eta}{\eta''}\\
        \eta''(x) = p\\
        \neg \left(\isMaximal{p}\right)
    }{
        \prefixstep{\eta}{\waitTm{\eta'}{t}{x}{e}}{n}{\waitTm{\eta''}{t}{x}{e}}{\emp{t}}
    }

        \ENDOFLINE
    \infer[S-Wait-2]{
        \prefixConcatRel{\eta'}{\eta}{\eta''}\\
        \eta''(x) = p\\
        \isMaximal{p}\\
        \prefixstep{\eta''}{e[\flatten{p}/x]}{n}{e'}{p'}
    }{
        \prefixstep{\eta}{\waitTm{\eta'}{t}{x}{e}}{n}{e'}{p'}
    }

        \ENDOFLINE
    \infer[S-Fix]{
        \overline{M} \downarrow \theta\\
        \prefixstep{\eta}{\cutTm{\Gamma}{A}{\fixsubst{e}{e}[\theta]}}{n}{e'}{p}
    }{
        \prefixstep{\eta}{\fixTmHistargs{\overline{M}}{A}{e}}{n+1}{e'}{p}
    }

    \infer[S-ArgsLet]{
        \argsstep{\eta}{A}{\Gamma}{n_1}{A'}{\eta'}\\
        \prefixstep{\eta'}{e}{n_2}{e'}{p}\\
        \derivrel{\eta'}{\Gamma}{\Gamma'}
    }{
        \prefixstep{\eta}{\cutTm{\Gamma}{A}{e}}{n_1 + n_2}{\cutTm{\Gamma'}{A'}{e'}}{p}
    }

    \end{mathpar}

    \caption{Semantics (part 2)}
    \label{fig:semantics2}
    \end{figure}

    \begin{figure}[t]
    \centering

    \caption{Arguments Semantics}
    \begin{mathpar}

        \INFERRULE[S-Args-Emp]{ }{
            \argsstep{\eta}{\cdot}{\cdot}{n}{\cdot}{\{\}}
        }

        \INFERRULE[S-Args-Sng]{
            \prefixstep{\eta}{e}{n}{e'}{p}
        }{
            \argsstep{\eta}{e}{\left(x : s\right)}{n}{e'}{\left\{x \mapsto p\right\}}
        }

        \INFERRULE[S-Args-Comma]{
            \argsstep{\eta}{A_1}{\Gamma}{n_1}{A_1'}{\eta_1}\\
            \argsstep{\eta}{A_2}{\Gamma'}{n_2}{A_2'}{\eta_2}
        }{
            \argsstep{\eta}{\parpairArgsTm{A_1}{A_2}}{\commactx{\Gamma}{\Gamma'}}{n_1+n_2}{\parpairArgsTm{A_1'}{A_2'}}{\eta_1 \cup \eta_2}
        }

        \INFERRULE[S-Args-Semic-1-1]{
            \argsstep{\eta}{A_1}{\Gamma}{n_1}{A_1'}{\eta_1}\\
            \neg\left(\maximalOn{\eta_1}{\Gamma}\right)\\
        }{
            \argsstep{\eta}{\catpairArgsTmA{A_1}{A_2}}{\semicctx{\Gamma}{\Gamma'}}{n_1}{\catpairArgsTmA{A_1'}{A_2}}{\eta_1 \cup \emp{\Gamma'}}
        }

        \INFERRULE[S-Args-Semic-1-2]{
            \argsstep{\eta}{A_1}{\Gamma}{n_1}{A_1'}{\eta_1}\\
            \maximalOn{\eta_1}{\Gamma}\\
            \argsstep{\eta}{A_2}{\Gamma'}{n_2}{A_2'}{\eta_2}\\
        }{
            \argsstep{\eta}{\catpairArgsTmA{A_1}{A_2}}{\semicctx{\Gamma}{\Gamma'}}{n_1+n_2}{\catpairArgsTmB{}{A_2'}}{\eta_1 \cup \eta_2}
        }

        \INFERRULE[S-Args-Semic-2]{
            \argsstep{\eta}{A}{\Gamma'}{n}{A'}{\eta'}
        }{
            \argsstep{\eta}{\catpairArgsTmB{A}}{\semicctx{\Gamma}{\Gamma'}}{n}{\catpairArgsTmB{A'}}{\emp{\Gamma} \cup \eta'}
        }

    \end{mathpar}
    \label{fig:args-semantics}
    \end{figure}

\subsubsection*{Recursive Argument Semantics}
The arguments semantics $\argsstep{\eta}{A}{\Gamma}{n}{A'}{\eta'}$ accepts an
environment $\eta$ and runs it through $A$ to produce an environment
$\envHasType{\eta'}{\Gamma}$. This relation is essentially the same as evaluating
a large nested tree \ruleName{T-Cat-R} \ruleName{T-Par-R} terms, structured
like the context $\Gamma$. The only difference is that, because context
derivatives do not remove the left component of a semicolon context (the
$\Gamma$ in $\semicctx{\Gamma}{\Delta}$) after a maximal prefix has arrived, we have a special term former $\catpairArgsTmB{A'}$
for cat-pair terms $\catpairArgsTmA{A}{A'}$ that have crossed over.
The context is required in the semantics so we can compute the empty environment in \ruleName{S-Args-Semic-1-1} and \ruleName{S-Args-Semic-2}.

\subsubsection*{Semantics of Buffering}
The semantics for \ruleName{Plus-L} and\ruleName{Star-L}
and \ruleName{Wait} buffer in their inputs until enough
of the input has arrived to run the term, where the particular value of ``enough''
depends on the rule in question.

To illustrate, consider the rules for Wait (\ruleName{S-Wait-1} and \ruleName{S-Wait-2} in Figure~\ref{fig:semantics2}).
In both cases, we take the incoming environment $\eta$, and concatenate it onto
the buffer $\eta'$, to get the combined
prefix $\eta''$. We then dispatch on whether $\eta''$ is enough input to run the
continuation $e$. In this case, ``enough'' means that $\eta''$ contains a
maximal prefix $p$ of $x : s$. If it does (\ruleName{P-Wait-2}), we run the continuation,
substituting the maximal prefix in for the (historical) occurrences of $x$. If it does not, we simply save $\eta''$ as the
new buffer in the resulting $\texttt{wait}$ term, and return the empty prefix in \ruleName{P-Wait-1}.

The semantics for \ruleName{Plus-L} and \ruleName{Star-L} are similar:
in all cases, we add the incoming prefix to the buffer, and then project
from the buffer. If not enough data has arrived, we return the empty prefix
and step to the same term but with an updated buffer.

\subsubsection*{Maximal Semantics Theorem}
If all input prefixes are maximal and the step terminates, then the output
prefixes are maximal. The contrapositive of this fact is crucial: if the output
of a step is not maximal, than some stream in the input must still be sending more
data.

\begin{lemma}[Maximal Semantics Auxiliary]
\label{lemma:batch-sem-aux}
$\,$
\begin{enumerate}
    \item If $\prefixstep{\eta}{e}{ }{e'}{p}$ and $\maximalOn{\eta}{e}$
we have that $\isMaximal{p}$.
    \item If $\argsstep{\eta}{A}{\Gamma}{ }{e'}{\eta'}$ and $\maximalOn{\eta}{A}$ then $\maximalOn{\eta'}{\Gamma}$
\end{enumerate}

\end{lemma}
    \jtheorem{}{
        By mutual induction on the derivation of $\prefixstep{\eta}{e}{ }{e'}{p}$
        and $\argsstep{\eta}{A}{\Gamma}{ }{e'}{\eta'}$.

        \jcase{1}{S-Var}{Immediate.}

        \jcase{2}{S-Eps-R}{Immediate.}

        \jcase{3}{S-One-R}{Immediate.}

        \jcase{4}{S-Par-R}{
            \jgivengoal{
                \caseFact{1} $\prefixstep{\eta}{e_1}{ }{e_1'}{p_1}$

                \caseFact{2} $\prefixstep{\eta}{e_2}{ }{e_2'}{p_2}$

                \caseFact{3} $\maximalOn{\eta}{\parpairTm{e_1}{e_2}}$
            }{
                $\isMaximal{\parp{p_1}{p_2}}$
            }

            \caseText{By (3) and the definition of $\text{fv}$}

            \caseFact{4} For $i=1,2$, for all $x \in \text{fv}(e_i)$, there is some $\isMaximal{p}$ such that $\eta(x) \mapsto p$.

            \caseText{By IH on (1)}

            \caseFact{5} $p_1$ is maximal

            \caseText{By IH on (2)}

            \caseFact{6} $p_1$ is maximal

            \caseText{The goal follows by (5) and (6).}
        }

        \jcase{5}{S-Cat-R-1}{
            \jgivengoal{
                \caseFact{1} $\prefixstep{\eta}{e_1}{ }{e_1'}{p_1}$

                \caseFact{2} $\neg\left(\isMaximal{p_1}\right)$

                \caseFact{3} $\maximalOn{\eta}{\catpairTm{e_1}{e_2}}$
            }{
                $\isMaximal{\catpA{p}}$
            }

            \caseText{By (3) and the definition of $\text{fv}$}

            \caseFact{4} For all $x \in \text{fv}(e_1)$, there is some $\isMaximal{p}$ such that $\eta(x) \mapsto p$.

            \caseText{By IH on (1)}

            \caseFact{5} $p_1$ is maximal

            \caseText{(2) and (5) are a contradiction}
        }

        \jcase{6}{S-Cat-R-2}{
            \jgivengoal{
                \caseFact{1} $\prefixstep{\eta}{e_1}{ }{e_1'}{p_1}$

                \caseFact{2} $\isMaximal{p_1}$

                \caseFact{3} $\prefixstep{\eta}{e_2}{ }{e_2'}{p_2}$

                \caseFact{5} $\maximalOn{\eta}{e_i}$ for $i=1,2$.
            }{
                $\isMaximal{\catpB{p_1}{p_2}}$
            }

            \caseText{By (4) and the definition of $\text{fv}$}

            \caseFact{5} For all $i=1,2$, for all $x \in \text{fv}(e_1)$, there is some $\isMaximal{p}$ such that $\eta(x) \mapsto p$.

            \caseText{By IH on (3)}

            \caseFact{6} $\isMaximal{p_2}$

            \caseText{The conclusion follows by (2) and (6)}
        }

        \jcase{7}{S-Par-L}{
            \jgivengoal{
                \caseFact{1} $\eta(z) \mapsto \parp{p_1}{p_2}$

                \caseFact{2} $\prefixstep{\eta[x\mapsto p_1,y\mapsto p_2]}{e}{ }{e'}{p}$

                \caseFact{3} $\maximalOn{\eta}{\letparTm{x}{y}{z}{e}}$
            }{
                $\isMaximal{p}$
            }
            \caseText{By (1) and (3), using the fact that $\text{fv}(\letparTm{x}{y}{z}{e}) = \{z\} \cup \text{fv}(e) \setminus \{x,y\}$}

            \caseFact{4} $\isMaximal{\parp{p_1}{p_2}}$

            \caseText{By inversion on (4)}

            \caseFact{5} $\isMaximal{p_1}$

            \caseFact{6} $\isMaximal{p_2}$

            \caseText{By (3), (5), and (6)}

            \caseFact{7} For all $u \in \text{fv}(e)$, there is some $\isMaximal{p'}$ such that $\eta[x \mapsto p_1,y\mapsto p_2](u) \mapsto p'$.

            \caseText{The conclusion follows by IH on (2), using (7)}
        }

        \jcase{8}{S-Cat-L-1}{
            \jgivengoal{
                \caseFact{1} $\eta(z) \mapsto \catpA{p}$

                \caseFact{2} $\prefixstep{\eta[x\mapsto p_1,y\mapsto \emp{t}]}{e}{ }{e'}{p'}$

                \caseFact{3} $\maximalOn{\eta}{\letcatTm{t}{x}{y}{z}{e}}$
            }{
                $\isMaximal{p'}$
            }
            \caseText{By (3), since $z \in \text{fv}(\letcatTm{t}{x}{y}{z}{e})$, we have}

            \caseFact{4} $\isMaximal{\catpA{p}}$

            \caseText{But (4) is a contradiction, so the conclusion follows.}
        }

        \jcase{9}{S-Cat-L-2}{
            \jgivengoal{
                \caseFact{1} $\eta(z) \mapsto \catpB{p_1}{p_2}$

                \caseFact{2} $\prefixstep{\eta[x\mapsto p_1,y\mapsto p_2]}{e}{ }{e'}{p'}$

                \caseFact{3} $\maximalOn{\eta}{\letcatTm{t}{x}{y}{z}{e}}$
            }{
                $\isMaximal{p'}$
            }

            \caseText{By (3), since $z \in \text{fv}(\letcatTm{t}{x}{y}{z}{e})$, we have}

            \caseFact{4} $\isMaximal{\catpB{p_1}{p_2}}$

            \caseText{Inverting (4)}

            \caseFact{5} $\isMaximal{p_1}$

            \caseFact{6} $\isMaximal{p_2}$

            \caseText{By (2), (5), and (6)}

            \caseFact{7} For all $u \in \text{fv}(e)$, there is some $\isMaximal{p'}$ such that $\eta[x \mapsto p_1,y\mapsto p_2](u) \mapsto p'$.

            \caseText{The goal follows by IH on (2), using (7)}
        }

        \jcase{10}{S-Plus-R-1}{
            \caseText{Immediate by IH, using the fact that $\text{fv}(\inr{e}) = \text{fv}(e)$}
        }

        \jcase{11}{S-Plus-R-2}{
            \caseText{Identical to previous}
        }

        \jcase{12}{S-Plus-L-1}{
            \jgivengoal{
                \caseFact{1} $\prefixConcatRel{\eta'}{\eta}{\eta''}$

                \caseFact{2} $\eta''(z) = \sumpEmp$

                \caseFact{3} $\maximalOn{\eta}{{\sumcaseTm{z}{\eta'}{r}{x}{e_1}{y}{e_2}}}$

            }{
                $\isMaximal{\emp{r}}$
            }

            \caseText{By (3), there exists some $p$ such that:}

            \caseFact{4} $\isMaximal{p}$

            \caseFact{5} $\eta(z) = p$.

            \caseText{By Theorem~\ref{thm:env-concat-maximal} on (1), (4) and (5)}

            \caseFact{6} $\isMaximal{\eta''(z)}$

            \caseText{But (2) and (6) are contradictory, since $\sumpEmp$ is not maximal.}
        }

        \jcase{13}{S-Plus-L-2}{
            \jgivengoal{
                \caseFact{1} $\prefixConcatRel{\eta'}{\eta}{\eta''}$

                \caseFact{2} $\eta''(z) = \sumpA{p}$

                \caseFact{3} $\prefixstep{\eta''[x \mapsto p]}{e_1}{ }{e_1'}{p'}$

                \caseFact{4} $\maximalOn{\eta}{{\sumcaseTm{z}{\eta'}{r}{x}{e_1}{y}{e_2}}}$

            }{
                \isMaximal{p'}
            }
                 \caseText{By Theorem~\ref{thm:env-concat-maximal} on
                   (1) and (4), we have}

                \caseFact{5} For all $x \in \text{fv}({\sumcaseTm{z}{\eta'}{r}{x}{e_1}{y}{e_2}})$, there is some $\isMaximal{p}$ such that $\eta''(x) \mapsto p$.

                \caseText{By (5) and (2)}

                \caseFact{6} $\isMaximal{\sumpA{p}}$

                \caseText{Inverting (6)}

                \caseFact{7} $\isMaximal{p}$

                \caseText{By (5) and (7)}

                \caseFact{8} For all $u \in \text{fv}(e_1)$, there is some $\isMaximal{p'}$ such that $\eta''[x \mapsto p](u) \mapsto p'$.

                \caseText{The goal follows immediately by IH.}
        }

        \jcase{14}{S-Plus-L-3}{
            \caseText{Identical to previous}
        }

        \jcase{15}{S-Star-R-1}{Immediate.}

        \jcase{16}{S-Star-R-2-1}{Identical to S-Cat-R-1}

        \jcase{17}{S-Star-R-2-2}{Identical to S-Cat-R-2}

        \jcase{18}{S-Star-L-1}{Identical to S-Plus-L-1}

        \jcase{19}{S-Star-L-2}{
            \jgivengoal{
                \caseFact{1} $\prefixConcatRel{\eta'}{\eta}{\eta''}$

                \caseFact{2} $\eta''(z) = \stpDone$

                \caseFact{3} $\prefixstep{\eta''}{e_1}{n}{e_1'}{p}$

                \caseFact{4} $\maximalOn{\eta}{{\starcaseTm{s,r}{\eta'}{z}{e_1}{x}{xs}{e_2}}}$
            }{
                \isMaximal{p}
            }

            \caseText{By Theorem~\ref{thm:env-concat-maximal} on (1) and (4), and specializing to $\text{fv}(e_1)$ have}

            \caseFact{5} For all $x \in \text{fv}(e_1)$, there is some $\isMaximal{p}$ so that $\eta''(x) = p$.

            \caseText{The goal follows immediately by (5) and IH on (3).}

        }

        \jcase{20}{S-Star-L-3}{
            \jgivengoal{
                \caseFact{1} $\prefixConcatRel{\eta'}{\eta}{\eta''}$

                \caseFact{2} $\eta''(z) = \stpA{p}$

                \caseFact{3} $\prefixstep{\eta''}{e_1}{n}{e_1'}{p}$

                \caseFact{4} $\maximalOn{\eta}{{\starcaseTm{s,r}{\eta'}{z}{e_1}{x}{xs}{e_2}}}$
            }{
                \isMaximal{p}
            }
                \caseText{By Theorem~\ref{thm:env-concat-maximal} on (1) and (4), we have}

                \caseFact{5} $\maximalOn{\eta''}{{\starcaseTm{s,r}{\eta'}{z}{e_1}{x}{xs}{e_2}}}$

                \caseText{In particular, $\isMaximal{\eta''(z)}$, but this is a contradiction with (2).}
        }

        \jcase{21}{S-Star-L-3}{Identical to S-Plus-L-2}

        \jcase{22}{S-Let}{
            \jgivengoal{
                \caseFact{1} $\prefixstep{\eta}{e_1}{n_1}{e_1'}{p}$

                \caseFact{2} $\prefixstep{\eta[x\mapsto p]}{e_2}{n_2}{e_2'}{p'}$

                \caseFact{4} $\maximalOn{\eta}{\cutTm{x}{e_1}{e_2}}$
            }{
                \isMaximal{p'}
            }
            \caseText{By IH on (1), using the fact that $\text{fv}(e_1) \subseteq \text{fv}(\cutTm{x}{e_1}{e_2})$}, we have

            \caseFact{4} $\isMaximal{p}$

            \caseText{By (3), using the fact that $\text{fv}(e_2) \setminus \{x\} \subseteq \text{fv}(\cutTm{x}{e_1}{e_2})$}

            \caseFact{5} For all $y \in \text{fv}(e_2)$, there exists $\isMaximal{p'}$ so that $\eta[x \mapsto p](y) \mapsto p'$

            \caseText{The goal follows immediately y IH on (2) and (5)}
        }

        \jcase{23}{S-HistPgm}{Immediate by Definition~\ref{def:type-ctx-flatten}}

        \jcase{24}{S-Wait-1}{Identical to Plus-L-1.}

        \jcase{24}{S-Wait-2}{
            \jgivengoal{
                \caseFact{1} $\prefixConcatRel{\eta'}{\eta}{\eta''}$

                \caseFact{2} $\eta''(x) = p$

                \caseFact{3} $\isMaximal{p}$

                \caseFact{4} $\prefixstep{\eta''}{e[\flatten{p}/x]}{n}{e'}{p'}$

                \caseFact{4} $\maximalOn{\eta}{\waitTm{\eta}{t}{x}{e}}$
            }{
                \isMaximal{p'}
            }

            \caseText{By Theorem~\ref{thm:env-concat-maximal} on (5) and (1):}

            \caseFact{6} For all $y \in \text{fv}(\waitTm{\eta}{t}{x}{e})$, there is some $\isMaximal{p}$ such that $\eta''(y)\mapsto p$.

            \caseText{But since $\text{fv}(e[\flatten{p}/x]) \subset \text{fv}(\waitTm{\eta}{t}{x}{e})$, we have:}

            \caseFact{7} For all $y \in \text{fv}(e[\flatten{p}/x])$, there is some $\isMaximal{p}$ such that $\eta''(y)\mapsto p$.

            \caseText{The goal follows immediately by IH on (4) with (7)}
        }

        \jcase{25}{S-Fix}{
            \jgivengoal{
                \caseFact{1} $\overline{M} \downarrow \theta$

                \caseFact{2} $\prefixstep{\eta}{\cutTm{\Gamma}{A}{\fixsubst{e}{e}[\theta]}}{n}{e'}{p}$

                \caseFact{3} $\maximalOn{\eta}{\fixTmHistargs{\overline{M}}{A}{e}}$
            }{
                $\isMaximal{p}$
            }
            \caseText{By noting that $\text{fv}(\fixTmHistargs{\overline{M}}{A}{e}) = \text{fv}(A)$ and $\text{fv}({\cutTm{\Gamma}{A}{\fixsubst{e}{e}[\theta]}}) = \text{fv}(A)$:}

            \caseText{4} $\maximalOn{\eta}{{\cutTm{\Gamma}{A}{\fixsubst{e}{e}[\theta]}}}$

            \caseText{Goal follows immediately by IH on (1) and (3)}
        }

        \jcase{26}{S-ArgsLet}{Immediate by two uses of IH.}

        \jcase{27}{S-Args-Emp}{Immediate.}
        \jcase{28}{S-Args-Sng}{Immediate by IH}
        \jcase{29}{S-Args-Comma}{Same as \ruleName{S-Par-R}}
        \jcase{30}{S-Args-Semic-1-1}{
            \jgivengoal{
                \caseFact{1} $\argsstep{\eta}{A_1}{\Gamma}{n_1}{A_1'}{\eta_1}$

                \caseFact{2} $\neg\left(\maximalOn{\eta_1}{\Gamma}\right)$

                \caseFact{3} $\maximalOn{\eta}{\catpairArgsTmA{A_1}{A_2}}$
            }{
                $\maximalOn{\eta_1 \cup \emp{\Gamma'}}{\commactx{\Gamma}{\Gamma'}}$

            }
                \caseText{By (3)}

                \caseFact{4} $\maximalOn{\eta}{A_1}$

                \caseText{By IH on (1) and (4)}

                \caseFact{5} $\maximalOn{\eta_1}{\Gamma}$

                \caseText{(2) and (5) are a contradiction.}
        }
        \jcase{31}{S-Args-Semic-1-2}{
            \jgivengoal{
                
                \caseFact{1} $\argsstep{\eta}{A_1}{\Gamma}{n_1}{A_1'}{\eta_1}$

                \caseFact{2} $\maximalOn{\eta_1}{\Gamma}$

                \caseFact{3} $\argsstep{\eta}{A_2}{\Gamma'}{n_2}{A_2'}{\eta_2}$

                \caseFact{4} $\maximalOn{\eta}{\catpairArgsTmA{A_1}{A_2}}$
            }{
                $\maximalOn{\eta_1 \cup \eta_2}{\commactx{\Gamma}{\Gamma'}}$
            }

            \caseText{By IH on (3) and (4)}

            \caseFact{5} $\maximalOn{\eta_2}{\Gamma'}$

            \caseText{The goal follows by (2) and (5)}
        }

        \jcase{32}{S-Args-Semic-2-1}{Immediate by IH.}
    }

\begin{theorem}[Maximal Semantics]
\label{thm:batch-sem}
Suppose:
\begin{enumerate}
    \item $\prefixstep{\eta}{e}{n}{e'}{p}$
    \item $\tck{\cdot}{\Gamma}{\emptyset}{e}{s}{i}$
    \item $\maximalOn{\eta}{\Gamma}$
\end{enumerate}
Then, $\isMaximal{p}$.
\end{theorem}
\begin{proof}
    Because $\tck{\cdot}{\Gamma}{\emptyset}{e}{s}${i}, we have that $\text{fv}(e) \subseteq \dom{\Gamma}$. Thus,
    $\maximalOn{\eta}{e}$, and so the goal follows by Lemma~\ref{lemma:batch-sem-aux}
\end{proof}

\begin{theorem}[Maximal Semantics Extension]
\label{thm:max-sem-extend}
$\,$
\begin{enumerate}
    \item If $\prefixstep{\eta}{e}{n}{e'}{p}$ and $\isMaximal{p}$ and $\prefixConcatRel{\eta}{\eta'}{\eta''}$, and $\prefixstep{\eta''}{e}{n'}{e''}{p'}$,
then $p = p'$.
    \item If $\argsstep{\eta_1}{A}{\Gamma}{n}{A'}{\eta_2}$ and $\maximalOn{\eta_2}{\Gamma}$ and $\prefixConcatRel{\eta_1}{\eta_1'}{\eta_1''}$, and $\argsstep{\eta_1''}{A}{\Gamma}{n'}{A''}{\eta_2'}$,
then $\eta_2|_\Gamma = \eta_2|_\Gamma$.
\end{enumerate}

\end{theorem}
\begin{proof}
    Mutual induction on $\prefixstep{\eta}{e}{n}{e'}{p}$ and $\argsstep{\eta_1}{A}{\Gamma}{n}{A'}{\eta_2}$, using Theorem~\ref{thm:concat-maximal}.
\end{proof}

\subsubsection*{Semantics Theorems}
\label{app:incr-sem-thms}

\begin{theorem}[Semantics Inputs Determine Outputs]
    \label{thm:prefix-sem-det}
$\,$
\begin{enumerate}
    \item If $\prefixstep{\eta}{e}{n}{e'}{p'}$ and
    $\prefixstep{\eta}{e}{n'}{e''}{p''}$, then
    $e' = e''$, and $p' = p''$.
    \item If $\argsstep{\eta}{A}{\Gamma}{n}{A'}{\eta'}$
    and
    $\argsstep{\eta}{A}{\Gamma}{n'}{A''}{\eta''}$
    then $A' = A''$ and $\eta' = \eta''$.
\end{enumerate}
\end{theorem}
\begin{proof}
    By inspection.
\end{proof}

\begin{theorem}[Semantics Monotonicity]
    \label{thm:prefix-sem-mono}
$\,$
\begin{enumerate}
    \item If $\prefixstep{\eta}{e}{n}{e'}{p'}$ and $n' \geq n$, then
    $\prefixstep{\eta}{e}{n'}{e'}{p'}$.
    \item If $\argsstep{\eta}{A}{\Gamma}{n}{A'}{\eta'}$ and $n' \geq n$,
    then $\argsstep{\eta}{A}{\Gamma}{n'}{A'}{\eta'}$.
\end{enumerate}
\end{theorem}
\begin{proof}
    Mutual induction.
\end{proof}

\begin{theorem}[Soundness]
    \label{thm:prefix-pres}
$\,$
\begin{enumerate}
    \item Suppose
    \begin{enumerate}
        \item $\tck{\cdot}{\Gamma}{\emptyset}{e}{s}{i}$
        \item $\prefixstep{\eta}{e}{n}{e'}{p}$
        \item $\envHasType{\eta}{\Gamma}$
    \end{enumerate}
    Then,
    \begin{enumerate}
        \item $\prefixHasType{p}{s}$
        \item If $\derivrel{\eta}{\Gamma}{\Gamma'}$ and $\derivrel{p}{s}{s'}$, then $\tck{\cdot}{\Gamma'}{\emptyset}{e'}{s'}{\I}$
        \item If $i = \I$ and $\emptyOn{\eta}{e}$ then $\isEmpty{p}$
    \end{enumerate}
    \item Suppose
    \begin{enumerate}
        \item $\tck{\cdot}{\Gamma_1}{\emptyset}{A}{\Gamma_1}{i}$
        \item $\argsstep{\eta}{A}{\Gamma_2}{n}{A'}{\eta'}$
        \item $\envHasType{\eta}{\Gamma_1}$
    \end{enumerate}
    Then,
    \begin{enumerate}
        \item $\prefixHasType{\eta'}{\Gamma_2}$ 
        \item If $\derivrel{\eta}{\Gamma_1}{\Gamma_1'}$ and $\derivrel{\eta'}{\Gamma_2}{\Gamma_2'}$, then $\tck{\cdot}{\Gamma_1'}{\emptyset}{A'}{\Gamma_2'}{\I}$
        \item If $i = \I$ and $\emptyOn{\eta}{A}$ then $\emptyOn{\eta'}{\Gamma_2}$
    \end{enumerate}
\end{enumerate}
    
\end{theorem}
\jtheorem{ }{
    By mutual induction on the semantics. In the cases for the term (non-argument) semantics, we also do an inner induction on
    the typing derivation $\tck{\cdot}{\Gamma}{\emptyset}{e}{s}{i}$. All of these inner inductions have two cases: one for the corresponding
    syntax-directed rule, and one for \ruleName{T-Sub}. We handle all of the cases with \ruleName{T-Sub} simultaneously, in the first case of this proof.

    \jcase{1}{T-Sub}{
        \jgivengoalThree{
            \caseFact{1} $\tck{\cdot}{\Delta}{\emptyset}{e}{s}{i}$

            \caseFact{2} $\subty{\Gamma}{\Delta}$

            \caseFact{3} $\prefixstep{\eta}{e}{ }{e'}{p}$

            \caseFact{4} $\envHasType{\eta}{\Gamma}$
        }{
            $\prefixHasType{p}{s}$
        }{
            If $\derivrel{\eta}{\Gamma}{\Gamma'}$ and $\derivrel{p}{s}{s'}$, then $\tck{\cdot}{\Gamma'}{\emptyset}{e'}{s'}{\I}$
        }{
            If $i = \I$ and $\emptyOn{\eta}{e}$ then $\isEmpty{p}$
        }

        \caseText{By Theorem~\ref{thm:subty-env} on (2) and (4)}

        \caseFact{5} $\envHasType{\eta}{\Delta}$

        \caseText{By IH on (1), using (3) and (5)}

        \caseFact{6} $\prefixHasType{p}{s}$

        \caseFact{7} If $\derivrel{\eta}{\Delta}{\Delta'}$ and $\derivrel{p}{s}{s'}$, then $\tck{\cdot}{\Delta'}{\emptyset}{e'}{s'}{\I}$

        \caseFact{8} If $i = I$ and $\emptyOn{\eta}{e}$ then $\isEmpty{p}$

        \caseText{\textbf{Goal A} is complete by (6), and \textbf{Goal C} by (8). To prove \textbf{Goal B}, we suppose there are $\Gamma'$ and $s'$ so that:}

        \caseFact{9} $\derivrel{\eta}{\Gamma}{\Gamma'}$

        \caseFact{10} $\derivrel{p}{s}{s'}$

        \caseText{By Theorem~\ref{thm:derivrel-env-fun} on (5), there is some $\Delta'$ so that:}

        \caseFact{11} $\derivrel{\eta}{\Delta}{\Delta'}$

        \caseText{By (7), (10), and (11) we have:}

        \caseFact{12} $\tck{\cdot}{\Delta'}{\emptyset}{e'}{s'}{\I}$

        \caseText{By Theorem~\ref{thm:deriv-subty},}

        \caseFact{13} $\subty{\Gamma'}{\Delta'}$

        \caseText{\textbf{Goal B} follows by T-Sub on (12) and (13)}
    }

    \jcase{2}{S-Eps-R}{Immediate.}

    \jcase{3}{S-One-R}{Immediate.}

    \jcase{4}{S-Var}{
        \jgivengoalThree{
            \caseFact{1} $\tck{\cdot}{\Gamma(x:s)}{\emptyset}{x}{s}{i}$

            \caseFact{2} $\envHasType{\eta}{\Gamma(x:s)}$

            \caseFact{3} $\eta(x) \mapsto p$
        }{
            $\prefixHasType{p}{s}$
        }{
            If $\derivrel{\eta}{\Gamma(x:s)}{\Gamma'}$ and $\derivrel{p}{s}{s'}$, then $\tck{\cdot}{\Gamma'}{\emptyset}{x}{s'}{\I}$
        }{
            If $i = \I$ and $\emptyOn{\eta}{x}$ then $\isEmpty{p}$
        }

        \caseText{\textbf{Goal A} follows immediately by Theorem~\ref{thm:env-lookup} on (2) and (3). \textbf{Goal C} is immediate. For \textbf{Goal B}, we assume:}

        \caseFact{4} $\derivrel{\eta}{\Gamma(x:s)}{\Gamma'}$

        \caseFact{5} $\derivrel{p}{s}{s'}$

        \caseText{By Theorem~\ref{thm:env-lookup-deriv} on (3), (4), and (5), there is some $\Gamma''(-)$ so that:}

        \caseFact{6} $\Gamma' = \Gamma''(x : s')$

        \caseText{Then \textbf{Goal B} follows by (6) and T-Var.}

    }

    \jcase{4}{S-Par-R}{
        \jgivengoalThree{
            \caseFact{1} $\tck{\cdot}{\Gamma}{\emptyset}{e_1}{s}{i}$

            \caseFact{2} $\tck{\cdot}{\Gamma}{\emptyset}{e_2}{t}{i}$

            \caseFact{3} $\envHasType{\eta}{\Gamma}$

            \caseFact{4} $\prefixstep{\eta}{e_1}{ }{e_1'}{p_1}$

            \caseFact{5} $\prefixstep{\eta}{e_2}{ }{e_2'}{p_2}$
        }{
            $\prefixHasType{\parp{p_1}{p_2}}{s \| t}$
        }{
            If $\derivrel{\eta}{\Gamma}{\Gamma'}$ and $\derivrel{\parp{p_1}{p_2}}{s \| t}{s_0}$, then $\tck{\cdot}{\Gamma_0}{\emptyset}{\parpairTm{e_1'}{e_2'}}{s_0}{\I}$
        }{
            If $i = \I$ and $\emptyOn{\eta}{\parpairTm{e_1}{e_2}}$ then $\isEmpty{\parp{p_1}{p_2}}$
        }

        \caseText{By IH on (1), (3), and (4):}

        \caseFact{6} $\prefixHasType{p_1}{s}$

        \caseFact{7} If $\derivrel{\eta}{\Gamma}{\Gamma'}$ and $\derivrel{p_1}{s}{s'}$, then $\tck{\cdot}{\Gamma'}{\emptyset}{e_1'}{s'}{\I}$

        \caseFact{8} If $i = \I$ and $\emptyOn{\eta}{e_1}$ then $\isEmpty{p_1}$

        \caseText{By IH on (2), (4), and (5):}

        \caseFact{9} $\prefixHasType{p_2}{t}$

        \caseFact{10} If $\derivrel{\eta}{\Gamma}{\Gamma'}$ and $\derivrel{p_2}{t}{t'}$, then $\tck{\cdot}{\Gamma'}{\emptyset}{e_2'}{t'}$

        \caseFact{11} If $i = \I$ and $\emptyOn{\eta}{e_1}$ then $\isEmpty{p_1}$

        \caseText{\textbf{Goal A} follows by (6) and (9). \textbf{Goal C} follows by (8) and (11). For \textbf{Goal B}, we assume:}

        \caseFact{12} $\derivrel{\eta}{\Gamma}{\Gamma'}$

        \caseFact{13} $\derivrel{\parp{p_1}{p_2}}{s \| t}{s_0}$

        \caseText{Inverting (13), we have $s'$, and $t'$ so that $s_0 = s' \| t'$, and:}

        \caseFact{14} $\derivrel{p_1}{s}{s'}$

        \caseFact{15} $\derivrel{p_2}{t}{t'}$

        \caseText{By applying (7) to (12) and (14), we have:}

        \caseFact{16} $\tck{\cdot}{\Gamma'}{\emptyset}{e_1'}{s'}{\I}$

        \caseText{By applying (10) to (12) and (15), we have:}

        \caseFact{17} $\tck{\cdot}{\Gamma'}{\emptyset}{e_2'}{t'}{\I}$

        \caseText{\textbf{Goal B} follows by T-Par-R on (16) and (18)}
    }

    \jcase{5}{S-Cat-R-1}{
        \jgivengoalThree{
            \caseFact{1} $\prefixstep{\eta}{e_1}{ }{e_1'}{p}$

            \caseFact{2} $\neg \left(\isMaximal{p}\right)$

            \caseFact{3} $\envHasType{\eta}{\Gamma}$

            \caseFact{4} $\envHasType{\eta}{\Delta}$

            \caseFact{5} $(\exists x \in \dom{\Delta}. \neg \isEmpty{\eta(x)}) \implies (\forall x \in \dom{\Gamma}. \isMaximal{\eta(x)})$

            \caseFact{6} $\tck{\cdot}{\Gamma}{\emptyset}{e_1}{s}{i_1}$

            \caseFact{7} $\tck{\cdot}{\Delta}{\emptyset}{e_2}{t}{i_2}$

            \caseFact{8} $i_3 = \I \implies i_1 = \I \wedge \neg\left(\nullable{s}\right)$

        }{
            $\prefixHasType{\catpA{p}}{s \cdot t}$
        }{
            If $\derivrel{\eta}{\semicctx{\Gamma}{\Delta}}{\Gamma_0}$ and $\derivrel{\catpA{p}}{s \cdot t}{s_0}$, then $\tck{\cdot}{\Gamma_0}{\emptyset}{\catpairTm{e_1'}{e_2}}{s_0}{\I}$
        }{
            If $i_3 = \I$ and $\emptyOn{\eta}{\catpairTm{e_1}{e_2}}$, then $\isEmpty{\catpA{p}}$
        }

        \caseText{By IH on (1), (3), and (6)}

        \caseFact{9} $\prefixHasType{p}{s}$

        \caseFact{10} If $\derivrel{\eta}{\Gamma}{\Gamma'}$ and $\derivrel{p}{s}{s'}$, then $\tck{\cdot}{\Gamma'}{\emptyset}{e_1'}{s'}{\I}$

        \caseFact{11} If $i_1 = \I$ and $\emptyOn{\eta}{e_1}$ then $\isEmpty{p}$

        \caseText{\textbf{Goal A} follows by (9). For \textbf{Goal B}, we assume:}

        \caseFact{12} $\derivrel{\eta}{\semicctx{\Gamma}{\Delta}}{\Gamma_0}$

        \caseFact{13} $\derivrel{\catpA{p}}{s \cdot t}{s_0}$

        \caseText{Inverting (12) and (13), we have $\Gamma'$, $\Delta'$, and $s'$ so that $\Gamma_0 = \semicctx{\Gamma'}{\Delta'}$ and $s_0 = s' \cdot t$, and:}

        \caseFact{14} $\derivrel{\eta}{\Gamma}{\Gamma'}$

        \caseFact{15} $\derivrel{\eta}{\Delta}{\Delta'}$

        \caseFact{16} $\derivrel{p}{s}{s'}$

        \caseText{Applying (11) to (14) and (16)}

        \caseFact{17} $\tck{\cdot}{\Gamma'}{\emptyset}{e_1}{s'}{\I}$

        \caseText{By Theorem~\ref{thm:batch-sem} on (1) and (6), using (2)}

        \caseFact{18} $\neg \left(\maximalOn{\eta}{\Gamma}\right)$

        \caseText{Therefore, by (18) and (5)}

        \caseFact{19} $\emptyOn{\eta}{\Delta}$

        \caseText{Then by Theorem~\ref{thm:empty-context-deriv} and Theorem~\ref{thm:derivrel-env-fun} on (17) and (13)}

        \caseFact{20} $\Delta = \Delta'$

        \caseText{Then, \textbf{Goal B} follows by T-Cat-R on (15) and (7)}

        \caseText{For \textbf{Goal C}, assume:}

        \caseFact{21} $i_3 = \I$

        \caseFact{22} $\emptyOn{\eta}{\catpairTm{e_1}{e_2}}$

        \caseText{By (8) with (21)}

        \caseFact{23} $i_1 = \I$

        \caseText{\textbf{Goal C} follows by (11), with (22) and (23).}
    }

    \jcase{6}{S-Cat-R-2}{
        \jgivengoalThree{
            \caseFact{1} $\prefixstep{\eta}{e_1}{ }{e_1'}{p_1}$

            \caseFact{2} $\isMaximal{p_1}$

            \caseFact{3} $\prefixstep{\eta}{e_2}{ }{e_2'}{p_2}$

            \caseFact{4} $\envHasType{\eta}{\Gamma}$

            \caseFact{5} $\envHasType{\eta}{\Delta}$

            \caseFact{6} $(\exists x \in \dom{\Delta}. \neg \isEmpty{\eta(x)}) \implies (\forall x \in \dom{\Gamma}. \isMaximal{\eta(x)})$

            \caseFact{7} $\tck{\cdot}{\Gamma}{\emptyset}{e_1}{s}{i_1}$

            \caseFact{8} $\tck{\cdot}{\Delta}{\emptyset}{e_2}{t}{i_2}$

            \caseFact{9} $i_3 = \I \implies i_1 = \I \wedge \neg\left(\nullable{s}\right)$
        }{
            $\prefixHasType{\catpB{p_1}{p_2}}{s \cdot t}$
        }{
            If $\derivrel{\eta}{\semicctx{\Gamma}{\Delta}}{\Gamma_0}$ and $\derivrel{\catpB{p_1}{p_2}}{s \cdot t}{s_0}$ then $\tck{\cdot}{\Gamma_0}{\emptyset}{e_2'}{s_0}{\I}$
        }{
            If $i_3 = \I$ and $\emptyOn{\eta}{\catpairTm{e_1}{e_2}}$ then $\isEmpty{\catpB{p_1}{p_2}}$
        }

        \caseText{By IH on (1), (4), and (7)}

        \caseFact{10} $\prefixHasType{p_1}{s}$

        \caseFact{11} If $\derivrel{\eta}{\Gamma}{\Gamma'}$ and $\derivrel{p_1}{s}{s'}$, then $\tck{\cdot}{\Gamma'}{\emptyset}{e_1}{s'}{\I}$

        \caseFact{12} if $i_1 = \I$ and $\emptyOn{\eta}{e_1}$ then $\isEmpty{p_1}$

        \caseText{By IH on (3), (5), and (8)}

        \caseFact{12} $\prefixHasType{p_2}{t}$

        \caseFact{13} If $\derivrel{\eta}{\Delta}{\Delta'}$ and $\derivrel{p_2}{t}{t'}$, then $\tck{\cdot}{\Delta'}{\emptyset}{e_2}{t'}{\I}$

        \caseText{\textbf{Goal A} follows by (2), (9), and (11). For \textbf{Goal B}, we assume}

        \caseFact{14} $\derivrel{\eta}{\semicctx{\Gamma}{\Delta}}{\Gamma_0}$

        \caseFact{15} $\derivrel{\catpB{p_1}{p_2}}{s \cdot t}{s_0}$

        \caseText{By inversion on (14) and (15), there are $\Gamma'$, $\Delta'$, and $t'$ so that $\Gamma_0 = \semicctx{\Gamma'}{\Delta'}$, and $s_0 = t'$, and:}

        \caseFact{16} $\derivrel{\eta}{\Gamma}{\Gamma'}$

        \caseFact{17} $\derivrel{\eta}{\Delta}{\Delta'}$

        \caseFact{18} $\derivrel{p_2}{t}{t'}$

        \caseText{By (13) on (17) and (18)}

        \caseFact{18} $\tck{\cdot}{\Delta'}{\emptyset}{e_2'}{t'}{\I}$

        \caseText{\textbf{Goal B} follows by T-Sub with $\subty{\semicctx{\Gamma'}{\Delta'}}{\Delta'}$ For \textbf{Goal B}, assume:}

        \caseFact{19} $i_3 = \I$

        \caseFact{20} $\emptyOn{\eta}{\catpairTm{e_1}{e_2}}$

        \caseText{By (9) with (19) and (20)}

        \caseFact{21} $i_1 = \I$

        \caseFact{22} $\neg\left(\nullable{s}\right)$

        \caseText{By (12), with (21) and (20)}

        \caseFact{23} $\isEmpty{p_1}$

        \caseText{By Theorem~\ref{thm:empty-and-maximal-imply-nullable} with (2), (23), and (10)}

        \caseFact{24} $\nullable{s}$

        \caseText{But (22) and (24) are contradictory.}
    }

    \jcase{7}{S-Par-L}{
        \jgivengoalThree{
            \caseFact{1} $\eta(z) \mapsto \parp{p_1}{p_2}$

            \caseFact{2} $\prefixstep{\eta[x\mapsto p_1,y\mapsto p_2]}{e}{ }{e'}{p}$

            \caseFact{3} $\envHasType{\eta}{\Gamma(z : s \| t)}$

            \caseFact{4} $\tck{\cdot}{\Gamma(\commactx{x : s}{y : t})}{\emptyset}{e}{r}{i}$
        }{
            $\prefixHasType{p'}{r}$
        }{
            If $\derivrel{\eta}{\Gamma(z : s \| t)}{\Gamma_0}$ and $\derivrel{p}{r}{r'}$, then $\tck{\cdot}{\Gamma_0}{\emptyset}{\letparTm{x}{y}{z}{e'}}{r'}{\I}$
        }{
            If $i = \I$ and $\emptyOn{\eta}{\letparTm{x}{y}{z}{e}}$ then $\isEmpty{p}$
        }

        \caseText{By Theorem~\ref{thm:env-lookup} on (1) and (3):}

        \caseFact{5} $\prefixHasType{\parp{p_1}{p_2}}{s \| t}$

        \caseText{Inverting (5)}

        \caseFact{6} $\prefixHasType{p_1}{s}$

        \caseFact{7} $\prefixHasType{p_2}{t}$

        \caseText{By Theorem~\ref{thm:env-par-bind} on (3), (6) and (7)}

        \caseFact{8} $\envHasType{\eta[x \mapsto p_1,y\mapsto p_2]}{\Gamma(\commactx{x:s}{y:t})}$

        \caseText{By IH on (2), (4), and (8)}

        \caseFact{9} $\prefixHasType{p}{r}$

        \caseFact{10} For all $\Gamma_0'$ and $r'$, if $\derivrel{\eta[x \mapsto p_1,y \mapsto p_2]}{\Gamma(\commactx{x : s}{y : t})}{\Gamma_0'}$ and  $\derivrel{p}{r}{r'}$ then $\tck{\cdot}{\Gamma_0'}{\emptyset}{e'}{r'}{\I}$

        \caseFact{11} If $i = \I$ and $\emptyOn{\eta[x \mapsto p_1, y \mapsto p_2]}{e}$ then $\isEmpty{p}$

        \caseText{\textbf{Goal A} is complete by (9). For \textbf{Goal B}, we assume that there are $\Gamma_0$ and $r'$ so that:}

        \caseFact{12} $\derivrel{\eta}{\Gamma(z : s \| t)}{\Gamma_0}$

        \caseFact{13} $\derivrel{p}{r}{r'}$

        \caseText{By two uses of Theorem~\ref{thm:derivrel-fun}, we have $s'$ and $t'$ so that}

        \caseFact{14} $\derivrel{\parp{p_1}{p_2}}{s\|t}{s' \| t'}$

        \caseText{By Theorem~\ref{thm:env-lookup-deriv} on (1), (3), and (14), we have $\Gamma'(-)$ so that:}

        \caseFact{15} $\Gamma_0 = \Gamma'(z : s' \| t')$

        \caseText{By Theorem~\ref{thm:env-par-bind-deriv} on (1), (12), and (15)}

        \caseFact{16} $\derivrel{\eta[x\mapsto p_1,y\mapsto p_2]}{\Gamma(\commactx{x:s}{y:t})}{\Gamma'(\commactx{x:s'}{y:t'})}$

        \caseText{By (10) on (13) and (16), we have:}

        \caseFact{17} $\tck{\cdot}{\Gamma'(\commactx{x:s'}{y:t'})}{\emptyset}{e'}{r'}{\I}$

        \caseText{\textbf{Goal B} follows by T-Par-L on (17)}

        \caseText{For \textbf{Goal C}, assume:}

        \caseFact{18} $i = \I$
        
        \caseFact{19} $\emptyOn{\eta}{\letparTm{x}{y}{z}{e}}$

        \caseText{In particular with (19), since $z \in \text{fv}(\letparTm{x}{y}{z}{e})$, we have $\isEmpty{\eta(z)}$ and so:}

        \caseFact{20} $\isEmpty{p_1}$

        \caseFact{21} $\isEmpty{p_2}$

        \caseText{By (19), (20), and (21)}

        \caseFact{22} $\emptyOn{\eta[x \mapsto p_1,y \mapsto p_2]}{e}$

        \caseText{\textbf{Goal C} follows by (11), with (18) and (22)}
    }

    \jcase{8}{S-Cat-L-1}{
        \jgivengoalThree{
            \caseFact{1} $\tck{\cdot}{\Gamma(\semicctx{x : s}{y : t})}{\emptyset}{e}{r}{i}$

            \caseFact{2} $\envHasType{\eta}{\Gamma(z : s \cdot t)}$

            \caseFact{3} $\eta(z) \mapsto \catpA{p}$

            \caseFact{4} $\prefixstep{\eta[x\mapsto p,y\mapsto \emp{t}]}{e}{ }{e'}{p'}$
        }{
            $\prefixHasType{p'}{r}$
        }{
            $\forall \Gamma_0,r'$ if $\derivrel{\eta}{\Gamma(z : s \cdot t)}{\Gamma_0}$ and $\derivrel{p'}{r}{r'}$, then $\tck{\cdot}{\Gamma_0}{\emptyset}{\letcatTm{t}{x}{y}{z}{e'}}{r'}{\I}$
        }{
            If $i = \I$ and $\emptyOn{\eta}{\letcatTm{t}{x}{y}{z}{e}}$ then $\isEmpty{p'}$
        }

        \caseText{By Theorem~\ref{thm:env-lookup} on (2) and (3):}

        \caseFact{5} $\prefixHasType{\catpA{p}}{s \cdot t}$

        \caseText{Inverting (5)}

        \caseFact{6} $\prefixHasType{p}{s}$

        \caseText{By Theorem~\ref{thm:env-cat-bind-1} on (2) and (6)}

        \caseFact{7} $\envHasType{\eta[x \mapsto p,y\mapsto\emp{t}]}{\Gamma(\semicctx{x:s}{y:t})}$

        \caseText{By IH on (2), (4), and (7)}

        \caseFact{8} $\prefixHasType{p'}{r}$

        \caseFact{9} For all $\Gamma_0$ and $r'$, if $\derivrel{\eta[x \mapsto p,y\mapsto\emp{t}]}{\Gamma(\semicctx{x:s}{y:t})}{\Gamma_0}$ and $\derivrel{p'}{r}{r'}$, then $\tck{\cdot}{\Gamma_0}{\emptyset}{e'}{r'}{\I}$

        \caseFact{10} If $i = \I$ and $\emptyOn{\eta[x \mapsto p, y \mapsto\emp{t}]}{e}$ then $\isEmpty{p'}$

        \caseText{\textbf{Goal A} is completed by (8). For \textbf{Goal B}, we assume that there are $\Gamma_0$ and $r'$ such that:}

        \caseFact{11} $\derivrel{\eta}{\Gamma(z : s \cdot t)}{\Gamma_0}$

        \caseFact{12} $\derivrel{p'}{r}{r'}$

        \caseText{By Theorem~\ref{thm:derivrel-fun} on (6), there is some $s'$ so that}

        \caseFact{13} $\derivrel{p}{s}{s'}$

        \caseText{And so by definition from (13)}

        \caseFact{14} $\derivrel{\catpA{p}}{s \cdot t}{s' \cdot t}$

        \caseText{By Theorem~\ref{thm:env-lookup-deriv} on (3), (11), and (14), there is some $\Gamma'(-)$ so that}

        \caseFact{15} $\Gamma_0 = \Gamma'(z : s' \cdot t)$

        \caseText{By Theorem~\ref{thm:env-cat-bind-deriv-1} on (3) and (11) using (15)}

        \caseFact{16} $\derivrel{\eta[x \mapsto p,y\mapsto \emp{t}]}{\Gamma(\semicctx{x:s}{y:t})}{\Gamma'(\semicctx{x:s'}{y:t})}$

        \caseText{By (9), with (12) and (16)}

        \caseFact{17} $\tck{\cdot}{\Gamma'(\semicctx{x:s'}{y:t})}{\emptyset}{e'}{r'}$

        \caseText{\textbf{Goal B} follows by T-Cat-L on (17).}

        \caseText{For \textbf{Goal C}, assume:}

        \caseFact{18} $i = \I$

        \caseFact{19} $\emptyOn{\eta}{\letcatTm{t}{x}{y}{z}{e}}$

        \caseText{In particular with (19), since $z \in \text{fv}(\letcatTm{t}{x}{y}{z}{e})$, we have $\isEmpty{\eta(z)}$ and so:}

        \caseFact{20} $\isEmpty{\catpA{p}}$

        \caseText{Inverting (20)}

        \caseFact{21} $\isEmpty{p}$

        \caseText{By (21), (19) and Theorem~\ref{thm:emp-is-empty}}

        \caseFact{22} $\emptyOn{\eta[x \mapsto p,y \mapsto \emp{t}]}{e}$

        \caseText{\textbf{Goal C} follows by (10), with (18) and (22)}
    }

    \jcase{9}{S-Cat-L-2}{
        \jgivengoalThree{
            \caseFact{1} $\tck{\cdot}{\Gamma(\semicctx{x : s}{y : t})}{\emptyset}{e}{r}{i}$

            \caseFact{2} $\envHasType{\eta}{\Gamma(z : s \cdot t)}$

            \caseFact{3} $\eta(z) \mapsto \catpB{p_1}{p_2}$

            \caseFact{4} $\prefixstep{\eta[x\mapsto p_1,y\mapsto p_2]}{e}{ }{e'}{p}$
        }{
            $\prefixHasType{p}{r}$
        }{
            $\forall \Gamma_0,r'$ if $\derivrel{\eta}{\Gamma(z:s\cdot t)}{\Gamma_0}$ and $\derivrel{p}{r}{r'}$ then $\tck{\cdot}{\Gamma_0}{\emptyset}{{\cutTm{x}{\sinkTm{p_1}}{e'[z/y]}}}{r'}{\I}$
        }{
            If $i = \I$ and $\emptyOn{\eta}{\letcatTm{t}{x}{y}{z}{e}}$ then $\isEmpty{p}$
        }

        \caseText{By Theorem~\ref{thm:env-lookup} on (2) and (3):}

        \caseFact{5} $\prefixHasType{\catpB{p_1}{p_2}}{s \cdot t}$

        \caseText{Inverting (5)}

        \caseFact{6} $\prefixHasType{p_1}{s}$

        \caseFact{7} $\isMaximal{p_1}$

        \caseFact{8} $\prefixHasType{p_2}{t}$

        \caseText{By Theorem~\ref{thm:env-cat-bind-2} on (2), (6), (7), and (8),}

        \caseFact{9} $\envHasType{\eta[x\mapsto p_1,y \mapsto p_2]}{\Gamma(\semicctx{x:s}{y:t})}$

        \caseText{By IH on (1), (9), and (4)}

        \caseFact{10} $\prefixHasType{p}{r}$

        \caseFact{11} For all $\Gamma_0$ and $r'$, if $\derivrel{\eta[x \mapsto p_1,y\mapsto p_2]}{\Gamma(\semicctx{x : s}{y : t})}{\Gamma_0}$ and $\derivrel{p}{r}{r'}$, then $\tck{\cdot}{\Gamma_0}{\emptyset}{e'}{r'}{\I}$

        \caseFact{12} If $i = \I$ and $\emptyOn{\eta}{\letcatTm{t}{x}{y}{z}{e}}$ then $\isEmpty{p}$

        \caseText{\textbf{Goal A} is complete by (10). For \textbf{Goal B}, we assume that there are $\Gamma_0$ and $r'$ such that:}

        \caseFact{13} $\derivrel{\eta}{\Gamma(z : s \cdot t)}{\Gamma_0}$

        \caseFact{14} $\derivrel{p}{r}{r'}$

        \caseText{By Theorem~\ref{thm:derivrel-fun} on (8), there is some $t'$ so that}

        \caseFact{15} $\derivrel{p_2}{t}{t'}$

        \caseText{By definition from (15)}

        \caseFact{16} $\derivrel{\catpB{p_1}{p_2}}{s \cdot t}{t'}$

        \caseText{By Theorem~\ref{thm:env-lookup-deriv} on (13), (3), and (16), there is some $\Gamma'(-)$ so that:}

        \caseFact{17} $\Gamma_0 = \Gamma'(z : t')$

        \caseText{By Theorem~\ref{thm:env-cat-bind-deriv-2} on (13) and (3), there is some $s'$ so that}

        \caseFact{18} $\derivrel{\eta[x \mapsto p_1,y\mapsto p_2]}{\Gamma(\semicctx{x:s}{y:t})}{\Gamma'(\semicctx{x:s'}{y:t'})}$

        \caseFact{19} $\derivrel{p_1}{s}{s'}$

        \caseText{By (11) on (14) and (18), we have:}

        \caseFact{20} $\tck{\cdot}{\Gamma'(\semicctx{x:s'}{y:t'})}{\emptyset}{e'}{r'}{\I}$

        \caseText{By substitution on (20), we have:}

        \caseFact{21} $\tck{\cdot}{\Gamma'(\semicctx{x:s'}{z:t'})}{\emptyset}{e'[z/y]}{r'}{\I}$

        \caseText{By Theorem~\ref{thm:sink-typing} on (7), (6), and (19)}

        \caseFact{22} $\tck{\cdot}{\cdot}{\emptyset}{\sinkTm{p_1}}{s'}{\I}$

        \caseText{By T-Let on (21), (22).}

        \caseFact{23} $\tck{\cdot}{\Gamma'(\semicctx{\cdot}{z:t'})}{\emptyset}{\cutTm{x}{\sinkTm{p_1}}{e'[z/y]}}{r'}{\I}$

        \caseText{\textbf{Goal B} follows by (23) with the T-Sub using $\subty{\Gamma'(z : t')}{\Gamma'(\semicctx{\cdot}{z : t'})}$}.

        \caseText{For \textbf{Goal C}, assume:}

        \caseFact{24} $i = \I$

        \caseFact{25} $\emptyOn{\eta}{\letcatTm{t}{x}{y}{z}{e}}$

        \caseText{In particular with (25), since $z \in \text{fv}(\letcatTm{t}{x}{y}{z}{e})$, we have $\isEmpty{\eta(z)}$ and so:}

        \caseFact{26} $\isEmpty{\catpB{p_1}{p_2}}$

        \caseText{But (26) is impossible}
    }

    \jcase{10}{S-Plus-R-1}{
        \jgivengoalThree{
            \caseFact{1} $\tck{\cdot}{\Gamma}{\emptyset}{e}{s}{i}$

            \caseFact{2} $\envHasType{\eta}{\Gamma}$

            \caseFact{3} $\prefixstep{\eta}{e}{ }{e'}{p}$
        }{
            $\prefixHasType{\sumpA{p}}{s + t}$
        }{
            For all $\Gamma_0$ and $r$, if $\derivrel{\eta}{\Gamma}{\Gamma_0}$ and $\derivrel{\sumpA{p}}{s+t}{r}$ then $\tck{\cdot}{\Gamma_0}{\emptyset}{e'}{r}{\I}$
        }{
            If $\I = \J$ and $\emptyOn{\eta}{\sumInlTm{e}}$ then $\isEmpty{\sumpA{p}}$
        }

        \caseText{By IH on (1), (2), and (3)}

        \caseFact{4} $\prefixHasType{p}{s}$

        \caseFact{5} For all $\Gamma_0$ and $s'$, if $\derivrel{\eta}{\Gamma}{\Gamma_0}$ and $\derivrel{p}{s}{s'}$ then $\tck{\cdot}{\Gamma_0}{\emptyset}{e'}{s'}{\I}$

        \caseText{\textbf{Goal A} follows from (4), and \textbf{Goal B} follows from (5), inverting the derivation of $\derivrel{\sumpA{p}}{s+t}{r}$. The premise of \textbf{Goal C} is absurd.}
    }

    \jcase{11}{S-Plus-R-2}{
        \caseText{Identical to previous.}
    }

    \jcase{12}{S-Plus-L-1}{
        \jgivengoalThree{
            \caseFact{1} $\envHasType{\eta}{\Gamma(z:s+t)}$

            \caseFact{2} $\derivrel{\eta}{\Gamma(z:s+t)}{\Gamma'}$

            \caseFact{3} $\tck{\cdot}{\Gamma(x:s)}{\emptyset}{e_1}{r}{i_1}$

            \caseFact{4} $\tck{\cdot}{\Gamma(y:t)}{\emptyset}{e_2}{r}{i_2}$

            \caseFact{5} $\envHasType{\eta'}{\Gamma'}$

            \caseFact{6} $\prefixConcatRel{\eta'}{\eta}{\eta''}$

            \caseFact{7} $\eta''(z) = \sumpEmp$
        }{
            $\prefixHasType{\emp{r}}{r}$
        }{
            For all $\Gamma''$ and $r'$, if $\derivrel{\eta'}{\Gamma'}{\Gamma''}$ and $\derivrel{\emp{r}}{r}{r'}$ then $\tck{\cdot}{\Gamma''}{\emptyset}{\sumcaseTm{r}{\eta''}{z}{x}{e_1}{y}{e_2}}{r'}{\I}$
        }{
            If $i = \I$ and $\emptyOn{\eta}{\sumcaseTm{r}{\eta'}{z}{x}{e_1}{y}{e_2}}$ then $\isEmpty{\emp{r}}$
        }

        \caseText{\textbf{Goal A} is complete by Theorem~\ref{thm:empty-prefix-correct}, and \textbf{Goal C} is immediate by Theorem~\ref{thm:emp-is-empty}. For \textbf{Goal B}, assume there are $\Gamma''$ and $r'$ so that}

        \caseFact{8} $\derivrel{\eta'}{\Gamma'}{\Gamma''}$

        \caseFact{9} $\derivrel{\emp{r}}{r}{r'}$

        \caseText{By Theorem~\ref{thm:derivrel-emp}}

        \caseFact{10} $r = r'$

        \caseText{By Theorem~\ref{thm:env-concat-correct} on (1), (2), (5), (6), and (8)}

        \caseFact{11} $\envHasType{\eta''}{\Gamma(z : s + t)}$

        \caseFact{12} $\derivrel{\eta''}{\Gamma(z:s+t)}{\Gamma''}$

        \caseText{\textbf{Goal B} follows immediately by T-Plus-L on (11), (12), (3), (4), and (7).}
    }

    \jcase{13}{S-Plus-L-2}{
        \jgivengoalTwo{
            \caseFact{1} $\envHasType{\eta}{\Gamma(z:s+t)}$

            \caseFact{2} $\derivrel{\eta}{\Gamma(z:s+t)}{\Gamma'}$

            \caseFact{3} $\tck{\cdot}{\Gamma(x:s)}{\emptyset}{e_1}{r}{i_1}$

            \caseFact{4} $\tck{\cdot}{\Gamma(y:t)}{\emptyset}{e_2}{r}{i_2}$

            \caseFact{5} $i = \I \implies \eta(z) = \sumpEmp$

            \caseFact{6} $\envHasType{\eta'}{\Gamma'}$

            \caseFact{7} $\prefixConcatRel{\eta}{\eta'}{\eta''}$

            \caseFact{8} $\eta''(z) \mapsto \sumpA{p}$

            \caseFact{9} $\prefixstep{\eta''[x \mapsto p]}{e_1}{ }{e_1'}{p'}$
        }{
            $\prefixHasType{p'}{r}$
        }{
            If $\derivrel{\eta'}{\Gamma'}{\Gamma''}$ and $\derivrel{p'}{r}{r'}$ then $\tck{\cdot}{\Gamma''}{\emptyset}{e_1'[x/z]}{r'}{\I}$
        }{
            If $i = \I$ and $\emptyOn{\eta'}{\sumcaseTm{r}{\eta}{z}{x}{e_1}{y}{e_2}}$ then $\isEmpty{p'}$
        }

        \caseText{By Theorem~\ref{thm:env-concat-correct} on (1), (2), (5), (6), then:}

        \caseFact{10} $\prefixHasType{\eta''}{\Gamma(z:s+t)}$

        \caseFact{11} If $\derivrel{\eta'}{\Gamma'}{\Gamma''}$, then $\derivrel{\eta''}{\Gamma(z:s+t)}{\Gamma''}$

        \caseText{By Theorem~\ref{thm:env-lookup} on (10)}

        \caseFact{12} $\prefixHasType{\sumpA{p}}{s+t}$

        \caseText{By inversion on (12)}

        \caseFact{13} $\prefixHasType{p}{s}$

        \caseText{By Theorem~\ref{thm:env-subctx-bind} on (10) and (13)}

        \caseFact{14} $\prefixHasType{\eta''[x \mapsto p]}{\Gamma(x : s)}$

        \caseText{By IH on (3), (14), and (9)}

        \caseFact{15} $\prefixHasType{p'}{r}$

        \caseFact{16} If $\derivrel{\eta''[x \mapsto p]}{\Gamma(x:s)}{\Gamma_0}$ and $\derivrel{p'}{r}{r'}$, then $\tck{\cdot}{\Gamma_0}{\emptyset}{e_1'}{r'}$

        \caseText{\textbf{Goal A} is complete by (15). For \textbf{Goal B}, suppose that there are $\Gamma''$ and $r'$ such that:}

        \caseFact{17} $\derivrel{p'}{r}{r'}$

        \caseFact{18} $\derivrel{\eta'}{\Gamma'}{\Gamma''}$

        \caseText{By (11) and (18)}

        \caseFact{19} $\derivrel{\eta''}{\Gamma(z:s+t)}{\Gamma''}$

        \caseText{By Theorem~\ref{thm:derivrel-fun}, there is some $s'$ so that}

        \caseFact{20} $\derivrel{p}{s}{s'}$

        \caseText{By Theorem~\ref{thm:env-subctx-bind-deriv},  on (19), there is some $\Gamma'''(-)$ so that:}

        \caseFact{21} For all $\Delta'$ and $\Delta''$ and $\eta'$, if $\derivrel{\eta'}{\Delta'}{\Delta''}$ and $\agree{\eta}{\eta'}{\Delta}{\Delta'}$
                      then $\derivrel{\eta \cdot \eta'}{\Gamma(\Delta')}{\Gamma'''(\Delta'')}$

        \caseText{By (20) and (21), taking $\eta' = \{z \mapsto \stpA{p}\}$, we have:}

        \caseFact{22} $\Gamma'' = \Gamma'''(z : s')$

        \caseText{Also by (21):}

        \caseFact{23} $\derivrel{\eta''[x \mapsto p]}{\Gamma(x : s)}{\Gamma'''(x : s')}$

        \caseText{By (16) with (18), and (23)}

        \caseFact{24} $\tck{\cdot}{\Gamma'''(x : s')}{\emptyset}{e_1'}{r'}{\I}$

        \caseText{\textbf{Goal B} follows immediately by substituting $x$ for $z$ in (23).}

        \caseFact{For \textbf{Goal C}, assume:}

        \caseFact{25} $i = \I$

        \caseFact{26} $\emptyOn{\eta'}{\sumcaseTm{r}{\eta}{z}{x}{e_1}{y}{e_2}}$

        \caseText{In particular with 26, since $z \in \text{fv}(\sumcaseTm{r}{\eta}{z}{x}{e_1}{y}{e_2})$}

        \caseFact{27} $\isEmpty{\eta'(z)}$

        \caseText{By (5) with (25)}

        \caseFact{28} $\eta(z) = \sumpEmp$

        \caseText{By Definition of environment concatenation, using (7), (8), (28)}

        \caseFact{29} $\prefixConcatRel{\sumpEmp}{\eta'(z)}{\sumpA{p}}$

        \caseText{But by inversion, we note that the only case has $\eta'(z) = \sumpA{p_0}$ for some $p_0$, which contradicts (27). This completes \textbf{Goal C}.}
    }

    \jcase{14}{S-Plus-L-3}{Identical to previous.}

    \jcase{15}{S-Star-R-1}{Immediate.}

    \jcase{16}{S-Star-R-2-1}{Identical to S-Cat-R-1.}

    \jcase{17}{S-Star-R-2-2}{Identical to S-Cat-R-2.}

    \jcase{18}{S-Star-L-1}{Identical to S-Plus-L-1}

    \jcase{19}{S-Star-L-2}{Identical to S-Plus-L-2}

    \jcase{20}{S-Star-L-3}{Identical to S-Plus-L-2}

    \jcase{21}{S-Star-L-4}{Identical to S-Plus-L-2}

    \jcase{22}{S-Let}{
        \jgivengoalTwo{
            \caseFact{1} $\tck{\cdot}{\Delta}{\emptyset}{e_1}{s}{\I}$

            \caseFact{2} $\tck{\cdot}{\Gamma(x:s)}{\emptyset}{e_2}{t}{i}$

            \caseFact{3} $\envHasType{\eta}{\Gamma(\Delta)}$

            \caseFact{4} $\prefixstep{\eta}{e_1}{ }{e_1'}{p}$

            \caseFact{5} $\prefixstep{\eta[x\mapsto p]}{e_2}{ }{e_2'}{p'}$
        }{
            $\prefixHasType{p'}{t}$
        }{
            For all $\Gamma_0$ and $t'$, if $\derivrel{\eta}{\Gamma(\Delta)}{\Gamma_0}$ and $\derivrel{p'}{t}{t'}$ then $\tck{\cdot}{\Gamma_0}{\emptyset}{\cutTm{x}{e_1'}{e_2'}}{t'}{\I}$
        }{
            If $i = \I$ and $\emptyOn{\eta}{\cutTm{x}{e_1}{e_2}}$ then $\isEmpty{p'}$
        }

        \caseText{By Theorem~\ref{thm:env-subctx-lookup} on (3)}

        \caseFact{6} $\envHasType{\eta}{\Delta}$

        \caseText{By IH on (1), (6), and (4)}

        \caseFact{7} $\prefixHasType{p}{s}$

        \caseFact{8} For all $\Delta'$ and $s'$, if $\derivrel{\eta}{\Delta}{\Delta'}$ and $\derivrel{p}{s}{s'}$ then $\tck{\Delta'}{e_1'}{s'}{\I}$

        \caseFact{9} If $\I = \I$ and $\emptyOn{\eta}{e_1}$ then $\isEmpty{p}$

        \caseText{By Theorem~\ref{lemma:batch-sem-aux} and (9) on (5)}

        \caseFact{10} $\agree{\eta}{\{x \mapsto p\}}{\Delta}{x : s}$

        \caseText{By Theorem~\ref{thm:env-subctx-bind} on (4), (8), (10):}

        \caseFact{11} $\envHasType{\eta[x \mapsto p]}{\Gamma(x:s)}$

        \caseText{By IH on (3), (6), and (11)}

        \caseFact{12} $\prefixHasType{p'}{t}$

        \caseFact{13} For all $\Gamma_0'$ and $t'$, if $\derivrel{\eta[x \mapsto p]}{\Gamma(x:s)}{\Gamma_0'}$ and $\derivrel{p'}{t}{t'}$, then $\tck{\Gamma_0'}{e_2'}{t'}{\I}$

        \caseFact{14} If $i = \I$ and $\emptyOn{\eta[x \mapsto p]}{e_2}$ then $\isEmpty{p'}$

        \caseText{\textbf{Goal A} is complete by (12). For \textbf{Goal B}, we assume that there are $\Gamma_0$ and $t'$ such that:}

        \caseFact{15} $\derivrel{\eta}{\Gamma(\Delta)}{\Gamma_0}$

        \caseFact{16} $\derivrel{p'}{t}{t'}$

        \caseText{By Theorem~\ref{thm:env-subctx-bind-deriv} on (15), there is some $\Gamma'(-)$ such that}

        \caseFact{17} For all $\Delta'$ and $\Delta''$, if $\derivrel{\eta'}{\Delta'}{\Delta''}$ then $\derivrel{\eta \cdot \eta'}{\Gamma(\Delta')}{\Gamma'(\Delta'')}$

        \caseText{By Theorem~\ref{thm:derivrel-env-fun} on (7), there is some $\Delta'$ so that}

        \caseFact{18} $\derivrel{\eta}{\Delta}{\Delta'}$

        \caseText{By the uniqueness in Theorem~\ref{thm:derivrel-env-fun} and (17) and (18)}

        \caseFact{19} $\Gamma_0 = \Gamma'(\Delta')$

        \caseText{By Theorem~\ref{thm:derivrel-fun} on (8), there is some $s'$ so that}

        \caseFact{20} $\derivrel{p}{s}{s'}$

        \caseText{By (9) on (18) and (20)}

        \caseFact{21} $\tck{\cdot}{\Delta'}{\emptyset}{e_1'}{s'}{\I}$

        \caseText{By (17) on (20)}

        \caseFact{22} $\derivrel{\eta[x \mapsto p]}{\Gamma(x:s)}{\Gamma'(x:s')}$

        \caseText{By (13) on (22) and (16)}

        \caseFact{23} $\tck{\cdot}{\Gamma'(x:s')}{\emptyset}{e_2'}{t'}{\I}$

        \caseText{\textbf{Goal B} follows by T-Let on (21) and (23)}

        \caseText{For \textbf{Goal C}, assume:}

        \caseFact{24} $i = \I$

        \caseFact{25} $\emptyOn{\eta}{\cutTm{x}{e_1}{e_2}}$

        \caseText{By (25), applying (9)}

        \caseFact{26} $\isEmpty{p}$

        \caseText{By (25) and (26)}

        \caseFact{27} $\emptyOn{\eta[x \mapsto p]}{e_2}$

        \caseText{\textbf{Goal C} follows by (14) applied to (24) and (27)}
    }

    \jcase{23}{S-HistPgm}{
        \jgivengoalThree{
            \caseFact{1} $\lctck{\Omega}{M}{\flatten{s}}$

            \caseFact{2} $M \downarrow v$

            \caseFact{3} $p = \histValToPrefix{v}{s}$
        }{
            $\prefixHasType{p}{s}$
        }{
            If $\derivrel{\eta}{\Gamma}{\Gamma'}$ and $\derivrel{p}{s}{s'}$, then $\tck{\cdot}{\Gamma'}{\emptyset}{\sinkTm{p}}{s'}{\I}$
        }{
            If $\J = \I$ and $\emptyOn{\eta}{\emptyset}$ then $\isEmpty{p}$
        }

        \caseText{By Definition~\ref{def:histpgm}:}

        \caseFact{4} $v : \flatten{s}$

        \caseText{By Definition~\ref{def:type-ctx-flatten} on (4):}

        \caseFact{5} $\prefixHasType{p}{s}$

        \caseFact{6} $\isMaximal{p}$

        \caseText{\textbf{Goal A} follows by (5). For \textbf{Goal B}, assume:}

        \caseFact{7} $\derivrel{\eta}{\Gamma}{\Gamma'}$

        \caseFact{8} $\derivrel{p}{s}{s'}$

        \caseText{Then \textbf{Goal B} follows by Theorem~\ref{thm:sink-typing} on (5), (6), and (8). The premise of \textbf{Goal C} is absurd.}
    }

    \jcase{24}{S-Wait-1}{Identical to S-Plus-L-1}

    \jcase{25}{S-Wait-2}{
        \jgivengoalThree{
            \caseFact{1} $\envHasType{\eta'}{\Gamma(x : s)}$

            \caseFact{2} $\derivrel{\eta'}{\Gamma(x:s)}{\Gamma'}$

            \caseFact{3} $\tck{x : \flatten{s}}{\Gamma(\cdot)}{\emptyset}{e}{t}{i'}$

            \caseFact{4} $i = \I \implies \neg\left(\isMaximal{\eta(x)}\right) \wedge \neg\left(\nullable{s}\right)$

            \caseFact{5} $\envHasType{\eta}{\Gamma'}$

            \caseFact{6} $\prefixConcatRel{\eta'}{\eta}{\eta''}$

            \caseFact{7} $\eta''(x) = p$

            \caseFact{8} $\isMaximal{p}$

            \caseFact{9} $\prefixstep{\eta''}{e[\flatten{p}/x]}{n}{e'}{p'}$
        }{
            $\prefixHasType{p'}{t}$
        }{
            If $\derivrel{\eta}{\Gamma'}{\Gamma''}$ and $\derivrel{p'}{t}{t'}$, then $\tck{\cdot}{\Gamma''}{\emptyset}{e'}{t'}$
        }{
            If $i = \I$ and $\emptyOn{\eta}{{\waitTm{\eta'}{t}{x}{e}}}$ then $\isEmpty{p'}$
        }

        \caseText{By Theorem~\ref{thm:env-concat-correct} on (1), (2), (5), and (6)}

        \caseFact{10} $\envHasType{\eta''}{\Gamma(x : s)}$

        \caseFact{11} If $\envHasType{\eta'}{\Gamma'}$ and $\derivrel{\eta'}{\Gamma'}{\Gamma''}$, then $\derivrel{\eta''}{\Gamma(x:s)}{\Gamma''}$

        \caseText{By Theorem~\ref{thm:env-lookup} on (10) and (7) }

        \caseFact{12} $\prefixHasType{p}{s}$

        \caseText{By Definition~\ref{def:type-ctx-flatten} on (12) and (8)}

        \caseFact{13} $\flatten{p} : \flatten{s}$

        \caseText{By Definition~\ref{def:histpgm} with (3) and (13)}

        \caseFact{14} $\tck{\cdot}{\Gamma(\cdot)}{\emptyset}{e[\flatten{p}/x]}{t}{i'}$

        \caseText{By Theorem~\ref{thm:subty-env} on (10), using $\Gamma(x : s) \leq \Gamma(\cdot)$}

        \caseFact{15} $\envHasType{\eta''}{\Gamma(\cdot)}$

        \caseText{By IH on (9) with (14) and (15)}

        \caseFact{16} $\prefixHasType{p'}{t}$

        \caseFact{17} For all $\Gamma''$, if $\derivrel{\eta''}{\Gamma(\cdot)}{\Gamma''}$ and $\derivrel{p'}{t}{t'}$, then $\tck{\cdot}{\Gamma''}{\emptyset}{e'}{t'}{\I}$

        \caseText{(16) completes \textbf{Goal A}. For \textbf{Goal B}, suppose:}

        \caseFact{18} $\derivrel{\eta}{\Gamma'}{\Gamma''}$

        \caseFact{19} $\derivrel{p'}{t}{t'}$

        \caseText{By (18) and (11)}

        \caseFact{20} $\derivrel{\eta''}{\Gamma(x:s)}{\Gamma''}$

        \caseText{By Theorem~\ref{thm:env-subctx-bind-deriv}, on (20) using (6), there is some $\Gamma'''(-)$ so that:}

        \caseFact{21} For all $\Delta'$ and $\Delta''$ and $\eta_0$, if $\derivrel{\eta_0}{\Delta'}{\Delta''}$ and $\agree{\eta''}{\eta_0}{\Delta}{\Delta'}$
                      then $\derivrel{\eta'' \cdot \eta_0}{\Gamma(\Delta')}{\Gamma'''(\Delta'')}$

        \caseText{By (21), taking $\Delta' = \cdot$, $\Delta'' = \left(x : \deriv{p}{s}\right)$, and $\eta_0 = \{x \mapsto p\}$}

        \caseFact{22} $\Gamma'' = \Gamma'''(x : \deriv{p}{s})$h

        \caseText{By (21) again, taking $\Delta' = \Delta'' = \cdot$}

        \caseFact{23} $\derivrel{\eta''}{\Gamma(\cdot)}{\Gamma'''(\cdot)}$

        \caseText{By (17) on (23) and (19)}

        \caseFact{24} $\tck{\cdot}{\Gamma'''(\cdot)}{\emptyset}{e'}{t'}$

        \caseText{\textbf{Goal B} follows from (23), using $\Gamma'''(x : \deriv{p}{s}) \leq \Gamma'''(\cdot)$}

        \caseText{For \textbf{Goal C}, assume:}

        \caseFact{25} $i = \I$

        \caseFact{26} $\emptyOn{\eta}{{\waitTm{\eta'}{t}{x}{e}}}$

        \caseText{Since $x \in \text{fv}({\waitTm{\eta'}{t}{x}{e}})$, we have}

        \caseFact{27} $\isEmpty{\eta(x)}$

        \caseText{Applying (4) to (25)}

        \caseFact{27} $\neg\left(\isMaximal{\eta'(x)}\right)$
        
        \caseFact{28} $\neg\left(\nullable{s}\right)$

        \caseText{By the definition of environment concatenation, on (6) and (7)}

        \caseFact{29} $\prefixConcatRel{\eta'(x)}{\eta(x)}{p}$

        \caseText{By Theorem~\ref{thm:concat-maximal} on (29) and (8)}

        \caseFact{30} $\isMaximal{\eta(x)}$

        \caseText{But by Theorem~\ref{thm:empty-and-maximal-imply-nullable}, (27), (30), and (28) are contradictory. }
    }

    \jcase{26}{S-Fix}{
        \jgivengoalThree{
            \caseFact{1} $\tck{\Omega}{\Gamma}{\Omega \mid \Gamma \to s \, @ \, i}{e}{s}{i}$

            \caseFact{2} $\tck{\cdot}{\Gamma'}{\emptyset}{A}{\Gamma}{i}$

            \caseFact{3} $\lctck{\cdot}{\overline{M}}{\Omega}$

            \caseFact{4} $\envHasType{\eta}{\Gamma'}$

            \caseFact{5} $\overline{M} \downarrow \theta$

            \caseFact{6} $\prefixstep{\eta}{\cutTm{\Gamma}{A}{\fixsubst{e}{e}[\theta]}}{n}{e'}{p}$
        }{
            $\prefixHasType{p}{s}$
        }{
            If $\derivrel{\eta}{\Gamma'}{\Gamma''}$ and $\derivrel{p}{s}{s'}$, then $\tck{\cdot}{\Gamma''}{\emptyset}{e'}{s'}{\I}$
        }{
            If $i = \I$ and $\emptyOn{\eta}{\cutTm{\Gamma}{A}{e}}$ then $\isEmpty{p}$
        }

        \caseText{By Theorem~\ref{thm:fixsubst} on (1)}

        \caseFact{7} $\tck{\Omega}{\Gamma}{\emptyset}{\fixsubst{e}{e}}{s}{i}$

        \caseText{By Definition~\ref{def:histpgm} on (4)}

        \caseFact{8} $\theta : \Omega'$

        \caseText{Then again by Definition~\ref{def:histpgm} on (7) and (8)}

        \caseFact{9} $\tck{\cdot}{\Gamma}{\emptyset}{\fixsubst{e}{e}[\theta]}{s}{i}$

        \caseText{By \ruleName{T-ArgsLet} on (2) and (9)}

        \caseFact{10} $\tck{\cdot}{\Gamma'}{\emptyset}{{\cutTm{\Gamma}{A}{\fixsubst{e}{e}[\theta]}}}{s}{i}$

        \caseText{Goal follows immediately by IH on (6), with (4) and (10).}
    }

    \jcase{27}{S-ArgsLet}{
        \jgivengoalThree{
            \caseFact{1} $\tck{\cdot}{\Gamma_i}{\emptyset}{A}{\Gamma_o}$

            \caseFact{2} $\tck{\cdot}{\Gamma_o}{\emptyset}{e}{s}$

            \caseFact{3} $\envHasType{\eta}{\Gamma_i}$

            \caseFact{4} $\argsstep{\eta}{A}{\Gamma_o}{n_1}{A'}{\eta'}$

            \caseFact{5} $\prefixstep{\eta'}{e}{n_2}{e'}{p}$

            \caseFact{6} $\derivrel{\eta'}{\Gamma_o}{\Gamma_o'}$
        }{
            $\prefixHasType{p}{s}$
        }{
            For all $\Gamma_i'$ and $s'$, if $\derivrel{\eta}{\Gamma_i}{\Gamma_i'}$ and $\derivrel{p'}{s}{s'}$ then $\tck{\cdot}{\Gamma_i'}{\emptyset}{\cutTm{\Gamma_o'}{A'}{e'}}{s'}$
        }{
            If $i = \I$ and $\emptyOn{\eta}{\cutTm{\Gamma_o}{A}{e}}$ then $\isEmpty{p}$
        }

        \caseText{By IH on (4), with (1) and (3)}

        \caseFact{7} $\envHasType{\eta'}{\Gamma_o}$

        \caseFact{8} For all $\Gamma_i'$, if $\derivrel{\eta}{\Gamma_i}{\Gamma_i'}$ then $\tck{\cdot}{\Gamma_i'}{\emptyset}{A'}{\Gamma_o'}$

        \caseFact{9} If $i = \I$ and $\emptyOn{\eta}{A}$ then $\emptyOn{\eta'}{\Gamma_o}$

        \caseText{By IH on (5), with (2) and (7)}

        \caseFact{10} $\prefixHasType{p}{s}$

        \caseFact{11} For all $s'$, if $\derivrel{p}{s}{s'}$, then $\tck{\cdot}{\Gamma_o'}{\emptyset}{e'}{s'}$

        \caseFact{12} If $i = \I$ and $\emptyOn{\eta'}{e}$ then $\isEmpty{p}$

        \caseText{\textbf{Goal A} is (9). \textbf{Goal B} follows by assuming $\Gamma_i'$ and $s'$, specializing (8) and (10), and applying \ruleName{T-ArgsLet}.}

        \caseText{For \textbf{Goal C}, assume:}

        \caseFact{13} $i = \I$

        \caseFact{14} $\emptyOn{\eta}{\cutTm{\Gamma_o}{A}{e}}$

        \caseText{Because $\text{fv}(\cutTm{\Gamma_o}{A}{e}) = \text{fv}(A)$}

        \caseFact{15} $\emptyOn{\eta}{A}$

        \caseText{Applying (9) to (13) and (15)}

        \caseFact{16} $\emptyOn{\eta'}{\Gamma_o}$

        \caseText{Because of (2), $\text{fv}(e) \subseteq \Gamma_o$, and so (16) implies}

        \caseFact{17} $\emptyOn{\eta'}{e}$

        \caseText{\textbf{Goal C} is complete by applying (12) to (13) and (17)}
    }

    \jcase{28}{S-Args-Emp}{Immediate.}
    \jcase{29}{S-Args-Sng}{Immediate by IH.}
    \jcase{30}{S-Args-Comma}{Like \ruleName{T-Par-R}}
    \jcase{31}{S-Args-Semic-1-1}{Like \ruleName{T-Cat-R-1}}
    \jcase{32}{S-Args-Semic-1-2}{Like \ruleName{T-Cat-R-2}}
    \jcase{33}{S-Args-Semic-2-1}{Immediate by IH}
}

The following theorem proves that sink terms live up to their names. Given any maximal input prefix $p$,
the program $\sinkTm{p}$ will output an empty prefix of the appropriate type, and then step to itself.
\begin{theorem}[Sink Term Semantics Characterization]
    \label{thm:sink-step}
    If $\prefixHasType{p}{s}$, and $\isMaximal{p}$ then
    for all $n$ and $\eta$, we have $\prefixstep{\eta}{\sinkTm{p}}{n}{\sinkTm{p}}{\emp{\deriv{p}{s}}}$.
\end{theorem}

\begin{theorem}[Homomorphism Theorem]
    \label{thm:hom-thm}
\begin{enumerate}
    \item Suppose:
    \begin{itemize}
        \item $\tck{\cdot}{\Gamma}{\emptyset}{e}{s}{i}$
        \item $\envHasType{\eta}{\Gamma}$
        \item $\envHasType{\eta'}{\deriv{\eta}{\Gamma}}$
        \item $\prefixstep{\eta}{e}{n_1}{e'}{p}$
        \item $\prefixstep{\eta'}{e'}{n_2}{e''}{p'}$
        \item $\prefixstep{\prefixConcat{\eta}{\eta'}}{e}{n_1 + n_2}{e'''}{p''}$
    \end{itemize}
    Then $e'' = e'''$ and $\prefixConcatRel{p}{p'}{p''}$.
    \item Suppose:
    \begin{itemize}
        \item $\tck{\cdot}{\Gamma_i}{\emptyset}{A}{\Gamma_i}{i}$
        \item $\envHasType{\eta_i}{\Gamma_i}$
        \item $\envHasType{\eta_i'}{\deriv{\eta_i}{\Gamma}}$
        \item $\argsstep{\eta_i}{A}{\Gamma_o}{n_1}{A'}{\eta_o}$
        \item $\argsstep{\eta_i'}{A'}{\deriv{\eta_o}{\Gamma_o}}{n_2}{A''}{\eta_o'}$
        \item $\argsstep{\prefixConcat{\eta_i}{\eta_i'}}{A}{\Gamma_o}{n_1 + n_2}{A'''}{\eta_o''}$
    \end{itemize}
    Then $A'' = A'''$, and $\prefixConcatRel{\eta_o}{\eta_o'}{\eta_o''}$.
\end{enumerate}
\end{theorem}
\jtheorem{}{
By mutual induction on the semantics judgments, then inverting all other judgments.
To reduce clutter, we will omit the typing premises that simply go along for the ride in each case.
We name the cases by the rule used for the step of $e$, and then if they are not uniquely determined,
the step for $e'$ an then the step of $e$ on $\prefixConcat{\eta}{\eta'}$.

    \jcase{1}{S-Eps-R}{Immediate.}

    \jcase{2}{S-One-R}{Immediate.}

    \jcase{3}{S-Var}{Immediate.}

    \jcase{4}{S-Par-R}{
        \jgivengoalTwo{
            \caseFact{1} $\prefixstep{\eta_1}{e_1}{n_1}{e_1'}{p_1}$

            \caseFact{2} $\prefixstep{\eta_1}{e_2}{n_2}{e_2'}{p_2}$

            \caseFact{3} $\prefixstep{\eta_2}{e_1'}{n_1'}{e_1''}{p_1'}$

            \caseFact{4} $\prefixstep{\eta_2}{e_2'}{n_2'}{e_2''}{p_2'}$

            \caseFact{5} $\prefixstep{\prefixConcat{\eta_1}{\eta_2}}{e_1}{n_1''}{e_1'''}{p_1''}$

            \caseFact{6} $\prefixstep{\prefixConcat{\eta_1}{\eta_2}}{e_2}{n_2''}{e_2'''}{p_2''}$
        }{
            $\parpairTm{e_1''}{e_2''} = \parpairTm{e_1'''}{e_2'''}$
        }{
            $\prefixConcatRel{\parp{p_1}{p_2}}{\parp{p_1'}{p_2'}}{\parp{p_1''}{p_2''}}$
        }

        \caseText{By IH on (1), (3), and (5)}

        \caseFact{7} $e_1'' = e_1'''$

        \caseFact{8} $\prefixConcatRel{p_1}{p_1'}{p_1''}$

        \caseText{By IH on (2), (4), and (6)}

        \caseFact{9} $e_2'' = e_2'''$

        \caseFact{10} $\prefixConcatRel{p_2}{p_2'}{p_2''}$

        \caseText{\textbf{Goal A} is immediate by (7) and (9)}

        \caseText{\textbf{Goal B} is immediate by (8) and (10)}
    }

    \jcase{5}{S-Cat-R-1, S-Cat-R-1, S-Cat-R-1}{
        \jgivengoalTwo{
            \caseFact{1} $\prefixstep{\eta_1}{e_1}{n}{e_1'}{p_1}$

            \caseFact{2} $\neg \left(\isMaximal{p_1}\right)$

            \caseFact{3} $\prefixstep{\eta_2}{e_1'}{n'}{e_1''}{p_1'}$

            \caseFact{4} $\neg \left(\isMaximal{p_1'}\right)$

            \caseFact{5} $\prefixstep{\prefixConcat{\eta_1}{\eta_2}}{e_1}{n}{e_1'''}{p_1''}$

            \caseFact{6} $\neg \left(\isMaximal{p_1''}\right)$
        }{
            $\catpairTm{e_1''}{e_2} = \catpairTm{e_1'''}{e_2}$
        }{
            $\prefixConcatRel{\catpA{p_1}}{\catpA{p_1'}}{\catpA{p_1''}}$
        }

        \caseText{Immediate by IH on (1), with (3) and (5).}
    }

    \jcase{6}{S-Cat-R-1, S-Cat-R-1, S-Cat-R-2}{
        \jgivengoalTwo{
            \caseFact{1} $\prefixstep{\eta_1}{e_1}{n}{e_1'}{p_1}$

            \caseFact{2} $\neg \left(\isMaximal{p_1}\right)$

            \caseFact{3} $\prefixstep{\eta_2}{e_1'}{n'}{e_1''}{p_1'}$

            \caseFact{4} $\neg \left(\isMaximal{p_1'}\right)$

            \caseFact{5} $\prefixstep{\prefixConcat{\eta_1}{\eta_2}}{e_1}{n}{e_1'''}{p_1''}$

            \caseFact{6} $\isMaximal{p_1''}$

            \caseFact{7} $\prefixstep{\prefixConcat{\eta_1}{\eta_2}}{e_2}{n}{e_2'}{p_2}$
        }{
            $\catpairTm{e_1''}{e_2} = e_2'$
        }{
            $\prefixConcatRel{\catpA{p_1}}{\catpA{p_1'}}{\catpB{p_1''}{p_2}}$
        }

        \caseText{By IH on (1), with (3) and (5):}

        \caseFact{8} $e_1'' = e_1'''$

        \caseFact{9} $\prefixConcatRel{p_1}{p_1'}{p_1''}$

        \caseText{But (9) is impossible by Theorem~\ref{thm:concat-maximal}, since $p_1''$ is maximal, but neither $p_1$ nor $p_1'$ are.}
    }

    \jcase{7}{S-Cat-R-1, S-Cat-R-2, S-Cat-R-1}{
        \jgivengoalTwo{
            \caseFact{1} $\prefixstep{\eta_1}{e_1}{n}{e_1'}{p_1}$

            \caseFact{2} $\neg \left(\isMaximal{p_1}\right)$

            \caseFact{3} $\prefixstep{\eta_2}{e_1'}{n_1}{e_1''}{p_1'}$

            \caseFact{4} $\isMaximal{p_1'}$

            \caseFact{5} $\prefixstep{\eta_2}{e_2}{n_1}{e_2'}{p_2}$

            \caseFact{5} $\prefixstep{\prefixConcat{\eta_1}{\eta_2}}{e_1}{n}{e_1'''}{p_1''}$

            \caseFact{6} $\neg \left(\isMaximal{p_1''}\right)$
        }{
            $e_2' = \catpairTm{e_1''}{e_2}$
        }{
            $\prefixConcatRel{\catpA{p_1}}{\catpB{p_1'}{p_2}}{\catpA{p_1''}}$
        }

        \caseText{By IH on (1), with (3) and (5):}

        \caseFact{8} $e_1'' = e_1'''$

        \caseFact{9} $\prefixConcatRel{p_1}{p_1'}{p_1''}$

        \caseText{But (9) is impossible by Theorem~\ref{thm:concat-maximal}, since $p_1'$ is maximal, $p_1''$ is not.}
    }

    \jcase{8}{S-Cat-R-1, S-Cat-R-2, S-Cat-R-2}{
        \jgivengoalTwo{

            \caseFact{1} $\tck{\cdot}{\Gamma}{\emptyset}{e_1}{s}{i_1}$

            \caseFact{2} $\tck{\cdot}{\Delta}{\emptyset}{e_2}{t}{i_2}$

            \caseFact{3} $\maximalOn{\eta_1}{\Gamma} \vee \maximalOn{\eta_2}{\Delta}$

            \caseFact{4} $\prefixstep{\eta_1}{e_1}{n}{e_1'}{p_1}$

            \caseFact{5} $\neg \left(\isMaximal{p_1}\right)$

            \caseFact{6} $\prefixstep{\eta_2}{e_1'}{n_1}{e_1''}{p_1'}$

            \caseFact{7} $\isMaximal{p_1'}$

            \caseFact{8} $\prefixstep{\eta_2}{e_2}{n_1}{e_2'}{p_2}$

            \caseFact{9} $\prefixstep{\prefixConcat{\eta_1}{\eta_2}}{e_1}{n}{e_1'''}{p_1''}$

            \caseFact{10} $\isMaximal{p_1''}$

            \caseFact{11} $\prefixstep{\prefixConcat{\eta_1}{\eta_2}}{e_2}{n}{e_2''}{p_2'}$
        }{
            $e_2' = e_2''$
        }{
            $\prefixConcatRel{\catpA{p_1}}{\catpB{p_1'}{p_2}}{\catpB{p_1''}{p_2'}}$
        }

        \caseText{By IH on (4), with (6) and (8):}

        \caseFact{12} $e_1'' = e_1'''$

        \caseFact{13} $\prefixConcatRel{p_1}{p_1'}{p_1''}$

        \caseText{By the contrapositive of Theorem~\ref{thm:batch-sem} on (1), (3), and (5)}

        \caseFact{14} $\neg \left(\maximalOn{\eta_1}{\Gamma}\right)$

        \caseText{By (3) and (14)}

        \caseFact{15} $\emptyOn{\eta_1}{\Delta}$

        \caseText{Because $\dom{Delta} \geq \text{fv}(e_2)$, (15) implies}

        \caseFact{16} $\emptyOn{\eta_1}{e_2}$

        \caseText{By Theorem~\ref{thm:env-concat-emp} on (16)}

        \caseFact{17} $\left(\prefixConcat{\eta_1}{\eta_2}\right)|_{e_2} = \eta_2|_{e_2}$

        \caseText{So by (8) and (16)}

        \caseFact{18} $\prefixstep{\prefixConcat{\eta_1}{\eta_2}}{e_2}{n_1}{e_2''}{p_2}$

        \caseText{By (18) and Theorem~\ref{thm:prefix-sem-det}}

        \caseFact{19} $p_2 = p_2'$

        \caseFact{20} $e_2' = e_2''$

        \caseText{\textbf{Goal A} is immediate by (20), and \textbf{Goal B} follows by (19) and (13)}
    }

    \jcase{9}{S-Cat-R-2}{
        \jgivengoalTwo{
            \caseFact{1} $\prefixstep{\eta_1}{e_1}{n_1}{e_1'}{p_1}$

            \caseFact{2} $\isMaximal{p_1}$

            \caseFact{3} $\prefixstep{\eta_1}{e_2}{n_2}{e_2'}{p_2}$

            \caseFact{4} $\prefixstep{\eta_2}{e_2'}{n'}{e_2''}{p_2'}$

            \caseFact{5} $\prefixstep{\prefixConcat{\eta_1}{\eta_2}}{\catpairTm{e_1}{e_2}}{n''}{e_0}{p_0}$
        }{
            $e_2'' = e_0$
        }{
            $\prefixConcatRel{\catpB{p_1}{p_2}}{p_2'}{p_0}$
        }

        \caseText{We begin by inverting (5). The case of S-Cat-R-1 is impossible, because with Theorem~\ref{thm:concat-maximal}, this would contradict the maximality of $p_1$. Thus,}

        \caseFact{6} $\prefixstep{\prefixConcat{\eta_1}{\eta_2}}{e_1}{n_1'}{e_1''}{p_1'}$

        \caseFact{7} $\isMaximal{p_1'}$

        \caseFact{8} $\prefixstep{\prefixConcat{\eta_1}{\eta_2}}{e_2}{n_2'}{e_2'''}{p_2''}$

        \caseText{The inversion also tells us that $p_0 = \catpB{p_1''}{p_2''}$, and $e_0 = e_2'''$}

        \caseText{By IH on (3), (4), and (8)}

        \caseFact{9} $\prefixConcatRel{p_2}{p_2'}{p_2''}$

        \caseFact{10} $e_2''' = e_2''$

        \caseText{\textbf{Goal A} is complete by (10)}

        \caseText{By (9):}

        \caseFact{11} $\prefixConcatRel{\catpB{p_1''}{p_2}}{p_2'}{\catpB{p_1''}{p_2''}}$

        \caseText{By Theorem~\ref{thm:max-sem-extend} on (1) and (6)}

        \caseFact{12} $p_1' = p_1''$

        \caseText{\textbf{Goal B} follows by (11) and (12).}
    }

    \jcase{10}{S-Par-L}{Immediate by two uses of IH}

    \jcase{11}{S-Cat-L-1, S-Cat-L-1}{
        \jgivengoalTwo{
            \caseFact{1} $\eta_1(z) \mapsto \catpA{p_1}$

            \caseFact{2} $\prefixstep{\eta_1[x\mapsto p_1,y\mapsto \emp{t}]}{e}{n}{e'}{p_1'}$

            \caseFact{3} $\eta_2(z) \mapsto \catpA{p_2}$

            \caseFact{4} $\prefixstep{\eta_2[x\mapsto p_2,y\mapsto \emp{t}]}{e'}{n'}{e''}{p_2'}$

            \caseFact{5} $\prefixstep{\prefixConcat{\eta_1}{\eta_2}}{\letcatTm{t}{x}{y}{z}{e}}{n''}{e_0}{p_0}$
        }{
            $\letcatTm{t}{x}{y}{z}{e''} = e_0$
        }{
            $\prefixConcatRel{p_1'}{p_2'}{p_0}$
        }

        \caseText{By Definition,}

        \caseFact{6} $(\prefixConcat{\eta_1}{\eta_2})(z) = \catpA{\prefixConcat{p_1}{p_2}}$

        \caseText{By (6), inverting (5) can only conclude with S-Cat-L-1, and so $e_0 = \letcatTm{t}{x}{y}{z}{e'''}$ such that}

        \caseFact{7} $\prefixstep{\left(\prefixConcat{\eta_1}{\eta_2}\right)[x \mapsto \prefixConcat{p_1}{p_2}, y \mapsto \emp{t}]}{e}{n''}{e'''}{p_0}$

        \caseText{By definition:}

        \caseFact{8} ${\left(\prefixConcat{\eta_1}{\eta_2}\right)[x \mapsto \prefixConcat{p_1}{p_2}, y \mapsto \emp{t}]} = \prefixConcat{\eta_1[x \mapsto p_1, y \mapsto \emp{t}]}{\eta_2[x \mapsto p_2, y \mapsto \emp{t}]}$

        \caseText{By IH on (2), with (4) and (7), rewriting by (8)}

        \caseFact{9} $e''' = e''$

        \caseFact{10} $\prefixConcatRel{p_1'}{p_2'}{p_0}$

        \caseText{(9) and (10) complete \textbf{Goal A} and \textbf{Goal B}, respectively.}
    }

    \jcase{12}{S-Cat-L-1, S-Cat-L-2}{
        \jgivengoalTwo{
            \caseFact{1} $\eta_1(z) \mapsto \catpA{p_1}$

            \caseFact{2} $\prefixstep{\eta_1[x\mapsto p_1,y\mapsto \emp{t}]}{e}{n}{e'}{p_1'}$

            \caseFact{3} $\eta_2(z) \mapsto \catpB{p_1'',p_2}$

            \caseFact{4} $\prefixstep{\eta_2[x\mapsto p_1'',y\mapsto p_2]}{e'}{n'}{e''}{p_2'}$

            \caseFact{5} $\prefixstep{\prefixConcat{\eta_1}{\eta_2}}{\letcatTm{t}{x}{y}{z}{e}}{n''}{e_0}{p_0}$
        }{
            $\cutTm{x}{\sinkTm{p_1''}}{e''[z/y]} = e_0$
        }{
            $\prefixConcatRel{p_1'}{p_2'}{p_0}$
        }

        \caseText{By Definition,}

        \caseFact{6} $(\prefixConcat{\eta_1}{\eta_2})(z) = \prefixConcat{\catpA{p_1}}{\catpB{p_1''}{p_2}} = \catpB{\prefixConcat{p_1}{p_1''}}{p_2}$

        \caseText{By (6), inverting (5) can only conclude with S-Cat-L-2, and so $e_0 = \cutTm{x}{\sinkTm{\prefixConcat{p_1}{p_1''}}}{e'''}$, such that:}

        \caseFact{7} $\prefixstep{\eta_2[x \mapsto \prefixConcat{p_1}{p_1''}, y \mapsto p_2]}{e}{n''}{e'''}{p_0}$

        \caseText{By IH on (2), with (4) and (7)}

        \caseFact{8} $e''' = e''$

        \caseFact{9} $\prefixConcatRel{p_1'}{p_2'}{p_0}$

        \caseText{\textbf{Goal B} follows by (9). By Theorem~\ref{thm:sink-concat}:}

        \caseFact{10} $\sinkTm{p_1''} = \sinkTm{\prefixConcat{p_1}{p_1''}}$

        \caseText{\textbf{Goal A} follows by (8) and (10).}
    }

    \jcase{13}{S-Cat-L-2}{
        \jgivengoalTwo{
            \caseFact{1} $\eta_1(z) \mapsto \catpB{p_1}{p_2}$

            \caseFact{2} $\prefixstep{\eta_1[x\mapsto p_1,y\mapsto p_2]}{e}{n}{e'}{p}$

            \caseFact{3} $\prefixstep{\eta_2}{\cutTm{x}{\sinkTm{p_1}}{e'[z/y]}}{n_1+n_2}{\cutTm{x}{e_0}{e_0'}}{p'}$

            \caseFact{4} $\prefixstep{\prefixConcat{\eta_1}{\eta_2}}{\letcatTm{t}{x}{y}{z}{e}}{n'}{e_1}{p''}$
        }{
            $e_1 = \cutTm{x}{e_0}{e_0'}$
        }{
            $\prefixConcatRel{p}{p'}{p''}$
        }

        \caseText{Inverting (3)}

        \caseFact{5} $\prefixstep{\eta_2}{\sinkTm{p_1}}{n_1}{e_0}{p_0}$

        \caseFact{6} $\prefixstep{\eta_2[x \mapsto p_0]}{e'[z/y]}{n_2}{e_0'}{p'}$

        \caseText{By Definition, there is some $p_2'$ such that $\eta_2(z) = p_2'$, and so}

        \caseFact{7} $\left(\prefixConcat{\eta_1}{\eta_2}\right)(z) = \catpB{p_1}{\prefixConcat{p_2}{p_2'}}$

        \caseText{Because of (7), inverting (4) can only end with S-Cat-L-2, and so we have that $e_1 = \cutTm{x}{\sinkTm{p_1}}{e''[y/z]}$, and}

        \caseFact{8} $\prefixstep{\left(\prefixConcat{\eta_1}{\eta_2}\right)[x \mapsto p_1, y\mapsto \prefixConcat{p_2}{p_2'}]}{e}{n'}{e''}{p''}$

        \caseText{By Theorem~\ref{thm:sink-step} and Theorem~\ref{thm:prefix-sem-det} on (5)}

        \caseFact{9} $e_0 = \sinkTm{p_1}$

        \caseFact{10} $p_0 = \emp{\deriv{p_1}{s}}$

        \caseText{Because $\eta_2(z) = p_2'$, we have that (6) equivalently says:}

        \caseFact{11} $\prefixstep{\eta_2[x \mapsto \emp{\deriv{p_1}{s}}, y \mapsto p_2']}{e'}{n_2}{e_0'}{p'}$

        \caseText{Then, we can compute:}

        \caseFact{12} $\left(\prefixConcat{\eta_1}{\eta_2}\right)[x \mapsto p_1, y\mapsto \prefixConcat{p_2}{p_2'}] = \prefixConcat{\left(\eta_1[x \mapsto p_1, y \mapsto p_2]\right)}{\left(\eta_2[x \mapsto \emp{\deriv{p_1}{s}}, y \mapsto p_2']\right)}$

        \caseText{So by IH on (2), with (11) and (8)}

        \caseFact{13} $e'' = e_0'$

        \caseFact{14} $\prefixConcatRel{p}{p'}{p''}$

        \caseText{\textbf{Goal A} follows by (9) and (13), with \textbf{Goal B} immediate by (14)}
    }

    \jcase{14}{S-Plus-R-1}{Immediate by IH.}

    \jcase{15}{S-Plus-R-2}{Immediate by IH.}

    \jcase{16}{S-Plus-L-1, S-Plus-L-1}{
        \jgivengoalTwo{
            \caseFact{1} $\prefixConcatRel{\eta_1'}{\eta_1}{\eta_1''}$

            \caseFact{2} $\eta_1''(z) = \sumpEmp$

            \caseFact{3} $\prefixConcatRel{\eta_1''}{\eta_2}{\eta_2'}$

            \caseFact{4} $\eta_2'(z) = \sumpEmp$

            \caseFact{5} $\prefixstep{\prefixConcat{\eta_1}{\eta_2}}{\sumcaseTm{r}{\eta_1'}{z}{x}{e_1}{y}{e_2}}{n}{e_0}{p}$
        }{
            $e_0 = \sumcaseTm{r}{\eta_2'}{z}{x}{e_1}{y}{e_2}$
        }{
            $\prefixConcatRel{\emp{t}}{\emp{t}}{p}$
        }

        \caseText{By Theorem~\ref{thm:env-concat-assoc}}

        \caseFact{6} $\prefixConcat{\eta_1'}{\left(\prefixConcat{\eta_1}{\eta_2}\right)} = \prefixConcat{\left(\prefixConcat{\eta_1'}{\eta_1}\right)}{\eta_2} = \prefixConcat{\eta_1''}{\eta_2} = \eta_2'$

        \caseText{By (6) and (4), the only rule can conclude with (5) is S-Plus-L-1. Inverting (5) and rewriting by (6), we have that}

        \caseFact{7} $e_0 = \sumcaseTm{r}{\eta_2'}{z}{x}{e_1}{y}{e_2}$

        \caseFact{8} $p = \emp{t}$

        \caseText{\textbf{Goal A} is (7), and \textbf{Goal B} follows by Theorem~\ref{thm:prefix-concat-emp} and (8)}
    }

    \jcase{17}{S-Plus-L-1, S-Plus-L-2}{
        \jgivengoalTwo{
            \caseFact{1} $\prefixConcatRel{\eta_1'}{\eta_1}{\eta_1''}$

            \caseFact{2} $\eta_1''(z) = \sumpEmp$

            \caseFact{3} $\prefixConcatRel{\eta_1''}{\eta_2}{\eta_2'}$

            \caseFact{4} $\eta_2'(z) = \sumpA{p}$

            \caseFact{5} $\prefixstep{\eta_2'[x \mapsto p]}{e_1}{n}{e_1'}{p'}$

            \caseFact{6} $\prefixstep{\prefixConcat{\eta_1}{\eta_2}}{\sumcaseTm{r}{\eta_1'}{z}{x}{e_1}{y}{e_2}}{n'}{e_0}{p''}$
        }{
            $e_0 = e_1'[z/x]$
        }{
            $\prefixConcatRel{\emp{t}}{p'}{p''}$
        }

        \caseText{By Theorem~\ref{thm:env-concat-assoc}}

        \caseFact{7} $\prefixConcat{\eta_1'}{\left(\prefixConcat{\eta_1}{\eta_2}\right)} = \prefixConcat{\left(\prefixConcat{\eta_1'}{\eta_1}\right)}{\eta_2} = \prefixConcat{\eta_1''}{\eta_2} = \eta_2'$

        \caseText{By (7) and (4), the only rule can conclude with (5) is S-Plus-L-4. Inverting (6) and rewriting by (7), we have that there is $e_1''$}

        \caseFact{8} $e_0 = e_1''[z/x]$

        \caseFact{9} $\prefixstep{\eta_2'[x \mapsto p]}{e_1}{n'}{e_1''}{p'}$

        \caseText{By Theorem~\ref{thm:prefix-sem-det} on (5) and (9)}

        \caseFact{10} $e_1'' = e_1'$

        \caseFact{11} $p'' = p'$

        \caseText{\textbf{Goal A} follows by (8) and (10), while \textbf{Goal B} follows by Theorem~\ref{thm:prefix-concat-emp} and (11)}
    }

    \jcase{18}{S-Plus-L-1, S-Plus-L-3}{Identical to Previous.}

    \jcase{19}{S-Plus-L-2}{
        \jgivengoalTwo{
            \caseFact{1} $\prefixConcatRel{\eta_1'}{\eta_1}{\eta_1''}$

            \caseFact{2} $\eta_1''(z) = \sumpA{p_1}$

            \caseFact{3} $\prefixstep{\eta_1''[x \mapsto p_1]}{e_1}{n}{e_1'}{p}$

            \caseFact{4} $\prefixstep{\eta_2}{e_1'[z/x]}{n'}{e}{p'}$

            \caseFact{5} $\prefixstep{\prefixConcat{\eta_1}{\eta_2}}{\sumcaseTm{r}{\eta_1'}{z}{x}{e_1}{y}{e_2}}{n''}{e'}{p''}$
        }{
            $e = e'$
        }{
            $\prefixConcatRel{p}{p'}{p''}$
        }

        \caseText{By Theorem~\ref{thm:env-concat-assoc}:}

        \caseFact{6} $\prefixConcat{\eta_1'}{\left(\prefixConcat{\eta_1}{\eta_2}\right)} = \prefixConcat{\left(\prefixConcat{\eta_1'}{\eta_1}\right)}{\eta_2} = \prefixConcat{\eta_1''}{\eta_2}$

        \caseText{By (6), there is some $p_1'$ such that:}

        \caseFact{7} $\eta_2(z) = p_1'$

        \caseFact{8} $\left(\prefixConcat{\eta_1''}{\eta_2}\right) = \sumpA{\prefixConcat{p_1}{p_1'}}$

        \caseText{Using (7), we can rewrite (4) as:}

        \caseFact{9} $\prefixstep{\eta_2[x \mapsto p_1']}{e_1'}{n'}{e}{p'}$

        \caseText{By (8) and (6), inverting (5) can only lead to S-Plus-L-2. Doing so yields:}

        \caseFact{10} $\prefixstep{\left(\prefixConcat{\eta_1''}{\eta_2}\right)[x \mapsto \prefixConcat{p_1}{p_1'}]}{e_1}{n''}{e_1''}{p''}$

        \caseFact{11} $e' = e_1''[z/x]$

        \caseText{We can compute that:}

        \caseFact{12} $\left(\prefixConcat{\eta_1''}{\eta_2}\right)[x \mapsto \prefixConcat{p_1}{p_1'}] = \prefixConcat{\left(\eta_1''[x \mapsto p_1]\right)}{\left(\eta_2'[x \mapsto p_1']\right)}$

        \caseText{So by (8), we can write (10) as}

        \caseText{13} $\prefixstep{\prefixConcat{\left(\eta_1''[x \mapsto p_1]\right)}{\left(\eta_2'[x \mapsto p_1']\right)}}{e_1}{n''}{e_1''}{p''}$

        \caseText{Both goals follow by IH on (3), with (9) and (13)}
    }

    \jcase{20}{S-Plus-L-3}{Identical to previous.}

    \jcase{21}{S-Star-R-1}{Immediate.}

    \jcase{22}{S-Star-R-2-1}{Identical to the cases for S-Cat-R-1}

    \jcase{23}{S-Star-R-2-2}{Identical to the cases for S-Cat-R-2}

    \jcase{24}{S-Star-L-1}{Identical to the cases for S-Plus-L-1}

    \jcase{25}{S-Star-L-2}{Identical to S-Plus-L-2}

    \jcase{26}{S-Star-L-3}{Identical to S-Plus-L-2}

    \jcase{27}{S-Star-L-4}{Identical to S-Plus-L-2}

    \jcase{28}{S-Let}{
        \jgivengoalTwo{
            \caseFact{1} $\prefixstep{\eta_1}{e_1}{n_1}{e_1'}{p_1}$

            \caseFact{2} $\prefixstep{\eta_1[x\mapsto p_1]}{e_2}{n_2}{e_2'}{p_1'}$

            \caseFact{3} $\prefixstep{\eta_2}{e_1'}{n_1'}{e_1''}{p_2}$

            \caseFact{4} $\prefixstep{\eta_2[x\mapsto p_2]}{e_2'}{n_2'}{e_2''}{p_2'}$

            \caseFact{5} $\prefixstep{\prefixConcat{\eta_1}{\eta_2}}{e_1}{n_1''}{e_1'''}{p_3}$

            \caseFact{6} $\prefixstep{\left(\prefixConcat{\eta_1}{\eta_2}\right)[x\mapsto p_3]}{e_2}{n_2''}{e_2'''}{p_3'}$
        }{
            $\cutTm{x}{e_1''}{e_2''} = \cutTm{x}{e_1'''}{e_2'''}$
        }{
            $\prefixConcatRel{p_1'}{p_2'}{p_3'}$
        }

        \caseText{By IH on (1), with (3) and (4)}

        \caseFact{7} $e_1'' = e_1'''$

        \caseFact{8} $\prefixConcatRel{p_1}{p_2}{p_3}$

        \caseText{By (8):}

        \caseFact{9} $\left(\prefixConcat{\eta_1}{\eta_2}\right)[x\mapsto p_3] = \prefixConcat{\left(\eta_1[x\mapsto p_1]\right)}{\left(\eta_2[x \mapsto p_2]\right)}$

        \caseText{By (9), we can rewrite (6) to:}

        \caseFact{10} $\prefixstep{\prefixConcat{\left(\eta_1[x\mapsto p_1]\right)}{\left(\eta_2[x \mapsto p_2]\right)}}{e_2}{n_2''}{e_2'''}{p_3'}$

        \caseText{By IH on (2), with (4) and (7)}

        \caseFact{11} $e_2'' = e_2'''$

        \caseFact{12} $\prefixConcatRel{p_1'}{p_2'}{p_3'}$

        \caseText{\textbf{Goal A} is immediate by (7) and (11), and \textbf{Goal B} is (12).}
    }

    \jcase{29}{S-HistPgm}{Immediate by Theorem~\ref{thm:sink-step} and Theorem~\ref{thm:sink-concat}.}

    \jcase{30}{S-Wait-1, S-Wait-1}{
        \jgivengoalTwo{
            \caseFact{1} $\prefixConcatRel{\eta_1'}{\eta_1}{\eta_1''}$

            \caseFact{2} $\eta_1''(x) = p$

            \caseFact{3} $\neg \left(\isMaximal{p}\right)$

            \caseFact{4} $\prefixConcatRel{\eta_1''}{\eta_2}{\eta_2'}$

            \caseFact{5} $\eta_2'(x) = p'$

            \caseFact{6} $\neg \left(\isMaximal{p'}\right)$

            \caseFact{7} $\prefixstep{\prefixConcat{\eta_1}{\eta_2}}{\waitTm{\eta_1'}{t}{x}{e}}{n''}{e'}{\emp{t}}$
        }{
            $e' = \waitTm{\eta_2'}{t}{x}{e}$
        }{
            $\prefixConcatRel{\emp{t}}{\emp{t}}{\emp{t}}$
        }

        \caseText{Note that, by Theorem~\ref{thm:prefix-concat-assoc}}

        \caseFact{8} $\prefixConcat{\eta_1'}{\left(\prefixConcat{\eta_1}{\eta_2}\right)} = \prefixConcat{\left(\prefixConcat{\eta_1'}{\eta_1}\right)}{\eta_2} = \prefixConcat{\eta_1''}{\eta_2} = \eta_2'$

        \caseText{Because of (8), when we invert (7), only S-Wait-1 can apply. Inverting (7) completes \textbf{Goal A}, and \textbf{Goal B} follows by Theorem~\ref{thm:prefix-concat-emp}}
    }

    \jcase{31}{S-Wait-1, S-Wait-2}{
        \jgivengoalTwo{
            \caseFact{1} $\prefixConcatRel{\eta_1'}{\eta_1}{\eta_1''}$

            \caseFact{2} $\eta_1''(x) = p$

            \caseFact{3} $\neg \left(\isMaximal{p}\right)$

            \caseFact{4} $\prefixConcatRel{\eta_1''}{\eta_2}{\eta_2'}$

            \caseFact{5} $\eta_2'(x) = p'$

            \caseFact{6} $\isMaximal{p'}$

            \caseFact{7} $\prefixstep{\eta_2'}{e[\flatten{p'}/x]}{n}{e'}{p''}$

            \caseFact{8} $\prefixstep{\prefixConcat{\eta_1}{\eta_2}}{\waitTm{\eta_1'}{t}{x}{e}}{n'}{e''}{p'''}$
        }{
            $e'' = e'$
        }{
            $\prefixConcatRel{\emp{t}}{p''}{p'''}$
        }

        \caseText{By Theorem~\ref{thm:prefix-concat-assoc}}

        \caseFact{8} $\prefixConcat{\eta_1'}{\left(\prefixConcat{\eta_1}{\eta_2}\right)} = \prefixConcat{\left(\prefixConcat{\eta_1'}{\eta_1}\right)}{\eta_2} = \prefixConcat{\eta_1''}{\eta_2} = \eta_2'$

        \caseText{Because of (8), when we invert (7), only S-Wait-2 can apply. Inverting (7) and applying Theorem~\ref{thm:prefix-sem-det} completes \textbf{Goal A}. \textbf{Goal B} is complete by Theorem~\ref{thm:prefix-concat-emp}}

    }

    \jcase{32}{S-Wait-2}{
        \jgivengoalTwo{
            \caseFact{1} $\prefixConcatRel{\eta_1'}{\eta_1}{\eta_1''}$

            \caseFact{2} $\eta_1''(x) = p$

            \caseFact{3} $\isMaximal{p}$

            \caseFact{4} $\prefixstep{\eta_1''}{e[\flatten{p}/x]}{n}{e'}{p'}$

            \caseFact{5} $\prefixstep{\eta_2}{e'}{n'}{e''}{p''}$

            \caseFact{6} $\prefixstep{\prefixConcat{\eta_1}{\eta_2}}{\waitTm{\eta_1'}{t}{x}{e}}{n''}{e'''}{p'''}$
        }{
            $e'' = e'''$
        }{
            $\prefixConcatRel{p'}{p''}{p'''}$
        }

        \caseText{By Theorem~\ref{thm:prefix-concat-assoc} on (1):}

        \caseFact{7} $\prefixConcat{\eta_1'}{\left(\prefixConcat{\eta_1}{\eta_2}\right)} = \prefixConcat{\left(\prefixConcat{\eta_1'}{\eta_1}\right)}{\eta_2} = \prefixConcat{\eta_1''}{\eta_2}$

        \caseText{By definition using (7), there is some $p_0$ so that $\eta_2(x) = p_0$, and}

        \caseFact{8} $\left(\prefixConcat{\eta_1''}{\eta_2}\right)(x) = \prefixConcat{p}{p_0}$

        \caseText{By Theorem~\ref{thm:concat-maximal} on (3) and (8)}

        \caseFact{9} $\prefixConcat{p}{p_0} = p$

        \caseText{By (9) and (3), inverting (6) can only conclude by S-Wait-2. Doing so yields:}

        \caseFact{10} $\prefixstep{\prefixConcat{\eta_1''}{\eta_2}}{e[\flatten{p}/x]}{n''}{e'''}{p'''}$

        \caseText{\textbf{Goal A} and \textbf{Goal B} follow immediately by IH on (4), with (5), (10).}
    }

    \jcase{33}{S-Fix}{Immediate by IH and Definition~\ref{def:histpgm}}

    \jcase{34}{S-ArgsLet}{Immediate by two uses of IH.}

    \jcase{35}{S-Args-Emp}{Immediate.}

    \jcase{36}{S-Args-Sng}{Immediate by IH.}

    \jcase{37}{S-Args-Comma}{Immediate by two uses of IH.}

    \jcase{38}{S-Args-Semic-1-1, S-Args-Semic-1-1}{Immediate by IH.}

    \jcase{39}{S-Args-Semic-1-1, S-Args-Semic-1-2}{Like the \ruleName{S-Cat-R-1}, \ruleName{S-Cat-R-2} case.}

    \jcase{40}{S-Args-Semic-1-2, S-Args-Semic-2}{Immediate by IH.}

    \jcase{41}{S-Args-Semic-2}{Immediate by IH.}

}

\subsection{Determinism}
\label{app:determinism}

\begin{theorem}[Determinism Theorem]
Suppose:
  \begin{enumerate}
    \item[(1)] $\tck{\cdot}{\commactx{\Gamma}{\Gamma'}}{\emptyset}{e}{s}{i}$
    \item $\envHasType{\eta}{\commactx{\Gamma}{\Gamma'}}$
    \item $\prefixstep{\eta|_{\Gamma} \cup \emp{\Gamma'}}{e}{n_1}{e_1}{p_1}$ and $\prefixstep{\eta|_{\Gamma'} \cup \emp{\Gamma}}{e_1}{n_2}{e_2}{p_2}$
    \item $\prefixstep{\eta|_{\Gamma'} \cup \emp{\Gamma}}{e}{n_1'}{e_1'}{p_1'}$ and $\prefixstep{\eta|_{\Gamma} \cup \emp{\Gamma'}}{e_1'}{n_2'}{e_2'}{p_2'}$.
    \item $\prefixstep{\eta}{e}{ }{e'}{p}$
  \end{enumerate}
  Then $e' = e_2 = e_2'$ and $p = \prefixConcat{p_1}{p_2} = \prefixConcat{p_1'}{p_2'}$.
\end{theorem}
\begin{proof}
  By Theorem~\ref{thm:env-concat-emp}, $\eta|_{\Gamma} \cdot \emp{\Gamma} = \eta|_{\Gamma}$, and  $\eta|_{\Gamma'} \cdot \emp{\Gamma'} = \eta|_{\Gamma'}$.
  Then, we can compute the concatenation of the subsequent input environments in (2), and those in (3).
  \begin{align*}
    {\prefixConcat{\left(\eta|_{\Gamma} \cup \emp{\Gamma'}\right)}{\left(\eta|_{\Gamma'} \cup \emp{\Gamma}\right)}}
    &= \left(\eta_{\Gamma} \cdot \emp{\Gamma}\right) \cup \left(\eta|_{\Gamma'} \cdot \emp{\Gamma'}\right)\\
    &= \eta|_\Gamma \cup \eta_{\Gamma'}\\
    &= \eta
  \end{align*}
  By the same argument, 
  $${\eta|_{\Gamma'} \cup \emp{\Gamma}} \cdot  {\eta|_{\Gamma} \cup \emp{\Gamma'}} = \eta $$
  By two uses of Theorem~\ref{thm:hom-thm}, we have:
  $e' = e_2 = e_2'$, and $p = \prefixConcat{p_1}{p_2} = \prefixConcat{p_1'}{p_2'}$, as required.
\end{proof}


\newpage

\section{Events}
\label{app:eventsApp}

\emph{Events} allow us to represent a \core{} prefix as a sequence of totally
ordered items, while retaining information needed to infer the rich structure of
the prefix representations. In this section, we define serialization and
deserialization functions from sequences of events to prefixes and back.
We further prove that, for any type $s$, the size of the possible events that may occur
on a channel sending events of type $s$ (and its derivatives) is bounded.
This section serves to justify our claim from Section~\ref{sec:core-lang}
that \core{} can be run atop a traditional stream processing system where streams are sequences.

The grammar of events is:
$$
  x ::= \oneev \mid \parevA{x} \mid \parevB{x} \mid \sumPuncA \mid \sumPuncB \mid \catPunc \mid \catevA{x}
$$

\begin{definition}[Event Typing Relation]
    We define a binary relation $\eventfor{x}{s}$ as follows:

    \begin{mathpar}
      \inferrule{ }{\eventfor{\oneev}{1}}

      \inferrule{\eventfor{x}{s}}{\eventfor{\parevA{x}}{s \| t}}

      \inferrule{\eventfor{x}{t}}{\eventfor{\parevB{x}}{s \| t}}

        \ENDOFLINE

      \inferrule{ }{\eventfor{\sumPuncA}{s + t}}

      \inferrule{ }{\eventfor{\sumPuncB}{s + t}}

      \inferrule{\nullable{s}}{\eventfor{\catPunc}{s \cdot t}}

        \ENDOFLINE

      \inferrule{\eventfor{x}{s}}{\eventfor{\catevA{x}}{s \cdot t}}

      \inferrule{ }{\eventfor{\sumPuncA}{s^\star}}

      \inferrule{ }{\eventfor{\sumPuncB}{s^\star}}

    \end{mathpar}
  \end{definition}

  Note that $s + t$ and $s^\star$ share the same punctuation events. Intuitively, this is because
  $s^\star$ can be unrolled as $\varepsilon + (s \cdot s^\star)$.

  \begin{definition}[Event Derivative Relation]
    We define a ternary relation $\derivrel{x}{s}{s'}$.

    \begin{mathpar}

      \inferrule{ }{
      \derivrel{\oneev}{1}{\varepsilon}
    }

    \inferrule{
      \derivrel{x}{s}{s'}
    }{
      \derivrel{\parevA{x}}{s \| t}{s' \| t}
    }

    \inferrule{
      \derivrel{x}{t}{t'}
    }{
      \derivrel{\parevB{x}}{s \| t}{s \| t'}
    }

    \inferrule{
    }{
      \derivrel{\sumPuncA}{s + t}{s}
    }

        \ENDOFLINE

    \inferrule{ }{
      \derivrel{\sumPuncB}{s + t}{t}
    }

    \inferrule{
      \nullable{s}
    }{
      \derivrel{\catPunc}{s \cdot t}{t}
    }

    \inferrule{
      \derivrel{x}{s}{s'}
    }{
      \derivrel{\catevA{x}}{s \cdot t}{s' \cdot t}
    }

    \inferrule{ }{
        \derivrel{\sumPuncA}{s^\star}{\varepsilon}
    }

        \ENDOFLINE

    \inferrule{ }{
        \derivrel{\sumPuncB}{s^\star}{s \cdot s^\star}
    }

    \end{mathpar}
  \end{definition}

  \begin{theorem}[Event Derivative Function]
    \label{thm:event-deriv-function}
    If $\eventfor{x}{s}$, there is a unique $s'$ such that $\derivrel{x}{s}{s'}$.
  \end{theorem}

  Because of Theorem~\ref{thm:event-deriv-function}, if we know $\eventfor{x}{s}$,
  we may write the unique $s'$ such that $\derivrel{x}{s}{s'}$ as
  $\deriv{x}{s}$.

  \begin{definition}[Events Typing and Derivatives Relations]
    We lift event typing to lists by derivatives.
    \begin{mathpar}
      \inferrule{ }{
        \eventsfor{[]}{s}
      }

      \inferrule{
        \eventfor{x}{s}\\
        \derivrel{x}{s}{s'}\\
        \eventsfor{xs}{s'}
      }{
        \eventsfor{\cons{x}{xs}}{s}
      }
    \end{mathpar}

    We also lift derivatives to lists of events in the natural way.
    \begin{mathpar}
      \inferrule{ }{
        \derivrel{[]}{s}{s}
      }

      \inferrule{
        \derivrel{x}{s}{s'}\\
        \derivrel{xs}{s'}{s''}
      }{
        \derivrel{\cons{x}{xs}}{s}{s''}
      }
    \end{mathpar}
  \end{definition}

  \begin{theorem}[Events Derivative Function]
    \label{thm:events-deriv-function}
    If $\eventsfor{xs}{s}$, there is a unique $s'$ such that $\derivrel{xs}{s}{s'}$.
  \end{theorem}

  Because of Theorem~\ref{thm:events-deriv-function}, if we know $\eventsfor{xs}{s}$,
  we may write the unique $s'$ such that $\derivrel{xs}{s}{s'}$ as
  $\deriv{xs}{s}$.

  \begin{theorem}[Empty List of Events]
    \label{thm:empty-eventsfor}
    For all $s$, we have $\eventsfor{[]}{s}$, and $\derivrel{[]}{s}{s}$
  \end{theorem}

  \begin{theorem}[Events Concatenation]
    \label{thm:eventsfor-concat-deriv}
    If
    \begin{enumerate}
      \item $\eventsfor{xs}{s}$
      \item $\derivrel{xs}{s}{s'}$
      \item $\eventsfor{ys}{s'}$
      \item $\derivrel{ys}{s'}{s''}$
    \end{enumerate}
    Then, $\eventsfor{xs \texttt{++} ys}{s}$, and $\derivrel{xs \texttt{++} ys}{s}{s''}$.

    In other words, if $\eventsfor{xs}{s}$ and $\eventsfor{ys}{\deriv{xs}{s}}$, then $\eventsfor{xs \texttt{++} ys}{s}$
    and $\deriv{xs\texttt{++}ys}{s} = \deriv{ys}{\deriv{xs}{s}}$.
  \end{theorem}
  \begin{proof}
    Induction on $xs$.
  \end{proof}

\subsection*{Events}

\subsection*{Events to Prefix}

\begin{definition}[Event(s) to Prefix]
 \begin{mathpar}
   \inferrule{ }{ \EToP{\oneev}{1}{\onepB} }

   \inferrule{ \EToP{x}{s}{p} }{ \EToP{\parevA{x}}{s \| t}{\parp{p}{\emp{t}}}}

        \ENDOFLINE

   \inferrule{ \EToP{x}{t}{p} }{ \EToP{\parevB{x}}{s \| t}{\parp{\emp{s}}{p} }}

   \inferrule{ \EToP{x}{s}{p}}{ \EToP{\catevA{x}}{s \cdot t}{\catpA{p}}}

        \ENDOFLINE

   \inferrule{ \nullable{s} \\ \donebatch{s}{b}}{ \EToP{\catPunc}{s \cdot t}{\catpB{b}{\emp{t}}}}

   \inferrule{ }{ \EToP{\sumPuncA}{s + t}{\sumpA{\emp{s}}}}

   \inferrule{ }{ \EToP{\sumPuncB}{s + t}{\sumpB{\emp{t}}}}

        \ENDOFLINE

   \inferrule{ }{ \EToP{\sumPuncA}{s^\star}{\stpDone}}

   \inferrule{ }{ \EToP{\sumPuncB}{s^\star}{\stpA{\emp{s}}}}

        \ENDOFLINE

   \inferrule{ }{ \ESToP{[]}{s}{\emp{s}} }

   \inferrule{
     \EToP{t}{s}{p}\\
     \delta_p s \sim s'\\
     \ESToP{ts}{s'}{p'}\\
     p' \cdot p \sim p''
   }{
     \ESToP{t::ts}{s}{p''}
   }
  \end{mathpar}

\end{definition}

\begin{theorem}[Event to Prefix Function]
  \label{thm:etop-function}
  If $\eventfor{x}{s}$ then there is a unique $\prefixHasType{p}{s}$ such that $\EToP{x}{s}{p}$.
\end{theorem}
\begin{proof}
    Induction on $\eventfor{x}{s}$.
\end{proof}

\begin{theorem}[Events to Prefix Function]
  \label{thm:estop-function}
  If $\eventsfor{xs}{s}$, then there is a unique $\prefixHasType{p}{s}$
  such that $\ESToP{xs}{s}{p}$.
\end{theorem}
\begin{proof}
    Induction on $xs$, using Theorem~\ref{thm:etop-function}.
\end{proof}

\subsection*{Prefix to Events}

\begin{definition}[PToES]

  In this definition, we write occasionally lift event constructors to lists, writing $\widehat{f}(xs)$ for $\left[f(x) \mid x \in xs\right]$.
  Also, we write $xs \| ys$ for the set of \emph{shuffles} of the lists $xs$ and $ys$.
  \begin{mathpar}

    \inferrule[PToES-Eps]{ }{ \PToES{\epsp}{\varepsilon}{ [] } }

    \inferrule[PToES-One-A]{ }{ \PToES{\onepA}{1}{ [] } }

    \inferrule[PToES-One-B]{ }{ \PToES{\onepB}{1}{ [\oneev] } }

        \ENDOFLINE

    \inferrule[PToES-Par]{
      \PToES{p}{s}{xs}\\
      \PToES{p'}{t}{ys}\\
      zs \in \parevAmap{xs} \| \parevBmap{ys}
    }{
      \PToES{\parp{p}{p'}}{s \| t}{zs}
    }

    \inferrule[PToES-Cat-A]{
      \PToES{p}{s}{xs}
    }{
      \PToES{\catpA{p}}{s \cdot t}{\catevAmap{xs}}
    }

        \ENDOFLINE

    \inferrule[PToES-Cat-B]{
      \PToES{b^\circ}{s}{xs}\\
      \PToES{p}{t}{ys}\\
    }{
      \PToES{\catpB{b}{p}}{s \cdot t}{\catevAmap{xs} \texttt{++} \cons{\catPunc}{ys} }
    }

    \inferrule[PToES-Sum-Emp]{ }{
      \PToES{\sumpEmp}{s+t}{[]}
    }

    \inferrule[PToES-Sum-A]{
      \PToES{p}{s}{xs}
    }{
      \PToES{\sumpA{p}}{s + t}{\cons{\sumPuncA}{xs}}
    }

        \ENDOFLINE

    \inferrule[PToES-Sum-B]{
      \PToES{p}{t}{xs}
    }{
      \PToES{\sumpB{p}}{s + t}{\cons{\sumPuncB}{xs}}
    }

    \inferrule[PToES-Star-Emp]{ }{
      \PToES{\stpEmp}{s^\star}{[]}
    }

    \inferrule[PToES-Star-Done]{ }{
      \PToES{\stpDone}{s^\star}{\sumPuncA}
    }

        \ENDOFLINE

    \inferrule[PToES-Star-A]{
      \PToES{p}{s}{xs}\\
    }{
      \PToES{\stpA{p}}{s^\star}{\cons{\sumPuncB}{\catevAmap{xs}}}
    }

        \ENDOFLINE

    \inferrule[PToES-Star-B]{
      \PToES{b^\circ}{s}{xs}\\
      \PToES{p}{s^\star}{ys}
    }{
      \PToES{\stpB{b}{p}}{s^\star}{ \cons{\sumPuncB}{ {\catevAmap{xs}} \texttt{++} {\cons{\catPunc}{ys}}} }
    } 


  \end{mathpar}
\end{definition}

\begin{theorem}[PToES Empty]
  \label{thm:ptoes-empty}
  If $\PToES{p}{s}{xs}$, then $xs = []$ iff $p = \emp{s}$
\end{theorem}
\begin{proof}
    Induction on $\PToES{p}{s}{xs}$.
\end{proof}

\begin{theorem}[PToES Left Total]
  \label{thm:ptoes-function}
  If $\prefixHasType{p}{s}$ then there exists (not necessarily unique) $xs$ such that $\PToES{p}{s}{xs}$
\end{theorem}
\begin{proof}
    Induction on $\prefixHasType{p}{s}$.
\end{proof}

\begin{lemma}[PToES Relation Derivative Agreement]
    If
    \begin{enumerate}
        \item $\prefixHasType{p}{s}$
        \item $\PToES{p}{s}{xs}$
        \item $\derivrel{p}{s}{s'}$
    \end{enumerate}
    then $\eventsfor{xs}{s}$ and $\derivrel{xs}{s}{s'}$
\end{lemma}
\begin{proof}
    Induction on $\PToES{p}{s}{xs}$.
\end{proof}


\subsection*{Event Size}
Each event carries tag information about where it appears within a structured stream; this is
necessary for us to recover the rich prefix structure. Importantly, for a given stream type there
is an upper bound on the amount of tag information to be included on any event in any stream of that type.

\begin{definition}[Event size]
    We define the \emph{size} of an event recursively:
    \begin{align*}
        \evSize{\oneev} & = 1\\
        \evSize{\parevA{x}} & = 1 + \evSize{x}\\
        \evSize{\parevB{x}} & = 1 + \evSize{x}\\
        \evSize{\sumPuncA} & = 1 \\
        \evSize{\sumPuncB} & = 1 \\
        \evSize{\catPunc} & = 1\\
        \evSize{\catevA(x)} & = 1 + \evSize{x}
    \end{align*}
\end{definition}

We lift this to lists of events in the natural way:
\begin{definition}[Event List Size]
    \begin{align*}
        \evSize{[]} & = 0 \\
        \evSize{\cons{x}{xs}} & = \evSize{x} + \evSize{xs}
    \end{align*}
\end{definition}

To construct an \emph{a priori} bound on the size of any event to appear in stream, we recurse on the type of the stream:
\begin{definition}[Event size bound]
    \begin{align*}
        \evSizeBound{\varepsilon} & = 0\\
        \evSizeBound{1} & = 1\\
        \evSizeBound{s \| t} & = 1 + \max{(\evSizeBound{s}, \evSizeBound{t})}\\
        \evSizeBound{s \cdot t}& = \max{(1 + \evSizeBound{s}, \evSizeBound{t})} \\
        \evSizeBound{s + t} & = \max{(1, \evSizeBound{s}, \evSizeBound{t})} \\
        \evSizeBound{s^\star} & = \max{(1, 1 + \evSizeBound{s})}
    \end{align*}
\end{definition}

\begin{theorem}[Bounded Event Size]
    For all $s$, there is some $N = \evSizeBound(s)$ such that for any $\eventsfor{\XS}{s}$
    and any $x \in \XS$, we have that $|x| \leq N$, where $| \cdot |$
    denotes the size of the AST.
  \end{theorem}
  \begin{proof}
    Induction on 
  \end{proof}

\subsection*{Serialization and Deserialization}

We turn now to the final result of \ref{app:eventsApp}, that we can serialize a prefix $p$ into
a list of events $xs$, secure in the knowledge that when we deserialize $xs$ we will obtain the same prefix $p$.

Towards this result, we introduce a series of lemmas that allow us to use the tag information encoded in each event to
recover the prefix structure during deserialization. Observe that the shape of each lemma mirrors that of the
corresponding serialization ($\dagger$) constructor.

\begin{lemma}[EsToP Par Recovery]
    If \begin{enumerate}
        \item $zs \in \parevAmap{xs} \| \parevBmap{ys}$ (where $\|$ is list shuffle)
        \item $\ESToP{xs}{s}{p}$
        \item $\ESToP{ys}{t}{p'}$
    \end{enumerate}
    then $\ESToP{zs}{s \| t}{\parp{p}{p'}}$
\end{lemma}
\begin{proof}
    Induction on $zs$.
\end{proof}

\begin{lemma}[EsToP Cat Recovery]
    If $\ESToP{xs}{s}{\prefixOf{b}}$ and $\ESToP{ys}{t}{p}$ then
    $$\ESToP{\catevAmap{xs} {++} \cons{\catPunc}{ys}}{s \cdot t}{\catpB{b}{p}}$$
\end{lemma}
\begin{proof}
    Induction on $xs$.
\end{proof}

\begin{lemma}[EsToP Star Recovery]
    If $\ESToP{xs}{s}{\prefixOf{b}}$ and $\ESToP{ys}{s^\star}{p}$ then
    $$\ESToP{\cons{\sumPuncB}{\catevAmap{xs}} {++} \cons{\catPunc}{ys}}{s^\star}{\stpB{b}{p}}$$
\end{lemma}
\begin{proof}
    Induction on $xs$.
\end{proof}

\begin{theorem}[Serialization/Deserialization Round Trip]
    \label{thm:estop-deterministic}
    If $\PToES{p}{s}{\XS}$, then $\ESToP{\XS}{s}{p}$.
  \end{theorem}
  \begin{proof}
    Induction on $\PToES{p}{s}{\XS}$.
  \end{proof}


\end{document}